\documentclass[12pt, onecolumn, letterpaper]{article}

\usepackage[utf8]{inputenc} 
\usepackage[T1]{fontenc}
\usepackage{url}
\usepackage{ifthen}
\usepackage{cite}
\usepackage[cmex10]{amsmath} 
\usepackage{amssymb}
\usepackage{algorithm}
\usepackage{algorithmic}
\usepackage{subfigure}
\usepackage{tikz}
\usepackage{algorithm}
\usepackage{algorithmic}
\usepackage{amsmath}
\usepackage{amsthm}
\usepackage{amssymb}
\usepackage{ecltree}
\usepackage{enumerate}
\usepackage{mathrsfs}
\usepackage{xcolor}
\usepackage{comment}
\usepackage{tikz}
\usepackage{wrapfig}
\usepackage{mathtools}
\usetikzlibrary{trees}

\usepackage[normalem]{ulem}
\usepackage{here}
\usepackage{cases}
\usepackage{stackengine}
\usepackage{longtable}
\usepackage{diagbox}

\newtheorem{definition}{Definition}
\newtheorem{lemma}{Lemma}
\newtheorem{theorem}{Theorem}

\newtheorem{corollary}{Corollary}
\newtheorem{remark}{Remark}
\newtheorem{example}{Example}

\def\trans{\mathrm{\tau}}

\def\PREF{\mathcal{P}}

\def\hdec{\text{-}\mathrm{dec}}
\def\ext{\mathrm{ext}}
\def\irr{\mathrm{irr}}
\def\reg{\mathrm{reg}}

\def\hopt{\text{-}\mathrm{opt}}

\def\Huff{\mathrm{Huff}}

\def\pref{\mathrm{pref}}
\def\suff{\mathrm{suff}}
\def\kernel{\mathcal{R}}

\def \prefset{\mathscr{P}}

\newcommand{\argmin}{\mathop{\rm arg~min}\limits}

\begin{document}
\title{Properties of $k$-bit Delay Decodable Codes} 
\author{Kengo Hashimoto, Ken-ichi Iwata}

\date{University of Fukui, \\
             E-mail: \{khasimot, k-iwata\}@u-fukui.ac.jp}

\maketitle
\begin{abstract}
The class of $k$-bit delay decodable codes, source codes allowing decoding delay of at most $k$ bits for $k \geq 0$, can attain a shorter average codeword length than Huffman codes.
This paper discusses the general properties of the class of $k$-bit delay decodable codes with a finite number of code tables and  proves two theorems which enable us to limit the scope of code-tuples to be considered when discussing optimal $k$-bit delay decodable code-tuples.
\end{abstract}

\section{Introduction}
\label{sec:introduction}

It is known that one can achieve a shorter average codeword length than Huffman codes by allowing multiple code tables and some decoding delay.
AIFV (almost instantaneous fixed-to-variable length) codes developed by Yamamoto, Tsuchihashi, and Honda \cite{Yamamoto2015} attain a shorter average codeword length than Huffman codes by using a time-variant encoder with two code tables and allowing decoding delay of at most two bits.
AIFV codes are generalized to AIFV-$m$ codes, which can achieve a shorter average codeword length than AIFV codes for $m \geq 3$, allowing $m$ code tables and decoding delay of at most $m$ bits \cite{Hu2017}.
The worst-case redundancy of AIFV-$m$ codes is analyzed in \cite{Hu2017, Fujita2020} for $m = 2,3,4,5$. 
The literature \cite{IY:ISITA16, IY:ITW17, Fujita2019, Fujita2018, ISIT2018, ISITA2018, Golin2019, Golin2020, ISIT2020, Golin2021, Golin2022, Sumigawa2017, Hashimoto2019, ITW2020} propose the code construction and coding method of AIFV and AIFV-$m$ codes.
Extensions of AIFV-$m$ codes are proposed in \cite{Sugiura2018, Sugiura2022}.

The literature \cite{Hashimoto2022} formalized a binary encoder with a finite number of code tables as a \emph{code-tuple} and
introduced the class of code-tuples decodable with a delay of at most $k$ bits as the class of \emph{$k$-bit delay decodable codes}, which includes the class of AIFV-$k$ codes as a proper subclass.
Also, \cite{Hashimoto2022} proved that Huffman codes achieve the optimal average codeword length in the class of $1$-bit delay decodable code-tuples.
The literature \cite{Hashimoto2021} indicates that the class of AIFV codes achieves the optimal average codeword length in the class of $2$-bit delay decodable code-tuples with two code tables.

This paper discusses the general properties of $k$-bit delay decodable code-tuples for $k \geq 0$ and 
proves two theorems as the main results.
The first theorem guarantees that it is not the case that one can achieve an arbitrarily small average codeword length by using arbitrarily many code tables.
This leads to the existence of an \emph{optimal} $k$-bit delay decodable code-tuple, which achieves an average codeword length shorter than or equal to any other $k$-bit delay code-tuple.
The first theorem also gives an upper bound of the required number of code tables for an optimal $k$-bit delay decodable code-tuples.
The second theorem gives a necessary condition for a $k$-bit delay decodable code-tuple to be optimal, which is a generalization of a property of Huffman codes that each internal node in a code tree has two child nodes.
Both theorems enable us to limit the scope of code-tuples to be considered when discussing optimal $k$-bit delay decodable code-tuples.
As an application of these theorems, we can prove of the optimality of AIFV codes in the class of $2$-bit delay decodable codes with a finite number of code tables \cite{Hashimoto2023}.
 
This paper is organized as follows.
In Section \ref{sec:preliminary}, we prepare some notations, describe our data compression scheme, introduce some notions including $k$-bit delay decodable codes, and show their basic properties used to prove our main result.
Then we prove two theorems in Section \ref{sec:main} as the main results of this paper.
Lastly, we conclude this paper in Section \ref{sec:conclusion}.
To clarify the flow of the discussion, we relegate the proofs of most of the lemmas to \ref{sec:proofs}.
The main notations are listed in \ref{sec:notation}.

\begin{table}
\caption{Two examples of an code-tuple: $F^{(\alpha)}(f^{(\alpha)}_0, f^{(\alpha)}_1, f^{(\alpha)}_2, \allowbreak \trans^{(\alpha)}_0, \trans^{(\alpha)}_1, \trans^{(\alpha)}_2)$ and $F^{(\beta)}(f^{(\beta)}_0, f^{(\beta)}_1, f^{(\beta)}_2, \trans^{(\beta)}_0, \trans^{(\beta)}_1, \trans^{(\beta)}_2)$}
\label{tab:code-tuple}
\centering
\begin{tabular}{c | lclclc}
\hline
$s \in \mathcal{S}$ & $f^{(\alpha)}_0$ & $\trans^{(\alpha)}_0$ & $f^{(\alpha)}_1$ & $\trans^{(\alpha)}_1$ & $f^{(\alpha)}_2$ & $\trans^{(\alpha)}_2$\\
\hline
a & 01 & 0 & 00 & 1 & 1100 & 1\\
b & 10 & 1 & $\lambda$ & 0 & 1110 & 2\\
c & 0100 & 0 & 00111 & 1 & 111000 & 2\\
d & 01 & 2 & 00111 & 2 & 110 & 2\\
\hline
\end{tabular}

\vspace{8pt}

\begin{tabular}{c | lclclc}
\hline
$s \in \mathcal{S}$ & $f^{(\beta)}_0$ & $\trans^{(\beta)}_0$ & $f^{(\beta)}_1$ & $\trans^{(\beta)}_1$ & $f^{(\beta)}_2$ & $\trans^{(\beta)}_2$\\
\hline
a & $\lambda$ & 1 & 0110 & 1 & $\lambda$ & 2\\
b & 101 & 2 & 01 & 1 & $\lambda$ & 2\\
c & 1011 & 1 & 0111 & 1 & $\lambda$ & 2\\
d & 1101 & 2 & 01111 & 1 & $\lambda$ & 2\\
\hline
\end{tabular}
\end{table}

\section{Preliminaries}
\label{sec:preliminary}

First, we define some notations as follows.
Most of the notations in this paper are based on \cite{Hashimoto2022}.
Let $\mathbb{R}$ denote the set of all real numbers, and let $\mathbb{R}^m$ denote the set of all $m$ dimensional real row vectors for an integer $m \geq 1$.
Let $|\mathcal{A}|$ denote the cardinality of a finite set $\mathcal{A}$.
Let $\mathcal{A} \times \mathcal{B}$ denote the Cartesian product of $\mathcal{A}$ and $\mathcal{B}$, that is, $\mathcal{A} \times \mathcal{B} \coloneqq \{(a, b) : a \in \mathcal{A}, b \in \mathcal{B}\}$.
Let $\mathcal{A}^k$ (resp. $\mathcal{A}^{\leq k}$, $\mathcal{A}^{\geq k}$, $\mathcal{A}^{\ast}$, $\mathcal{A}^{+}$) denote the set of all sequences of length $k$ (resp. of length less than or equal to $k$, of length greater than or equal to $k$, of finite length, of finite positive length) over a set $\mathcal{A}$.
Thus, $\mathcal{A}^{+} = \mathcal{A}^{\ast} \setminus \{\lambda\}$, where $\lambda$ denotes the empty sequence.
The length of a sequence $\pmb{x}$ is denoted by $|\pmb{x}|$, in particular, $|\lambda| = 0$.
For a non-empty sequence $ \pmb{x} = x_1x_2\ldots x_{n}$, 
we define $\pref(\pmb{x}) \coloneqq x_1x_2\ldots x_{n-1}$ and $\suff(\pmb{x}) \coloneqq x_2\ldots x_{n-1}x_{n}$.
Namely, $\pref(\pmb{x})$ (resp. $\suff(\pmb{x})$) is the sequence obtained by deleting the last (resp. first) letter from $\pmb{x}$.
We say $\pmb{x} \preceq \pmb{y}$ if $\pmb{x}$ is a prefix of $\pmb{y}$, that is, there exists a sequence $\pmb{z}$, possibly $\pmb{z} = \lambda$, such that $\pmb{y} = \pmb{x}\pmb{z}$.
Also, we say $\pmb{x} \prec \pmb{y}$ if $\pmb{x} \preceq \pmb{y}$ and $\pmb{x} \neq \pmb{y}$.
For sequences $\pmb{x}$ and $\pmb{y}$ such that $\pmb{x} \preceq \pmb{y}$,
let $\pmb{x}^{-1}\pmb{y}$ denote the unique sequence $\pmb{z}$ such that $\pmb{x}\pmb{z} = \pmb{y}$.
Note that a notation $\pmb{x}^{-1}$ behaves like the ``inverse element'' of $\pmb{x}$ as stated in the following statements (i)--(iii).
\begin{itemize}
\item[(i)] For any $\pmb{x}$, we have $\pmb{x}^{-1}\pmb{x} = \lambda$.
\item[(ii)] For any $\pmb{x}$ and $\pmb{y}$ such that $\pmb{x} \preceq \pmb{y}$, we have $\pmb{x}\pmb{x}^{-1}\pmb{y} =\pmb{y}$.
\item[(iii)] For any $\pmb{x}, \pmb{y}$, and $\pmb{z}$ such that $\pmb{x}\pmb{y} \preceq \pmb{z}$, we have $(\pmb{x}\pmb{y})^{-1}\pmb{z} = \pmb{y}^{-1}\pmb{x}^{-1}\pmb{z}$.
\end{itemize}
The main notations used in this paper are listed in \ref{sec:notation}.

We now describe the details of our data compression system. 
In this paper, we consider a data compression system consisting of a source, an encoder, and a decoder.

\begin{itemize}
\item Source: We consider an i.i.d.\ source, which outputs a sequence $\pmb{x} = x_1x_2\ldots x_n$ of symbols of the source alphabet $\mathcal{S} = \{s_1, s_2, \ldots, s_{\sigma}\}$, where $n$ and $\sigma$ denote the length of $\pmb{x}$ and the alphabet size, respectively.
Each source output follows a fixed probability distribution $(\mu(s_1), \mu(s_2), \ldots, \mu(s_{\sigma}))$, where $\mu(s_i)$ is the probability of occurrence of $s_i$ for $i = 1, 2, \ldots, \sigma$.
In this paper, we assume $\sigma \geq 2$.

\item Encoder: The encoder has $m$ fixed code tables $f_0, f_1, \ldots, f_{m-1} : \mathcal{S} \rightarrow \mathcal{C}^{\ast}$, where $\mathcal{C} \coloneqq \{0, 1\}$ is the coding alphabet.
The encoder reads the source sequence $\pmb{x} \in \mathcal{S}^{\ast}$ symbol by symbol from the beginning of $\pmb{x}$ and encodes them according to the code tables.
For the first symbol $x_1$, we use an arbitrarily chosen code table from $f_0, f_1, \ldots, f_{m-1}$.
For $x_2, x_3, \ldots, x_n$, we determine which code table to use to encode according to $m$ fixed mappings $\trans_0, \trans_1, \ldots, \trans_{m-1} : \mathcal{S} \rightarrow [m] \coloneqq \{0, 1, 2, \ldots, m-1\}$.
More specifically, if the previous symbol $x_{i-1}$ is encoded by the code table $f_j$, then the current symbol $x_i$ is encoded by the code table $f_{\trans_j(x_{i-1})}$.
Hence, if we use the code table $f_i$ to encode $x_1$, then a source sequence $\pmb{x} = x_1x_2\ldots x_n$ is encoded to a codeword sequence $f(\pmb{x}) \coloneqq f_{i_1}(x_1)f_{i_2}(x_n)\ldots f_{i_n}(x_n)$, where
\begin{equation}
i_j \coloneqq
\begin{cases}
i &\,\,\text{if}\,\, j = 1,\\
\trans_{i_{j-1}}(x_{j-1})  &\,\,\text{if}\,\, j \geq 2
\end{cases}
\end{equation}
for $j = 1, 2, \ldots, n$.

\item Decoder: 
The decoder reads the codeword sequence $f(\pmb{x})$ bit by bit from the beginning of $f(\pmb{x})$.
Each time the decoder reads a bit, the decoder recovers as long prefix of $\pmb{x}$ as the decoder can uniquely identify from the prefix of $f(\pmb{x})$ already read.
We assume that the encoder and decoder share the index of the code table used to encode $x_1$ in advance.
\end{itemize}

\subsection{Code-tuples}
\label{subsec:treepair}

The behavior of the encoder and decoder for a given source sequence is completely determined by $m$ code tables $f_0, f_1, \ldots, f_{m-1}$ and $m$ mappings $\trans_0, \trans_1, \ldots, \trans_{m-1}$ if we fix the index of code table used to encode $x_1$.
Accordingly, we name a tuple $F(f_0, f_1, \ldots, f_{m-1}, \trans_0, \trans_1, \ldots, \trans_{m-1})$ as a \emph{code-tuple} $F$ and identify a source code with a code-tuple $F$.

\begin{definition}
  \label{def:treepair}
Let $m$ be a positive integer.
An \emph{$m$-code-tuple} $F(f_0, f_1, \ldots, f_{m-1}, \allowbreak \trans_0, \trans_1, \ldots, \trans_{m-1})$ is a tuple of
$m$ mappings $f_0, f_1, \ldots, f_{m-1} : \mathcal{S} \rightarrow \mathcal{C}^{\ast}$ and $m$ mappings $\trans_0, \trans_1, \ldots, \trans_{m-1} : \mathcal{S} \rightarrow [m]$.

We define $\mathscr{F}^{(m)}$ as the set of all $m$-code-tuples.
Also, we define
$\mathscr{F} \coloneqq  \mathscr{F}^{(1)} \cup \mathscr{F}^{(2)} \cup \mathscr{F}^{(3)} \cup \cdots.$
An element of $\mathscr{F}$ is called a \emph{code-tuple}.
\end{definition}

We write $F(f_0, f_1, \ldots, f_{m-1}, \trans_0, \trans_1, \ldots, \trans_{m-1})$ also as $F(f, \trans)$ or $F$ for simplicity.
For $F \in \mathscr{F}^{(m)}$, let $|F|$ denote the number of code tables of $F$, that is, $|F| \coloneqq m$.
We write $[|F|] = \{0, 1, 2, \ldots, |F|-1\}$ as $[F]$ for simplicity.

\begin{example}
Table \ref{tab:code-tuple} shows two examples $F^{(\alpha)}$ and $F^{(\beta)}$ of a $3$-code-tuple 
for $\mathcal{S} = \{\mathrm{a}, \mathrm{b}, \mathrm{c}, \mathrm{d}\}$.
\end{example}

\begin{example}
\label{ex:encode}
We consider encoding of a source sequence $\pmb{x} = x_1x_2x_3x_4 \coloneqq \mathrm{badb}$ with the code-tuple $F(f, \trans) \coloneqq F^{(\alpha)}(f^{(\alpha)}, \trans^{(\alpha)})$ in Table \ref{tab:code-tuple}.
If $x_1 = \mathrm{b}$ is encoded with the code table $f_0$, then the encoding process is as follows.
\begin{itemize}
\item $x_1 = \mathrm{b}$ is encoded to $f_0(\mathrm{b}) = 10$. The index of the next code table is $\trans_0(\mathrm{b}) = 1$.
\item $x_2 = \mathrm{a}$ is encoded to $f_1(\mathrm{a}) = 00$. The index of the next code table is $\trans_1(\mathrm{a}) = 1$.
\item $x_3 = \mathrm{d}$ is encoded to $f_1(\mathrm{d}) = 00111$. The index of the next code table is $\trans_1(\mathrm{d}) = 2$.
\item $x_4 = \mathrm{b}$ is encoded to $f_2(\mathrm{b}) = 1110$. The index of the next code table is $\trans_2(\mathrm{b}) = 2$.
\end{itemize}
As the result, we obtain a codeword sequence $\pmb{c} \coloneqq f_0(\mathrm{b})f_1(\mathrm{a})f_1(\mathrm{d})f_2(\mathrm{b}) = 1000001111110$.

The decoding process of $\pmb{c} = 1000001111110$ is as follows.
\begin{itemize}
\item After reading the prefix $10$ of $\pmb{c}$, the decoder can uniquely identify $x_1 = \mathrm{b}$ and $10 = f_0(\mathrm{b})$. The decoder can also know that $x_2$ should be decoded with $f_{\trans_0(\mathrm{b})} = f_1$.
\item After reading the prefix $1000 = f_0(\mathrm{c})f_0(\mathrm{a})$ of $\pmb{c}$, the decoder still cannot uniquely identify  $x_2 = \mathrm{a}$ because there remain three possible cases: the case $x_2 = \mathrm{a}$, the case $x_2 = \mathrm{c}$, and the case $x_2 = \mathrm{d}$.
\item After reading the prefix $10000$ of $\pmb{c}$, the decoder can uniquely identify $x_2 = \mathrm{a}$ and $10000 = f_0(\mathrm{b})f_1(\mathrm{a})0$. The decoder can also know that $x_3$ should be decoded with $f_{\trans_1(\mathrm{a})} = f_1$.
\item After reading the prefix $100000111 = f_0(\mathrm{b})f_1(\mathrm{a})\allowbreak f_1(\mathrm{d})$ of $\pmb{c}$, the decoder still cannot uniquely identify $x_3 = \mathrm{d}$ because there remain two possible cases: the case $x_3 = \mathrm{c}$ and the case $x_3 = \mathrm{d}$.
\item After reading the prefix $10000011111$ of $\pmb{c}$, the decoder can uniquely identify $x_3 = \mathrm{d}$ and $10000011111 = f_0(\mathrm{b})f_1(\mathrm{a})f_1(\mathrm{d})11$.
The decoder can also know that $x_4$ should be decoded with $f_{\trans_1(\mathrm{d})} = f_2$.
\item After reading the prefix $\pmb{c} = 1000001111110$, the decoder can uniquely identify $x_4 = \mathrm{b}$ and $1000001111110 = f_0(\mathrm{b})f_1(\mathrm{a})f_1(\mathrm{d})f_2(\mathrm{b})$.
\end{itemize}
As the result, the decoder recovers the original sequence $\pmb{x} = \mathrm{badb}$.
\end{example}

In encoding $\pmb{x} = x_1x_2 \ldots x_{n} \in \mathcal{S}^{\ast}$ with $F(f, \trans) \in \mathscr{F}$, 
the $m$ mappings $\trans_0, \trans_1, \ldots, \trans_{m-1}$ determine which code table to use to encode $x_2, x_3, \ldots, x_n$.
However, there are choices of which code table to use for the first symbol $x_1$.
For $i \in [F]$ and $\pmb{x} \in \mathcal{S}^{\ast}$, we define $f^{\ast}_i(\pmb{x}) \in \mathcal{C}^{\ast}$ as the codeword sequence in the case where $x_1$ is encoded with $f_i$.
Also, we define $\trans^{\ast}_i(\pmb{x}) \in [F]$ as the index of the code table used next after encoding $\pmb{x}$ in the case where $x_1$ is encoded with $f_i$.
We give formal definitions of $f^{\ast}_i$ and $\trans^{\ast}_i$ in the following Definition \ref{def:f_T} as recursive formulas.

\begin{definition}
 \label{def:f_T}
For $F(f, \trans) \in \mathscr{F}$ and $i \in [F]$, we define a mapping $f_i^{\ast} : \mathcal{S}^{\ast} \rightarrow \mathcal{C}^{\ast}$ and a mapping $\trans_i^{\ast} : \mathcal{S}^{\ast} \rightarrow [F]$ as
\begin{equation}
\label{eq:fstar}
f_i^{\ast}(\pmb{x}) = 
\begin{cases}
\lambda &\,\,\text{if}\,\, \pmb{x} = \lambda,\\
f_i(x_1)f_{\trans_i(x_1)}^{\ast}(\suff(\pmb{x})) &\,\,\text{if}\,\, \pmb{x} \neq \lambda,\\
\end{cases}
\end{equation}
\begin{equation}
\label{eq:tstar}
\trans_i^{\ast}(\pmb{x}) = 
\begin{cases}
i &\,\,\text{if}\,\, \pmb{x} = \lambda,\\
\trans^{\ast}_{\trans_i(x_1)}(\suff(\pmb{x})) &\,\,\text{if}\,\, \pmb{x} \neq \lambda\\
\end{cases}
\end{equation}
for $\pmb{x} = x_1 x_2 \ldots x_{n} \in \mathcal{S}^{\ast}$.
\end{definition}

\begin{example}
We consider $F(f, \trans) \coloneqq F^{(\alpha)}(f^{(\alpha)}, \trans^{(\alpha)})$ of Table \ref{tab:code-tuple}.
Then $f_0^{\ast}(\mathrm{badb})$ and $\trans^{\ast}_0(\mathrm{badb})$ is given as follows (cf. Example \ref{ex:encode}):
\begin{align*}
f_0^{\ast}(\mathrm{badb})
&= f_0(\mathrm{b}) f_1^{\ast}(\mathrm{adb})\\
&= f_0(\mathrm{b}) f_1(\mathrm{a}) f_1^{\ast}(\mathrm{db})\\
&= f_0(\mathrm{b}) f_1(\mathrm{a}) f_1(\mathrm{d}) f_2^{\ast}(\mathrm{b})\\
&= f_0(\mathrm{b}) f_1(\mathrm{a}) f_1(\mathrm{d}) f_2(\mathrm{b}) f_2^{\ast}(\lambda)\\
&= 1000001111110
\end{align*}
and
\begin{align*}
\trans^{\ast}_0(\mathrm{badb})
= \trans^{\ast}_{1}(\mathrm{adb})
= \trans^{\ast}_{1}(\mathrm{db})
= \trans^{\ast}_{2}(\mathrm{b})
= \trans^{\ast}_{2}(\mathrm{\lambda})
= 2.
\end{align*}
\end{example}

The following Lemma \ref{lem:f_T} follows from Definition \ref{def:f_T}.
\begin{lemma}
\label{lem:f_T}
For any $F(f, \trans) \in \mathscr{F}$, $i \in [F]$, and $\pmb{x}, \pmb{y} \in \mathcal{S}^{\ast}$, 
the following statements (i)--(iii) hold.
\begin{itemize}
\item[(i)] $f_i^{\ast}(\pmb{x} \pmb{y}) = f_i^{\ast}(\pmb{x}) f^{\ast}_{\trans_i^{\ast}(\pmb{x})}(\pmb{y})$. 
\item[(ii)] $\trans_i^{\ast}(\pmb{x} \pmb{y}) = \trans^{\ast}_{\trans^{\ast}_i(\pmb{x})}(\pmb{y})$.
\item[(iii)] If $\pmb{x} \preceq \pmb{y}$, then $f^{\ast}_i(\pmb{x}) \preceq f^{\ast}_i(\pmb{y})$.
\end{itemize}
\end{lemma}

\subsection{$k$-bit Delay Decodable Code-tuples}

In Example \ref{ex:encode}, despite $f_0(\mathrm{b})f_1(\mathrm{a}) = 1000$, to uniquely identify $x_1x_2 = \mathrm{ba}$, it is required to read $10000$ including the additional $1$ bit. 
Namely, a decoding delay of $1$ bit occurs to decode $x_2 = \mathrm{a}$.
Similarly, despite $f_0(\mathrm{b})f_0(\mathrm{a})f_1(\mathrm{d}) = 100000111$, to uniquely identify $x_1x_2x_3 = \mathrm{bad}$, it is required to read $100000111\allowbreak 11$ including the additional $2$ bits.
Namely, a decoding delay of $2$ bits occurs to decode $x_3 = \mathrm{d}$.
In general, in the decoding process with $F^{(\alpha)}$ in Table \ref{tab:code-tuple},
it is required to read the additional at most $2$ bits for the decoder to uniquely identify each symbol of a given source sequence.
We say a code-tuple is \emph{$k$-bit delay decodable} if the decoder can always uniquely identify each source symbol by reading the additional $k$ bits of the codeword sequence.
The code-tuple $F^{(\alpha)}$ is an example of a $2$-bit delay decodable code-tuple.
To state the formal definition of a $k$-bit delay decodable code-tuple,
we introduce the following Definitions \ref{def:pref} and \ref{def:prefstar}.

\begin{definition}
\label{def:pref}
For an integer $k \geq 0$, $F(f, \trans) \in \mathscr{F}, i \in [F]$, and $\pmb{b} \in \mathcal{C}^{\ast}$, we define 
\begin{equation}
\label{eq:pref1}
\PREF^k_{F, i}(\pmb{b}) \coloneqq \{\pmb{c} \in \mathcal{C}^k : \pmb{x} \in \mathcal{S}^{+}, f^{\ast}_i(\pmb{x}) \succeq \pmb{b}\pmb{c}, f_i(x_1) \succeq \pmb{b}  \},
\end{equation}
\begin{equation}
\label{eq:pref2}
\bar{\PREF}^k_{F, i}(\pmb{b}) \coloneqq \{\pmb{c} \in \mathcal{C}^k : \pmb{x} \in \mathcal{S}^{+}, f^{\ast}_i(\pmb{x}) \succeq \pmb{b}\pmb{c}, f_i(x_1) \succ \pmb{b}  \},
\end{equation}
where $x_1$ denotes the first symbol of $\pmb{x}$.
Namely, $\PREF^k_{F, i}(\pmb{b})$ (resp. $\bar{\PREF}^k_{F, i}(\pmb{b})$) is the set of all $\pmb{c} \in \mathcal{C}^k$ such that there exists $\pmb{x} = x_1x_2\ldots x_n \in \mathcal{S}^{+}$ satisfying $f^{\ast}_i(\pmb{x}) \succeq \pmb{b}\pmb{c}$ and $f_i(x_1) \succeq \pmb{b}$ (resp. $f_i(x_1) \succ \pmb{b}$).
\end{definition}

\begin{definition}
\label{def:prefstar}
For $F(f, \trans) \in \mathscr{F}, i \in [F]$, and $\pmb{b} \in \mathcal{C}^{\ast}$, we define 
\begin{equation}
\PREF^{\ast}_{F, i}(\pmb{b}) \coloneqq \PREF^{0}_{F, i}(\pmb{b}) \cup \PREF^{1}_{F, i}(\pmb{b}) \cup \PREF^{2}_{F, i}(\pmb{b}) \cup \cdots,
\end{equation}
\begin{equation}
\bar{\PREF}^{\ast}_{F, i}(\pmb{b}) \coloneqq \bar{\PREF}^{0}_{F, i}(\pmb{b}) \cup \bar{\PREF}^{1}_{F, i}(\pmb{b}) \cup \bar{\PREF}^{2}_{F, i}(\pmb{b}) \cup \cdots.
\end{equation}
\end{definition}

We write $\PREF^k_{F, i}(\lambda)$ (resp. $\bar{\PREF}^k_{F, i}(\lambda)$) as $\PREF^k_{F, i}$ (resp. $\bar{\PREF}^k_{F, i}$) for simplicity.
Also, we write $\PREF^{\ast}_{F, i}(\lambda)$ (resp. $\bar{\PREF}^{\ast}_{F, i}(\lambda)$) as $\PREF^{\ast}_{F, i}$ (resp. $\bar{\PREF}^{\ast}_{F, i}$).
We have
\begin{align}
\PREF^k_{F, i} &\overset{(\mathrm{A})}{=} \{\pmb{c} \in \mathcal{C}^k : \pmb{x} \in \mathcal{S}^{+}, f^{\ast}_i(\pmb{x}) \succeq \pmb{c}\}
\overset{(\mathrm{B})}{=} \{\pmb{c} \in \mathcal{C}^k : \pmb{x} \in \mathcal{S}^{\ast}, f^{\ast}_i(\pmb{x}) \succeq \pmb{c}\},\label{eq:pref3}
\end{align}
where (A) follows from (\ref{eq:pref1}), and (B) is justified as follows.
The relation ``$\subseteq$'' holds by $\mathcal{S}^{+} \subseteq \mathcal{S}^{\ast}$.
We show the relation ``$\supseteq$''.
We choose $\pmb{c} \in \mathcal{C}^k$ such that $f^{\ast}_i(\pmb{x}) \succeq \pmb{c}$ for some $\pmb{x} \in \mathcal{S}^{\ast}$ arbitrarily and show that $f^{\ast}_i(\pmb{x}') \succeq \pmb{c}$ for some $\pmb{x}' \in \mathcal{S}^{+}$.
The case $\pmb{x} \in \mathcal{S}^{+}$ is trivial.
In the case $\pmb{x} \in\{\lambda\} = \mathcal{S}^{\ast} \setminus \mathcal{S}^{+}$,
then since $\pmb{c} \preceq f^{\ast}_i(\pmb{x}) = f^{\ast}_i(\lambda) = \lambda$ by (\ref{eq:fstar}), we have $\pmb{c} = \lambda$, which leads to that any $\pmb{x}' \in \mathcal{S}^{+}$ satisfies $f^{\ast}_i(\pmb{x}') \succeq \lambda = \pmb{c}$.
Hence, the relation ``$\supseteq$'' holds.

\begin{example}
For $F^{(\alpha)}$ in Table \ref{tab:code-tuple}, we have
\begin{align*}
\PREF^0_{F^{(\alpha)}, 1}(00) &= \{\lambda\},\\
\PREF^1_{F^{(\alpha)}, 1}(00) &= \{0, 1\},\\
\PREF^2_{F^{(\alpha)}, 1}(00) &= \{00, 01, 10, 11\},\\
\PREF^3_{F^{(\alpha)}, 1}(00) &= \{000, 001, 010, 011, 100, 101, 111\}
\end{align*}
and
\begin{align*}
\bar{\PREF}^0_{F^{(\alpha)}, 1}(00) &= \{\lambda\},\\
\bar{\PREF}^1_{F^{(\alpha)}, 1}(00) &= \{1\},\\
\bar{\PREF}^2_{F^{(\alpha)}, 1}(00) &= \{11\},\\
\bar{\PREF}^3_{F^{(\alpha)}, 1}(00) &= \{111\}.
\end{align*}
For $F^{(\beta)}$ in Table \ref{tab:code-tuple}, we have
\begin{align*}
\PREF^0_{F^{(\beta)}, 2} &= \{\lambda\}, &\bar{\PREF}^0_{F^{(\beta)}, 2} &= \emptyset,\\
\PREF^1_{F^{(\beta)}, 2} &= \emptyset, &\bar{\PREF}^1_{F^{(\beta)}, 2} &= \emptyset.
\end{align*}
\end{example}

We consider the situation where the decoder has already read the prefix $\pmb{b}'$ of a given codeword sequence and identified a prefix $x_1x_2\ldots x_l$ of the original sequence $\pmb{x}$.
Then we have $\pmb{b}' = f_{i_1}(x_1)f_{i_2}(x_2)\ldots f_{i_l}(x_l)\pmb{b}$ for some $\pmb{b} \in \mathcal{C}^{\ast}$.
Put $i \coloneqq i_{l+1}$ and let $\{s_1, s_2, \ldots, s_r\}$ be the set of all symbols $s \in \mathcal{S}$ such that $f_i(s) = \pmb{b}$.
Then there are the following $r+1$ possible cases for the next symbol $x_{l+1}$: the case $x_{l+1} = s_1$, the case $x_{l+1} = s_2$, $\ldots,$ the case $x_{l+1} = s_r$, and the case $f_i(x_{l+1}) \succ \pmb{b}$.
For a code-tuple $F$ to be $k$-bit delay decodable, the decoder must always be able to distinguish these $r+1$ cases by reading the following $k$ bits of the codeword sequence.
Namely, it is required that the $r+1$ sets listed below are disjoint:
\begin{itemize}
\item $\PREF^k_{F, \trans_i(s_1)}$, the set of all possible following $k$ bits in the case $x_{l+1} = s_1$,
\item $\PREF^k_{F, \trans_i(s_2)}$, the set of all possible following $k$ bits in the case $x_{l+1} = s_2$,
\item $\cdots$,
\item $\PREF^k_{F, \trans_i(s_r)}$, the set of all possible following $k$ bits in the case $x_{l+1} = s_r$,
\item $\bar{\PREF}^k_{F, i}(\pmb{b})$, the set of all possible following $k$ bits in the case $f_i(x_{l+1}) \succ \pmb{b}$.
\end{itemize}
This discussion leads to the following Definition \ref{def:k-bitdelay}.

 \begin{definition}
  \label{def:k-bitdelay}
 Let $k \geq 0$ be an integer. 
A code-tuple $F(f, \trans)$ is said to be \emph{$k$-bit delay decodable} if the following conditions (i) and (ii) hold.
\begin{itemize}
\item[(i)] For any $i \in [F]$ and $s \in \mathcal{S}$, it holds that $\PREF^k_{F, \trans_i(s)} \cap \bar{\PREF}^k_{F, i}(f_i(s)) = \emptyset$.
\item[(ii)] For any $i \in [F]$ and $s, s' \in \mathcal{S}$, if $s \neq s'$ and $f_i(s) = f_i(s')$, then $\PREF^k_{F, \trans_i(s)} \cap \PREF^k_{F, \trans_i(s')} =  \emptyset$.
\end{itemize}
 For an integer $k \geq 0$, we define $\mathscr{F}_{k\hdec}$ as the set of all $k$-bit delay decodable code-tuples, that is, 
$\mathscr{F}_{k\hdec} \coloneqq \{F \in \mathscr{F} : F \text{ is } k \text{-bit delay decodable} \}$.
\end{definition}

Definition \ref{def:k-bitdelay} is equivalent to the definition of $k$-bit delay decodable codes in \cite{Hashimoto2022}.
See Appendix \ref{sec:equiv} for the proof.

\begin{example}
We define $F(f, \trans)$ as $F^{(\alpha)}$ in Table \ref{tab:code-tuple}.
Then we have $F \in \mathscr{F}_{2\hdec}$ while $F \not\in \mathscr{F}_{1\hdec}$ because
\begin{equation*}
\PREF^1_{F, \trans_0(\mathrm{a})} \cap \PREF^1_{F, \trans_0(\mathrm{d})}
= \{0, 1\} \cap \{1\} = \{1\} \neq \emptyset,
\end{equation*}
that is, $F$ does not satisfy Definition \ref{def:k-bitdelay} (ii) for $k = 1$.

Next, we define $F(f, \trans)$ as $F^{(\beta)}$ in Table \ref{tab:code-tuple}.
Then we have $F \in \mathscr{F}_{1\hdec}$ while $F \not\in \mathscr{F}_{0\hdec}$ because
\begin{equation*}
\PREF^0_{F, \trans_1(\mathrm{c})} \cap \bar{\PREF}^0_{F, 1}(f_1(\mathrm{c})) = \{\lambda\} \cap \{\lambda\} = \{\lambda\} \neq \emptyset,
\end{equation*}
that is, $F$ does not satisfy Definition \ref{def:k-bitdelay} (i) for $k = 0$.
\end{example}

\begin{remark}
A $k$-bit delay decodable code-tuple $F$ is not necessarily uniquely decodable, that is, the mappings $f^{\ast}_0, f^{\ast}_1, \ldots, f^{\ast}_{|F|-1}$ are not necessarily injective. 
For example, for $F^{(\alpha)} \in \mathscr{F}_{2\hdec}$ in Table \ref{tab:code-tuple}, we have ${f^{(\alpha)}_0}^{\ast}(\mathrm{bc}) = 1000111 = {f^{(\alpha)}_0}^{\ast}(\mathrm{bd})$. In general, it is possible that the decoder cannot uniquely recover the last few symbols of the original source sequence in the case where the rest of the codeword sequence is less than $k$ bits. In such a case, we should append additional information for practical use.
\end{remark}

For $F(f, \trans) \in \mathscr{F}$ and $i \in [F]$, the mapping $f_i$ is said to be \emph{prefix-free} if
for any $s, s^{\prime} \in \mathcal{S}$, if $f_i(s) \preceq f_i(s^{\prime})$, then $s = s^{\prime}$.
A $0$-bit delay decodable code-tuple is characterized as a code-tuple all of which code tables are prefix-free \cite[Lem. 4]{Hashimoto2022}.

\begin{lemma}
\label{lem:0dec-prefixfree}
A code-tuple $F(f, \trans) \in \mathscr{F}$ satisfies $F \in \mathscr{F}_{0\hdec}$ if and only if the code tables $f_0, f_1, \ldots, f_{|F|-1}$ are prefix-free.
\end{lemma}

\subsection{Extendable Code-tuples}

For the code-tuple $F^{(\beta)}$ in Table \ref{tab:code-tuple}, we can see that ${{f^{(\beta)}_2}^\ast}(\pmb{x}) = \lambda$ for any $\pmb{x} \in \mathcal{S}^{\ast}$.
To exclude such abnormal and useless code-tuples, we introduce a class $\mathscr{F}_{\ext}$ in the following Definition \ref{def:F_ext}.

\begin{definition}
\label{def:F_ext}
A code-tuple $F$ is said to be \emph{extendable} if $\PREF^1_{F, i} \neq \emptyset$ for any $i \in [F]$.
We define $\mathscr{F}_{\ext}$ as the set of all extendable code-tuples, that is,
$\mathscr{F}_{\ext} \coloneqq \{F \in \mathscr{F} : {}^{\forall}i \in [F]; \PREF^1_{F, i} \neq \emptyset\}$.
\end{definition}

\begin{example}
For $F^{(\alpha)}$ in Table \ref{tab:code-tuple}, we have
\begin{equation*}
\PREF^1_{F^{(\alpha)}, 0} = \{0, 1\}, \quad \PREF^1_{F^{(\alpha)}, 1} = \{0, 1\}, \quad \PREF^1_{F^{(\alpha)}, 2} = \{1\}.
\end{equation*}
Therefore, we have $F^{(\alpha)} \in \mathscr{F}_{\ext}$.
For $F^{(\beta)}$ in Table \ref{tab:code-tuple},
we have
\begin{equation*}
\PREF^1_{F^{(\beta)}, 0} = \{0, 1\}, \quad \PREF^1_{F^{(\beta)}, 1} = \{0\}, \quad \PREF^1_{F^{(\beta)}, 2} = \emptyset.
\end{equation*}
Since $\PREF^1_{F^{(\beta)}, 2} = \emptyset$, we have $F^{(\beta)} \not\in \mathscr{F}_{\ext}$.
\end{example}

The following Lemma \ref{lem:F_ext} shows that for an extendable code-tuple $F$, we can extend the length of $f^{\ast}_i(\pmb{x})$ as long as we want by appending symbols to $\pmb{x}$ appropriately.

\begin{lemma}
\label{lem:F_ext}
A code-tuple $F(f, \trans)$ is extendable if and only if for any $i \in [F]$ and integer $l \geq 0$, 
there exists $\pmb{x} \in \mathcal{S}^{\ast}$ such that $|f^{\ast}_i(\pmb{x})| \geq l$.
\end{lemma}

\begin{proof}[Proof of Lemma \ref{lem:F_ext}]
(Sufficiency) Fix $i \in [F]$ arbitrarily.
Applying the assumption with $l = 1$, we see that there exists $\pmb{x} \in \mathcal{S}^{\ast}$ such that $|f^{\ast}_i(\pmb{x})| \geq 1$.
Then there exists $c \in \mathcal{C}$ such that $f^{\ast}_i(\pmb{x}) \succeq c$,
which leads to $c \in \PREF^1_{F, i}$ by (\ref{eq:pref3}), that is, $\PREF^1_{F, i} \neq \emptyset$ as desired.

(Necessity) 
Assume $F \in \mathscr{F}_{\ext}$.
We prove by induction for $l$.
The base case $l = 0$ is trivial. We consider the induction step for $l \geq 1$.
By the induction hypothesis, there exists $\pmb{x} \in \mathcal{S}^{\ast}$ such that
\begin{equation}
\label{eq:h9xce9wdd1iv}
|f^{\ast}_i(\pmb{x})| \geq l-1.
\end{equation}
Also, by $F \in \mathscr{F}_{\ext}$, there exists $c \in \PREF^1_{F, \trans^{\ast}_i(\pmb{x})}$.
By (\ref{eq:pref3}), there exists $\pmb{y} \in \mathcal{S}^{\ast}$ such that
\begin{equation}
\label{eq:cak1mzrahy9z}
f^{\ast}_{\trans^{\ast}_i(\pmb{x})}(\pmb{y}) \succeq c.
\end{equation}
Thus, we obtain
\begin{equation}
|f^{\ast}_i(\pmb{x}\pmb{y})| \overset{(\mathrm{A})}{=} |f^{\ast}_i(\pmb{x})| + |f^{\ast}_{\trans^{\ast}_i(\pmb{x})}(\pmb{y})| \overset{(\mathrm{B})}{\geq} (l-1) + 1 = l,
\end{equation}
where
(A) follows from Lemma \ref{lem:f_T} (i),
and (B) follows from (\ref{eq:h9xce9wdd1iv}) and (\ref{eq:cak1mzrahy9z}).
This completes the induction.
\end{proof}

This property leads to the following Lemma \ref{lem:pref-inc} and Corollary \ref{cor:pref-inc}.

\begin{lemma}
\label{lem:pref-inc}
Let $k, k'$ be two integers such that $0 \leq k \leq k'$.
For any $F(f, \trans) \in \mathscr{F}_{\ext}, i \in [F]$, $\pmb{b} \in \mathcal{C}^{\ast}$, and $\pmb
{c} \in \mathcal{C}^k$, the following statements (i) and (ii) hold.
\begin{itemize}
\item[(i)] $\pmb{c} \in \PREF^k_{F, i}(\pmb{b}) \iff {}^{\exists}\pmb{c}' \in \mathcal{C}^{k'-k}; \pmb{c} \pmb{c}' \in \PREF^{k'}_{F, i}(\pmb{b}).$
\item[(ii)] $\pmb{c} \in \bar{\PREF}^k_{F, i}(\pmb{b}) \iff {}^{\exists}\pmb{c}' \in \mathcal{C}^{k'-k}; \pmb{c} \pmb{c}' \in \bar{\PREF}^{k'}_{F, i}(\pmb{b}).$
\end{itemize}
\end{lemma}

\begin{proof}[Proof of Lemma \ref{lem:pref-inc}]
We prove (i) only because (ii) follows by the similar argument. 

($\implies$):
Assume $\pmb{c} \in \PREF^k_{F, i}(\pmb{b})$.
Then by (\ref{eq:pref1}), there exists $\pmb{x} \in \mathcal{S}^{+}$ such that
\begin{equation}
\label{eq:saiykvl29fkc}
f^{\ast}_i(\pmb{x}) \succeq \pmb{b}\pmb{c},
\end{equation}
 \begin{equation}
\label{eq:14t9yng7xr22}
 f_i(x_1) \succeq \pmb{b}.
 \end{equation}
By $F \in \mathscr{F}_{\ext}$ and Lemma \ref{lem:F_ext}, there exists $\pmb{y} \in \mathcal{S}^{\ast}$ such that
\begin{equation}
\label{eq:3xtrpxddvj3t}
|f^{\ast}_{\trans^{\ast}_i(\pmb{x})}(\pmb{y})| \geq k'-k.
\end{equation}
Hence, we have
\begin{equation}
\label{eq:6s5jhli3q894}
|f^{\ast}_i(\pmb{x} \pmb{y})| \overset{(\mathrm{A})}{=}  |f^{\ast}_i(\pmb{x})| + |f^{\ast}_{\trans^{\ast}_i(\pmb{x})}(\pmb{y})| \overset{(\mathrm{B})}{\geq} |\pmb{b}\pmb{c}| + k'-k,
\end{equation}
where
(A) follows from Lemma \ref{lem:f_T} (i),
and (B) follows from (\ref{eq:saiykvl29fkc}) and (\ref{eq:3xtrpxddvj3t}).
By (\ref{eq:saiykvl29fkc}) and (\ref{eq:6s5jhli3q894}),
there exists $\pmb{c}' \in \mathcal{C}^{k'-k}$ such that
\begin{equation}
\label{eq:o055efargj7k}
f^{\ast}_i(\pmb{x} \pmb{y}) \succeq \pmb{b}\pmb{c}\pmb{c}'.
\end{equation}
Equations (\ref{eq:14t9yng7xr22}) and (\ref{eq:o055efargj7k}) lead to $\pmb{c}\pmb{c}' \in \PREF^{k'}_{F, i}(\pmb{b})$ by (\ref{eq:pref1}).

($\impliedby$):
Assume that there exists $\pmb{c}' \in \mathcal{C}^{k'-k}$ such that $\pmb{c}\pmb{c}' \in \PREF^{k'}_{F, i}(\pmb{b})$.
Then by (\ref{eq:pref1}), there exists $\pmb{x} = x_1x_2\ldots x_n \in \mathcal{S}^{+}$ such that
$f^{\ast}_i(\pmb{x}) \succeq \pmb{b}\pmb{c}\pmb{c}'$ and $f_i(x_1) \succeq \pmb{b}$.
This clearly implies $f^{\ast}_i(\pmb{x}) \succeq \pmb{b}\pmb{c}$ and $f_i(x_1) \succeq \pmb{b}$,
which leads to $\pmb{c} \in \PREF^k_{F, i}(\pmb{b})$ by (\ref{eq:pref1}).
\end{proof}

\begin{corollary}
\label{cor:pref-inc}
For  any integer $k \geq 0$, $F \in \mathscr{F}_{\ext}$, $i \in [F]$, and $\pmb{b} \in \mathcal{C}^{\ast}$,
we have $\PREF^k_{F, i}(\pmb{b}) = \emptyset$ (resp. $\bar{\PREF}^k_{F, i}(\pmb{b}) = \emptyset$) if and only if $\PREF^0_{F, i}(\pmb{b}) = \emptyset$ (resp. $\bar{\PREF}^0_{F, i}(\pmb{b}) = \emptyset$).
\end{corollary}

The following Lemma \ref{lem:longest} gives a lower bound of the length of a codeword sequence for $F \in \mathscr{F}_{\ext} \cap \mathscr{F}_{k\hdec}$.
See Appendix \ref{sec:proof-longest} for the proof of Lemma  \ref{lem:longest}.

\begin{lemma}
\label{lem:longest}
For any integer $k \geq 0$, $F(f, \trans) \in \mathscr{F}_{\ext} \cap \mathscr{F}_{k\hdec}, i \in [F]$, and $\pmb{x} \in \mathcal{S}^{\ast}$, we have $|f^{\ast}_i(\pmb{x})| \geq \lfloor |\pmb{x}| / |F|  \rfloor$.
\end{lemma}

\subsection{Average Codeword Length of Code-Tuple}
\label{subsec:evaluation}

We introduce the average codeword length $L(F)$ of a code-tuple $F$.
From now on, we fix an arbitrary probability distribution $\mu$ of the source symbols, that is, a real-valued function $\mu : \mathcal{S} \rightarrow \mathbb{R}$ such that $\sum_{s \in \mathcal{S}} \mu(s) = 1$ and $0 < \mu(s) \leq 1$ for any $s \in \mathcal{S}$.
Note that we exclude the case where $\mu(s) = 0$ for some $s \in \mathcal{S}$ without loss of generality.

First, for $F(f, \trans) \in \mathscr{F}$ and $i, j \in [F]$, we define the transition probability $Q_{i, j}(F)$ as the probability of using the code table $f_j$ next after using the code table $f_i$ in the encoding process.

 \begin{definition}
\label{def:transprobability}
For $F(f, \trans) \in \mathscr{F}$ and $i, j \in [F]$,
we define the \emph{transition probability} $Q_{i,j}(F)$ as
\begin{equation}
\label{eq:9x9htdrx1001}
 Q_{i,j}(F) \coloneqq \sum_{s \in \mathcal{S}, \trans_i(s) = j} \mu(s).
 \end{equation}
We also define the \emph{transition probability matrix} $Q(F)$ as the following $|F|\times|F|$ matrix:
 \begin{equation}
  \left[
    \begin{array}{cccc}
      Q_{0,0}(F) & Q_{0,1}(F) & \cdots & Q_{0, |F|-1}(F) \\
       Q_{1,0}(F) &  Q_{1,1}(F) & \cdots &  Q_{1, |F|-1}(F) \\
      \vdots & \vdots & \ddots & \vdots \\
      Q_{|F|-1, 0}(F) &  Q_{|F|-1, 1}(F) & \cdots &  Q_{|F|-1, |F|-1}(F) 
    \end{array}
  \right].
 \end{equation}
  \end{definition}

We fix $F \in \mathscr{F}$ and consider the encoding process with $F$.
Let $I_i \in [F]$ be the index of the code table used to encode the $i$-th symbol of a source sequence for $i = 1, 2, 3, \ldots$.
Then $\{I_i\}_{i = 1, 2, 3, \ldots}$ is a Markov process with the transition probability matrix $Q(F)$.
We consider a stationary distribution of the Markov process $\{I_i\}_{i = 1, 2, 3, \ldots}$, formally defined as follows.

\begin{definition}
\label{def:stationary}
For $F \in \mathscr{F}$, a solution $\pmb{\pi} = (\pi_0, \allowbreak \pi_1, \ldots, \pi_{|F|-1}) \in \mathbb{R}^{|F|}$ of the following simultaneous equations (\ref{eq:stationary1}) and (\ref{eq:stationary2})
is called a \emph{stationary distribution of $F$}:
\begin{numcases}{}
\pmb{\pi}Q(F) = \pmb{\pi}\label{eq:stationary1},\\
 \sum_{i \in [F]} \pi_i = 1. \label{eq:stationary2}
 \end{numcases}
\end{definition}

A code-tuple has at least one stationary distribution without a negative element as shown in the following Lemma \ref{lem:stationary}.
See Appendix \ref{subsec:proof-stationary} for the proof of Lemma \ref{lem:stationary}.
 \begin{lemma}
 \label{lem:stationary}
 For any $F \in \mathscr{F}$, there exists a stationary distribution $\pmb{\pi} = (\pi_0, \pi_1, \ldots, \pi_{|F|-1})$ of $F$ such that $\pi_i \geq 0$ for any $i \in [F]$.
 \end{lemma}
 
As stated later in Definition \ref{def:evaluation}, the average codeword length $L(F)$ of $F$ is defined depending on the stationary distribution $\pmb{\pi}$ of $F$.
However, it is possible that a code-tuple has multiple stationary distributions.
Therefore, we limit the scope of consideration to a class $\mathscr{F}_{\reg}$ defined as the following Definition \ref{def:regular}, which is the class of code-tuples with a unique stationary distribution.

 \begin{definition}
 \label{def:regular}
A code-tuple $F$ is said to be \emph{regular} if $F$ has a unique stationary distribution.
 We define $\mathscr{F}_{\reg}$ as the set of all regular code-tuples, that is,
$\mathscr{F}_{\reg} \coloneqq \{F \in \mathscr{F} : F \text{ is regular}\}.$
 For $F \in \mathscr{F}_{\reg}$, we define $\pmb{\pi}(F) = (\pi_0(F), \pi_1(F), \ldots, \pi_{|F|-1}(F))$
 as the unique stationary distribution of $F$.
 \end{definition}
 
Since the transition probability matrix $Q(F)$ depends on $\mu$, it might seem that the class $\mathscr{F}_{\reg}$ also depends on $\mu$.
However, we show later as Lemma \ref{lem:kernel} that in fact $\mathscr{F}_{\reg}$ is independent from $\mu$.
More precisely, whether a code-tuple $F(f, \trans)$ belongs to $\mathscr{F}_{\reg}$ depends only on $\trans_0, \trans_1, \ldots, \trans_{|F|-1}$.
 
We also note that for any $F \in \mathscr{F}_{\reg}$, the unique stationary distribution $\pmb{\pi}(F)$ of $F$ satisfies $\pi_i(F) \geq 0$ for any $i \in [F]$ by Lemma \ref{lem:stationary}.
  
The asymptotical performance (i.e., average codeword length per symbol) of a regular code-tuple does not depend on which code table we start encoding:
the average codeword length $L(F)$ of a regular code-tuple $F$ is
the weighted sum of the average codeword lengths of the code tables $f_0, f_1, \ldots, f_{|F|-1}$ weighted by the stationary distribution $\pmb{\pi}(F)$.
Namely, $L(F)$ is defined as the following Definition \ref{def:evaluation}. 
 
 \begin{definition} 
 \label{def:evaluation}
 For $F(f, \trans) \in \mathscr{F}$ and $i \in [F]$, we define the \emph{average codeword length $L_i(F)$ of the single code table} $f_i : \mathcal{S} \rightarrow \mathcal{C}^{\ast}$ as
  \begin{equation}
 L_i(F) \coloneqq \sum_{s \in \mathcal{S}} |f_i(s)| \cdot \mu(s).
  \end{equation}
For $F \in \mathscr{F}_{\reg}$,
we define the \emph{average codeword length $L(F)$ of the code-tuple $F$} as 
 \begin{equation}
 \label{eq:evaluation}
 L(F) \coloneqq \sum_{i \in [F]} \pi_i(F)L_i(F).
 \end{equation}
 \end{definition}

 \begin{example}
 \label{ex:evaluation}
We consider $F \coloneqq F^{(\alpha)}$ of Table \ref{tab:code-tuple}, where $(\mu(\mathrm{a}), \mu(\mathrm{b}), \mu(\mathrm{c}), \mu(\mathrm{d})) \allowbreak = (0.1, 0.2, 0.3, 0.4)$.
We have 
 \begin{equation*}
  Q(F) = \left[
    \begin{array}{cccc}
      0.4 & 0.2 & 0.4 \\
       0.2 &  0.4 & 0.4\\
      0 &  0.1 & 0.9 
    \end{array}
  \right].
\end{equation*}
The code-tuple $F$ has a unique stationary distribution $\pmb{\pi}(F) = (\pi_0(F), \pi_1(F), \allowbreak \pi_2(F)) = (1/20, 3/20, 16/20)$.
Hence, we have $F \in \mathscr{F}_{\reg}$.
Also, we have
\begin{align*}
L_0(F) = 2.6, \quad L_1(F) = 3.7, \quad L_2(F) = 4.2.
\end{align*}
Therefore, $L(F)$ is given as
\begin{align*}
L(F) &= \pi_0(F)L_0(F) + \pi_1(F)L_1(F) + \pi_2(F)L_2(F) 
= 4.045.
\end{align*}
 \end{example}

\begin{remark}
\label{rem:transitionprobability}
Note that $Q(F), L_i(F), L(F)$ and $\pmb{\pi}(F)$ depend on $\mu$.
However, since we are now discussing on a fixed $\mu$,
the average codeword length $L_i(F)$ of $f_i$ (resp. the transition probability matrix $Q(F)$) is determined only by the mapping $f_i$ (resp. $\trans_0, \trans_1, \ldots, \trans_{|F|-1}$) and therefore the stationary distribution $\pmb{\pi}(F)$ of a regular code-tuple $F$ is also determined only by $\trans_0, \trans_1, \ldots, \trans_{|F|-1}$.
\end{remark}

\subsection{Irreducible parts of Code-tuple}

As we can see from (\ref{eq:evaluation}), the code tables $f_i$ of $F(f, \trans) \in \mathscr{F}_{\reg}$ such that $\pi_i(F) = 0$ does not contribute to $L(F)$.
It is useful to remove such non-essential code tables and obtain an \emph{irreducible} code-tuple:
we say that a regular code-tuple $F$ is \emph{irreducible} if $\pi_i(F) > 0$ for any $i \in [F]$ as formally defined later in Definition \ref{def:F_irr}.
In this subsection, we introduce an \emph{irreducible part} of $F \in \mathscr{F}_{\reg}$, which is an irreducible code-tuple obtained by removing all the code tables $f_i$ such that $\pi_i(F) = 0$ from $F$.
The formal definition of an irreducible part of $F$ is stated using a notion of \emph{homomorphism} defined in the following Definition \ref{def:homo}.

\begin{definition}
\label{def:homo}
For $F(f, \trans), F'(f', \trans') \in \mathscr{F}$, a mapping $\varphi: [F'] \rightarrow [F]$ is called a \emph{homomorphism from $F'$ to $F$} if
\begin{equation}
\label{eq:phi1}
f'_i(s) = f_{\varphi(i)}(s),
\end{equation}
\begin{equation}
\label{eq:phi2}
\varphi(\trans'_i(s)) = \trans_{\varphi(i)}(s)
\end{equation}
for any $i \in [F']$ and $s \in \mathcal{S}$.
\end{definition}

Given a homomorphism of code-tuples, the following Lemma \ref{lem:duplicate} holds between the two code-tuples.
See Appendix \ref{subsec:proof-duplicate} for the proof of Lemma  \ref{lem:duplicate}.

\begin{lemma}
\label{lem:duplicate}
For any $F(f, \trans), F'(f', \trans') \in \mathscr{F}$ and a homomorphism $\varphi: [F'] \rightarrow [F]$ from $F'$ to $F$ ,
the following statements (i)--(vi) hold.
\begin{itemize}
\item[(i)] For any $i \in [F']$ and $\pmb{x} \in \mathcal{S}^{\ast}$, we have $f'^{\ast}_i(\pmb{x}) = f^{\ast}_{\varphi(i)}(\pmb{x})$ and $\varphi(\trans'^{\ast}_i(\pmb{x})) = \trans^{\ast}_{\varphi(i)}(\pmb{x})$.
\item[(ii)] For any $i \in [F']$ and $\pmb{b} \in \mathcal{C}^{\ast}$, we have $\PREF^{\ast}_{F', i}(\pmb{b}) = \PREF^{\ast}_{F, \varphi(i)}(\pmb{b})$ and $\bar{\PREF}^{\ast}_{F', i}(\pmb{b}) = \bar{\PREF}^{\ast}_{F, \varphi(i)}(\pmb{b})$.
\item[(iii)] For any stationary distribution $\pmb{\pi}' = (\pi'_0, \pi'_1, \ldots, \allowbreak \pi'_{|F'|-1})$ of $F'$, the vector $\pmb{\pi} = (\pi_0, \pi_1, \ldots, \pi_{|F|-1}) \in \mathbb{R}^{|F|}$ defined as
\begin{equation}
\label{eq:3720vpgjeuph}
\pi_j = \sum_{j' \in \mathcal{A}_j} \pi'_{j'} \,\,\text{for}\,\, j \in [F]
\end{equation}
is a stationary distribution of $F$, where \begin{equation}
\label{eq:twqab1injvre}
\mathcal{A}_i \coloneqq \{i' \in [F'] : \varphi(i') = i\}
\end{equation}
for $i \in [F]$.
\item[(iv)] If $F \in \mathscr{F}_{\ext}$, then $F' \in \mathscr{F}_{\ext}$.
\item[(v)] If $F, F' \in \mathscr{F}_{\reg}$, then $L(F') = L(F)$.
\item[(vi)] For any integer $k \geq 0$, if $F \in \mathscr{F}_{k\hdec}$, then $F' \in \mathscr{F}_{k\hdec}$.
\end{itemize}
\end{lemma}

We also introduce the set $\kernel_F$ for $F \in \mathscr{F}$ as the following Definition \ref{def:kernel}.
We state in Lemma \ref{lem:kernel} that we can characterize a regular code-tuple $F$ by $\kernel_F$.

 \begin{definition}
\label{def:kernel}
For $F(f, \trans) \in \mathscr{F}$, we define $\kernel_F$ as
\begin{equation}
\label{eq:7tuxvj14yeno}
\kernel_F \coloneqq \{i \in [F] : {}^{\forall}j \in [F]; {}^{\exists}\pmb{x} \in \mathcal{S}^{\ast}; \trans^{\ast}_j(\pmb{x}) = i\}.
\end{equation}
Namely, $\kernel_F$ is the set of indices $i$ of the code tables such that for any $j \in [F]$, there exists $\pmb{x} \in \mathcal{S}^{\ast}$ such that $\trans^{\ast}_j(\pmb{x}) = i$.
\end{definition}

\begin{example}
For $F^{(\alpha)}$ and $F^{(\beta)}$ in Table \ref{tab:code-tuple}, we have 
$\kernel_{F^{(\alpha)}} = \{0, 1, 2\}$ and $\kernel_{F^{(\beta)}} = \emptyset$.
\end{example}

\begin{lemma}
\label{lem:kernel}
For any $F \in \mathscr{F}$, the following statements (i) and (ii) hold.
\begin{itemize}
\item[(i)] $F \in \mathscr{F}_{\reg}$ if and only if $\kernel_F \neq \emptyset$.
\item[(ii)] If $F \in \mathscr{F}_{\reg}$, then for any $i \in [F]$, the following equivalence relation holds: $\pi_i(F) > 0 \iff i \in \kernel_F$.
\end{itemize}
\end{lemma}
The proof of Lemma \ref{lem:kernel} is given in Appendix. \ref{subsec:proof-kernel}.

By Lemma \ref{lem:kernel} (ii), a regular code-tuple $F(f, \trans)$ satisfies $\pi_i(F) > 0$ for any $i \in [F]$ if and only if $F$ is an irreducible code-tuple defined as follows.

\begin{definition}
\label{def:F_irr}
A code-tuple $F$ is said to be \emph{irreducible} if $\kernel_F = [F]$.
We define $\mathscr{F}_{\irr}$ as the set of all irreducible code-tuples, that is,
$\mathscr{F}_{\irr} \coloneqq \{F \in \mathscr{F} :\kernel_F = [F]\}.$
\end{definition}
Note that $\mathscr{F}_{\irr} \subseteq \mathscr{F}_{\reg}$ since $F \in \mathscr{F}_{\reg}$ is equivalent to $\kernel_{F} \neq \emptyset$ by Lemma \ref{lem:kernel} (i).

Now we define an \emph{irreducible part} $\bar{F}$ of a code-tuple $F$ as the following Definition \ref{def:irr-part}.

\begin{definition}
\label{def:irr-part}
An irreducible code-tuple $\bar{F}$ is called an \emph{irreducible part} of a code-tuple $F$ if
there exists an injective homomorphism $\varphi : [\bar{F}] \rightarrow [F]$ from $\bar{F}$ to $F$.
\end{definition}

The following property of $\bar{F}$ is immediately from Definition \ref{def:irr-part} and Lemma \ref{lem:duplicate} (iv)--(vi).

\begin{lemma}
\label{lem:remove}
For any integer $k \geq 0$, $F \in \mathscr{F}_{\reg} \cap \mathscr{F}_{\ext} \cap \mathscr{F}_{k\hdec}$, and an irreducible part $\bar{F}$ of $F$, we have $\bar{F} \in \mathscr{F}_{\irr} \cap \mathscr{F}_{\ext} \cap \mathscr{F}_{k\hdec}$ and  $L(\bar{F}) = L(F)$.
\end{lemma}

The existence of an irreducible part is guaranteed as the following Lemma  \ref{lem:irr-part}.
See Appendix \ref{subsec:proof-irr-part} for the proof of Lemma  \ref{lem:irr-part}.

\begin{lemma}
\label{lem:irr-part}
For any $F \in \mathscr{F}_{\reg}$, there exists an irreducible part $\bar{F}$ of $F$.
\end{lemma}

\section{Main Results}
\label{sec:main}

In this section, we discuss the average codeword length for code-tuples of the class $\mathscr{F}_{\reg} \cap \mathscr{F}_{\ext} \cap \mathscr{F}_{k\hdec}$ for $k \geq 0$ and prove Theorems \ref{thm:differ} and \ref{thm:complete} as the main results of this paper.

\subsection{Theorem \ref{thm:differ}}

The first theorem claims that for any $F \in \mathscr{F}_{\reg} \cap \mathscr{F}_{\ext} \cap \mathscr{F}_{k\hdec}$, there exists $F^{\dagger} \in \mathscr{F}_{\irr} \cap \mathscr{F}_{\ext} \cap \mathscr{F}_{k\hdec}$ such that $L(F^{\dagger}) \leq L(F)$ and $\PREF^k_{F^{\dagger}, 0}, \PREF^k_{F^{\dagger}, 1}, \ldots, \PREF^k_{F^{\dagger}, |F^{\dagger}|-1}$ are distinct.
Namely, Theorem \ref{thm:differ} guarantees that it suffices to consider only irreducible code-tuples with at most $2^{(2^k)}$ code tables to achieve a short average codeword length.
In particular, it is not the case that one can achieve an arbitrarily small average codeword length by using arbitrarily many code tables.
To state Theorem \ref{thm:differ}, we prepare the following Definition \ref{def:prefset}.

\begin{definition}
\label{def:prefset}
For an integer $k \geq 0$ and $F \in \mathscr{F}$, we define $\prefset^k_F$ as
\begin{equation}
\prefset^k_F \coloneqq \{\PREF^k_{F, i} : i \in [F]\}.
\end{equation}
\end{definition}

\begin{example}
For $F^{(\alpha)}$ in Table \ref{tab:code-tuple}, we have 
\begin{align*}
\prefset^0_{F^{(\alpha)}} &= \{\{\lambda\}\},\\
\prefset^1_{F^{(\alpha)}} &= \{\{0, 1\}, \{1\}\},\\
\prefset^2_{F^{(\alpha)}} &= \{\{01, 10\}, \{00, 01, 10\}, \{11\}\}.
\end{align*}
\end{example}

Note that $\PREF^k_{F, 0}, \PREF^k_{F, 1}, \ldots, \PREF^k_{F, |F|-1}$ are distinct if and only if $|\prefset^k_F| = |F|$.
Also, note that the following Lemma \ref{lem:irr-prefset} holds by Lemma \ref{lem:duplicate} (ii).
\begin{lemma}
\label{lem:irr-prefset}
For any integer $k \geq 0$, $F \in \mathscr{F}_{\reg}$, and an irreducible part $\bar{F}$ of $F$,
we have $\prefset^k_{\bar{F}} \subseteq \prefset^k_{F}$.
\end{lemma}

Using Definition \ref{def:prefset}, we state Theorem \ref{thm:differ} as follows.

\begin{theorem}
\label{thm:differ}
For any integer $k \geq 0$ and $F \in \mathscr{F}_{\reg} \cap \mathscr{F}_{\ext} \cap \mathscr{F}_{k\hdec}$, there exists $F^{\dagger} \in \mathscr{F}$ satisfying the following conditions (a)--(d).
\begin{itemize}
\item[(a)] $F^{\dagger}  \in \mathscr{F}_{\irr} \cap \mathscr{F}_{\ext} \cap \mathscr{F}_{k\hdec}$.
\item[(b)] $L(F^{\dagger}) \leq L(F)$.
\item[(c)] $\prefset^k_{F^{\dagger}} \subseteq \prefset^k_F$.
\item[(d)] $|\prefset^k_{F^{\dagger}}| = |F^{\dagger}|$.
\end{itemize}
\end{theorem}

\begin{table}
\caption{The code-tuple $F^{(\delta)}$ is an optimal $2$-bit delay decodable code-tuple satisfying Theorem \ref{thm:differ} (a)--(d) with $F = F^{(\gamma)}$, where $(\mu(\mathrm{a}), \mu(\mathrm{b}), \mu(\mathrm{c}), \mu(\mathrm{d})) = (0.1, 0.2, \allowbreak 0.3, 0.4)$}
\label{tab:optimal}
\centering
\begin{tabular}{c | lclclclc}
\hline
$s \in \mathcal{S}$ & $f^{(\gamma)}_0$ & $\trans^{(\gamma)}_0$ & $f^{(\gamma)}_1$ & $\trans^{(\gamma)}_1$ & $f^{(\gamma)}_2$ & $\trans^{(\gamma)}_2$ & $f^{(\gamma)}_3$ & $\trans^{(\gamma)}_3$\\
\hline
a & 0010 & 2 &100 & 1 & 1100 & 1 & 010 & 0\\
b & 0011 & 0 & 00 & 0 & $11$ & 2 & $011$ & 1\\
c & 000 & 1 & 01 & 1 & 01 & 1 & 100 & 0\\
d & $\lambda$ & 2 & 1 & 2 & 10 & 0 & 1 & 2\\
\hline
\end{tabular}

\vspace{8pt}

\begin{tabular}{c | lclclc}
\hline
$s \in \mathcal{S}$ & $f^{(\delta)}_0$ & $\trans^{(\delta)}_0$ & $f^{(\delta)}_1$ & $\trans^{(\delta)}_1$\\
\hline
a & 100 & 0 & 1100 & 0\\
b & 00 & 0 & $11$ & 1\\
c & 01 & 0 & 01 & 0\\
d & 1 & 1 & 10 & 0\\
\hline
\end{tabular}
\end{table}

\begin{example}
Let $(\mu(\mathrm{a}), \mu(\mathrm{b}), \mu(\mathrm{c}), \mu(\mathrm{d})) = (0.1, 0.2, \allowbreak 0.3, 0.4)$ and $F \coloneqq F^{(\gamma)}$ in Table \ref{tab:optimal}.
Then we have $F \in \mathscr{F}_{\reg} \cap \mathscr{F}_{\ext} \cap \mathscr{F}_{2\hdec}$, $L(F) \approx 1.98644$, and
$\prefset^2_F = \{\{00, 01, 10, 11\}, \{01, 10, 11\}\}$.
The code-tuple $F^{\dag} \coloneqq F^{(\delta)}$ in Table \ref{tab:optimal} satisfies Theorem \ref{thm:differ} (a)--(d) because $\kernel_{F^{\dag}} = \{0, 1\} = [F^{\dag}]$, $L(F^{\dag}) = 1.8667 \leq L(F)$, and $\prefset^2_{F^{\dag}} = \{\{00, 01, 10, 11\}, \{01, 10, 11\}\}$.
\end{example}

\begin{example}
We confirm that Theorem \ref{thm:differ} holds for $k = 0$.
Choose $F \in \mathscr{F}_{\reg} \cap \mathscr{F}_{\ext} \cap \mathscr{F}_{0\hdec}$ arbitrarily and define
$F^{\dag}(f^{\dag}, \trans^{\dag}) \in \mathscr{F}^{(1)}$ as
\begin{align}
f^{\dag}_0(s) &= f_p(s), \label{eq:9au4ayripz8y}\\
\trans^{\dag}_0(s) &= 0
\end{align}
for $s \in \mathcal{S}$, where
\begin{equation}
\label{eq:sjdc1c0vvtgn}
p \in \argmin_{i \in [F]} L_i(F).
\end{equation}
Namely, $F^{\dag}$ is the $1$-code-tuple consisting of the most efficient code table of $F$.

We can see that $F^{\dag}$ satisfies Theorem \ref{thm:differ} (a)--(d) as follows.
\begin{itemize}
\item[(a)] We obtain $F^{\dag} \in \mathscr{F}_{\irr}$ directly from $|F^{\dag}| = 1$.
By $F \in \mathscr{F}_{0\hdec}$ and Lemma \ref{lem:0dec-prefixfree}, all code tables of $F$ are prefix-free.
In particular, $f^{\dag}_0 = f_p$ is prefix-free and thus $F^{\dag} \in \mathscr{F}_{0\hdec}$.
Moreover, since $f^{\dag}_0$ is prefix-free and $\sigma \geq 2$, we have $f^{\dag}_0(s) \neq \lambda$ for some $s \in \mathcal{S}$, which shows $F^{\dag} \in \mathscr{F}_{\ext}$.
\item[(b)] We have
\begin{align*}
L(F^{\dag}) &= L_0(F^{\dag}) \overset{(\mathrm{A})}{=}  L_p(F)
= \sum_{i \in [F]}\pi_i(F)L_p(F) \overset{(\mathrm{B})}{\leq} \sum_{i \in [F]}\pi_i(F)L_i(F)
= L(F),
\end{align*}
where
(A) follows from (\ref{eq:9au4ayripz8y}),
and (B) follows from (\ref{eq:sjdc1c0vvtgn}).
\item[(c)] By $\prefset_{F^{\dag}}^0 = \{\{\lambda\}\} = \prefset_{F}^0$.
\item[(d)] By $|\prefset_{F^{\dag}}^0| = |\{\{\lambda\}\}| = 1 = |F^{\dag}|$.
\end{itemize}
\end{example}

As a preparation for the proof of Theorem \ref{thm:differ}, we state the following Lemmas \ref{lem:pref-unchanged}--\ref{lem:improve}.
See Appendix \ref{subsec:proof-pref-unchanged}--\ref{subsec:proof-improve} for the proofs of Lemmas \ref{lem:pref-unchanged}, \ref{lem:chooseone}, and \ref{lem:improve}.

\begin{lemma}
\label{lem:pref-unchanged}
Let $k \geq 0$ be an integer and let $F(f, \trans)$ and $F'(f', \trans')$ be code-tuples such that $|F| = |F'|$.
Assume that the following conditions (a) and (b) hold.
\begin{itemize}
\item[(a)] $f_i(s) = f'_i(s)$ for any $i \in [F]$ and $s \in \mathcal{S}$.
\item[(b)] $\PREF^k_{F, \trans_i(s)} = \PREF^k_{F, \trans'_i(s)}$ for any $i \in [F]$ and $s \in \mathcal{S}$.
\end{itemize}
Then the following statements (i)--(iii) hold.
\begin{itemize}
\item[(i)] For any $i \in [F']$ and $\pmb{b} \in \mathcal{C}^{\ast}$, we have $\PREF^k_{F, i}(\pmb{b}) = \PREF^k_{F', i}(\pmb{b})$ and $\bar{\PREF}^k_{F, i}(\pmb{b}) = \bar{\PREF}^k_{F', i}(\pmb{b})$.
\item[(ii)] If $F \in \mathscr{F}_{\ext}$, then $F' \in \mathscr{F}_{\ext}$.
\item[(iii)] If $F \in \mathscr{F}_{k\hdec}$, then $F' \in \mathscr{F}_{k\hdec}$.
\end{itemize}
\end{lemma}

\begin{lemma}
\label{lem:chooseone}
For any $F(f, \trans) \in \mathscr{F}_{\irr}$, $\mathcal{I} \subseteq [F]$, and $p \in \mathcal{I}$,
the code-tuple $F'(f', \trans') \in \mathscr{F}^{(|F|)}$ defined as (\ref{eq:cp8qdsjrwe14}) and (\ref{eq:hbm12zjixhzy}) satisfies $F' \in \mathscr{F}_{\reg}$:
\begin{align}
f'_i(s) &= f_i(s), \label{eq:cp8qdsjrwe14}\\
\trans'_i(s) &= 
\begin{cases}
p & \,\,\text{if}\,\, \trans_i(s) \in \mathcal{I},\\
\trans_i(s) & \,\,\text{if}\,\, \trans_i(s) \not\in \mathcal{I}
\end{cases}
\label{eq:hbm12zjixhzy}
\end{align}
for $i \in [F']$ and $s \in \mathcal{S}$.
\end{lemma}

\begin{lemma}
\label{lem:potential}
For any $F \in \mathscr{F}$, there exists $(h_0, h_1, \ldots, \allowbreak h_{|F|-1}) \in \mathbb{R}^{|F|}$ satisfying
\begin{equation}
\label{eq:potential}
{}^{\forall}i \in [F]; L(F) = L_i(F) + \sum_{j \in [F]}(h_j - h_i)Q_{i, j}(F).
\end{equation}
\end{lemma}

See \cite[Sec. 8.2]{Puterman} for proof of Lemma \ref{lem:potential}.
The vector $h$ called ``bias'' defined as \cite[(8.2.2)]{Puterman} satisfies (\ref{eq:potential}) of this paper.
This fact is shown as \cite[(8.2.12)]{Puterman} in \cite[Theorem 8.2.6]{Puterman},
where $g, r,$ and $P$ in \cite[(8.2.12)]{Puterman} correspond to the notations of this paper as follows:
\begin{equation*}
g =  \left[
    \begin{array}{c}
      L(F)\\
      L(F)\\
      \vdots\\
      L(F)\\
    \end{array}
  \right],\quad
 r =  \left[
    \begin{array}{c}
      L_0(F)\\
      L_1(F)\\
      \vdots\\
      L_{|F|-1}(F)\\
    \end{array}
  \right],\quad
 P = Q(F).
\end{equation*}

A real vector $(h_0, h_1, \ldots, h_{|F|-1})$ satisfying (\ref{eq:potential}) is not unique.
We refer to arbitrarily chosen one of them as $h(F) = (h_0(F), h_1(F), \ldots, h_{|F|-1}(F))$.

\begin{lemma}
\label{lem:improve}
For any $F(f, \trans), F'(f', \trans') \in \mathscr{F}_{\reg}$ such that $|F| = |F'|$,
if the following conditions (a) and (b) hold, then $L(F') \leq L(F)$.
\begin{itemize}
\item[(a)] $L_i(F) = L_i(F')$ for any $i \in [F]$.
\item[(b)] $h_{\trans_i(s)}(F) \geq h_{\trans'_i(s)}(F)$ for any $i \in [F]$ and $s \in \mathcal{S}$.
\end{itemize}
\end{lemma}

Using these lemmas, we now prove Theorem \ref{thm:differ}.

\begin{proof}[Proof of Theorem \ref{thm:differ}]
We fix an integer $k \geq 0$ arbitrarily and prove Theorem \ref{thm:differ} by induction for $|F|$.
For the base case $|F| = 1$, the code-tuple $F^{\dagger} \coloneqq F$ satisfies (a)--(d) of Theorem \ref{thm:differ} as desired.
We now consider the induction step for $|F| \geq 2$.

We consider an irreducible part $\bar{F}(\bar{f}, \bar{\trans})$ of $F$.
By Lemmas \ref{lem:remove} and \ref{lem:irr-prefset}, the following statements ($\bar{\mathrm{a}}$)--($\bar{\mathrm{c}}$) hold (cf. (a)--(c) of Theorem \ref{thm:differ}).
\begin{itemize}
\item[($\bar{\mathrm{a}}$)] $\bar{F}  \in \mathscr{F}_{\irr} \cap \mathscr{F}_{\ext} \cap \mathscr{F}_{k\hdec}$.
\item[($\bar{\mathrm{b}}$)] $L(\bar{F}) = L(F)$.
\item[($\bar{\mathrm{c}}$)] $\prefset^k_{\bar{F}} \subseteq \prefset^k_F$.
\end{itemize}
Therefore, if $|\prefset^k_{\bar{F}}| = |\bar{F}|$, then $F^{\dagger} \coloneqq \bar{F}$ satisfies (a)--(d) of Theorem \ref{thm:differ} as desired.
Thus, we now assume $|\prefset^k_{\bar{F}}| < |\bar{F}|$.
Then we can choose $i', j' \in [\bar{F}]$ such that $i' \neq j'$ and $\PREF^k_{\bar{F}, i'}  = \PREF^k_{\bar{F}, j'}$ by pigeonhole principle.
We define
$F'(f', \trans') \in \mathscr{F}^{(|\bar{F}|)}$ as
\begin{align}
f'_i(s) &= \bar{f}_i(s), \label{eq:21l9xdppl0h1}\\
\trans'_i(s) &= 
\begin{cases}
p & \,\,\text{if}\,\, \bar{\trans}_i(s) \in \mathcal{I},\\
\bar{\trans}_i(s) & \,\,\text{if}\,\, \bar{\trans}_i(s) \not\in \mathcal{I}
\end{cases}
\label{eq:ynxygdkzsrk1}
\end{align}
for $i \in [F']$ and $s \in \mathcal{S}$,
where
\begin{equation}
\mathcal{I} \coloneqq \{i \in [\bar{F}] : \PREF^k_{\bar{F}, i}  = \PREF^k_{\bar{F}, i'} (= \PREF^k_{\bar{F},  j'})\}
\end{equation}
and we choose 
\begin{equation}
\label{eq:tuhir8am5say}
p \in \arg\min_{i \in \mathcal{I}} h_i(\bar{F})
\end{equation}
arbitrarily.

Then we obtain $F' \in \mathscr{F}_{\reg}$ by applying Lemma \ref{lem:chooseone} since $\bar{F} \in \mathscr{F}_{\irr}$.
Also, we obtain $F' \in \mathscr{F}_{\ext} \cap \mathscr{F}_{k\hdec}$
and
\begin{equation}
\label{eq:55me7vpzr8he}
\prefset^k_{F'} = \prefset^k_{\bar{F}}
\end{equation}
for any $i \in [F']$ by applying Lemma \ref{lem:pref-unchanged} (i)--(iii) since $\bar{f}_i(s) = f'_i(s)$ and $\PREF^k_{\bar{F}, \bar{\trans}_i(s)} = \PREF^k_{\bar{F}, \trans'_i(s)}$ for any $i \in [\bar{F}]$ and $s \in \mathcal{S}$ by (\ref{eq:21l9xdppl0h1}) and (\ref{eq:ynxygdkzsrk1}).
Moreover, we can see
\begin{equation}
\label{eq:hpvmxtw3pohj}
L(F') \leq L(\bar{F})
\end{equation} by applying Lemma \ref{lem:improve} because $F'$ satisfies (a) (resp. (b)) of Lemma \ref{lem:improve} by (\ref{eq:21l9xdppl0h1}) (resp. (\ref{eq:ynxygdkzsrk1})--(\ref{eq:tuhir8am5say})).

Since $|\mathcal{I}| \geq |\{i', j'\}| \geq 2$, we have $\mathcal{I} \setminus \{p\} \neq \emptyset$.
Also, for any $i \in \mathcal{I} \setminus \{p\}$, we have $i \not\in \kernel_{F'}$ since for any $j \in [F'] \setminus \{i\}$, there exists no $\pmb{x} \in \mathcal{S}^{\ast}$ such that $\trans'^{\ast}_j(\pmb{x}) = i$ by (\ref{eq:ynxygdkzsrk1}).
Therefore, we have
\begin{equation}
\label{eq:jmckcr9oz3s3}
\kernel_{F'} \subsetneq [F'].
\end{equation}

For an irreducible part $\bar{F'}$ of $F'$, we have
\begin{equation}
|\bar{F'}|
= |\kernel_{F'}|
\overset{(\mathrm{A})}{<} |F'|
= |\bar{F}|
= |\kernel_F|
\leq |F|,
\end{equation}
where
(A) follows from $(\ref{eq:jmckcr9oz3s3})$.
Therefore, by applying the induction hypothesis to $\bar{F'}$, we can see that there exists $F^{\dagger} \in \mathscr{F}$ satisfying the following conditions ($\mathrm{a}^{\dagger}$)--($\mathrm{d}^{\dagger}$).
\begin{itemize}
\item[($\mathrm{a}^{\dagger}$)] $F^{\dagger} \in \mathscr{F}_{\irr} \cap \mathscr{F}_{\ext} \cap \mathscr{F}_{k\hdec}$.
\item[($\mathrm{b}^{\dagger}$)] $L(F^{\dagger}) \leq L(\bar{F'})$.
\item[($\mathrm{c}^{\dagger}$)] $\prefset^k_{F^{\dagger}} \subseteq \prefset^k_{\bar{F'}}$.
\item[($\mathrm{d}^{\dagger}$)] $|\prefset^k_{F^{\dagger}}| = |F^{\dagger}|$.
\end{itemize}
We can see that $F^{\dagger}$ is a desired code-tuple, that is, $F^{\dagger}$ satisfies (a)--(d) of Theorem \ref{thm:differ} as follows.
First, (a) and (d) are directly from ($\mathrm{a}^{\dagger}$) and ($\mathrm{d}^{\dagger}$), respectively.
We obtain (b) as follows:
\begin{equation}
L(F^{\dagger}) \overset{(\mathrm{A})}{\leq} L(\bar{F'})
\overset{(\mathrm{B})}{=} L(F')
\overset{(\mathrm{C})}{\leq} L(\bar{F})
\overset{(\mathrm{D})}{=} L(F),
\end{equation}
where
(A) follows from ($\mathrm{b}^{\dagger}$),
(B) follows Lemma \ref{lem:remove},
(C) follows from (\ref{eq:hpvmxtw3pohj}),
and (D) follows Lemma \ref{lem:remove}.
The condition (c) holds because
\begin{align*}
\prefset^k_{F^{\dagger}}
\overset{(\mathrm{A})}\subseteq \prefset^k_{\bar{F'}}
\overset{(\mathrm{B})}\subseteq \prefset^k_{F'}
\overset{(\mathrm{C})}{=} \prefset^k_{\bar{F}}
\overset{(\mathrm{D})}\subseteq \prefset^k_{F},
\end{align*}
where
(A) follows from ($\mathrm{c}^{\dagger}$),
(B) follows from Lemma \ref{lem:irr-prefset},
(C) follows from (\ref{eq:55me7vpzr8he}),
and (D) follows from Lemma \ref{lem:irr-prefset}.
\end{proof}

As a consequence of Theorem \ref{thm:differ}, we can prove the existence of an \emph{optimal} $k$-bit delay decodable code-tuple, that is,
$F^{\ast} \in \mathscr{F}_{\reg} \cap \mathscr{F}_{\ext} \cap \mathscr{F}_{k\hdec}$ such that 
$L(F^{\ast}) \leq L(F)$ for any $F \in \mathscr{F}_{\reg} \cap \mathscr{F}_{\ext} \cap \mathscr{F}_{k\hdec}$.
We prove this fact in Appendix \ref{sec:proof-goodsetTopt}.

We define $\mathscr{F}_{k\hopt}$ as the set of all optimal $k$-bit delay decodable code-tuples as following Definition \ref{def:optimalset}.
\begin{definition}
\label{def:optimalset}
For an integer $k \geq 0$, we define 
\begin{equation}
\label{eq:mk8q01zbvbov}
\mathscr{F}_{k\hopt} \coloneqq \argmin_{F \in \mathscr{F}_{\reg} \cap \mathscr{F}_{\ext} \cap \mathscr{F}_{k\hdec}} L(F).
\end{equation}
\end{definition}
Note that $\mathscr{F}_{k\hopt}$ depends on the probability distribution $\mu$ of the source symbols, and we are now discussing on an arbitrarily fixed $\mu$.

\begin{example}
\label{ex:optimal}
Let $(\mu(\mathrm{a}), \mu(\mathrm{b}), \mu(\mathrm{c}), \mu(\mathrm{d})) = (0.1, 0.2, \allowbreak 0.3, 0.4)$.
Then the code-tuple $F^{(\delta)}$ in Table \ref{tab:optimal} is an optimal $2$-bit delay decodable code-tuple with $L(F^{(\delta)}) \approx 1.8667$.
\end{example}

\subsection{Theorem \ref{thm:complete}}

Theorem \ref{thm:complete} gives a necessary condition for $F \in \mathscr{F}_{\reg} \cap \mathscr{F}_{\ext} \cap \mathscr{F}_{k\hdec}$ to be optimal.
Recall that every internal node in a code-tree of Huffman code has two child nodes because of its optimality.
This leads to that any bit sequence is a prefix of codeword sequence of some source sequence.
More formally, 
\begin{equation}
\label{eq:uf124xag8678}
{}^{\forall}\pmb{b} \in \mathcal{C}^{\ast}; {}^{\exists} \pmb{x} \in \mathcal{S}^{\ast}; f_{\Huff}(\pmb{x}) \succeq \pmb{b},
\end{equation}
where $f_{\Huff}(\pmb{x})$ is the codeword sequence of $\pmb{x}$ with the Huffman code.
The following Theorem \ref{thm:complete} is a generalization of this property of Huffman codes to $k$-bit delay decodable code-tuples for $k \geq 0$.

\begin{theorem}
\label{thm:complete}
For any integer $k \geq 0$, $F \in \mathscr{F}_{k\hopt}$, $i \in \kernel_F$, and $\pmb{b} = b_1b_2\ldots b_l \in \mathcal{C}^{\geq k}$, if $b_1b_2 \ldots b_k \in \PREF^k_{F, i}$, then $\pmb{b} \in \PREF^{\ast}_{F, i}$.
\end{theorem}

\begin{remark}
A Huffman code is represented by a $1$-code-tuple $F \in \mathscr{F}^{(1)}$.
We have $F \in \mathscr{F}_{0\hopt}$ by the optimality of Huffman codes.
Applying Theorem \ref{thm:complete} to $F$ with $k = 0$, we obtain
\begin{equation*}
{}^{\forall}\pmb{b} \in \mathcal{C}^{\ast}; \pmb{b} \in \PREF^{\ast}_{F, 0},
\end{equation*}
which is equivalent to (\ref{eq:uf124xag8678}), and thus Theorem \ref{thm:complete} is indeed a generalization of the property (\ref{eq:uf124xag8678}) of Huffman codes.
\end{remark}

\begin{table}
\caption{An example of $\pmb{x} \in \mathcal{S}^{\ast}$ such that $f^{\ast}_i(\pmb{x}) \succeq \pmb{b}$, where $F(f, \trans) \coloneqq F^{(\delta)}$ in Table \ref{tab:optimal}, $i \in \{0, 1\}$, and $\pmb{b} \in \mathcal{C}^3$}
\label{tab:complete}
\centering
\begin{tabular}{c | cccccccc}
\hline
\diagbox{i}{\pmb{b}}& $000$ & $001$ & $010$ & $011$ & $100$ & $101$ & $110$ & $111$ \\
\hline
$0$ & bb & ba & cb & ca & a & dc & dd & db \\
$1$ & - & - & cb & ca & db & da & a & ba\\
\hline
\end{tabular}
\end{table}

\begin{example}
For $F(f, \trans) \coloneqq F^{(\delta)}$ in Table \ref{tab:optimal}, we have $F \in \mathscr{F}_{2\hopt}$ for $(\mu(\mathrm{a}), \mu(\mathrm{b}), \mu(\mathrm{c}), \mu(\mathrm{d})) = (0.1, 0.2, \allowbreak 0.3, 0.4)$ (cf. Example \ref{ex:optimal}).
Theorem \ref{thm:complete} claims that for any $i \in \kernel_F = \{0, 1\}$ and $\pmb{b} \in \mathcal{C}^{\geq 2}$ such that $b_1b_2 \in \PREF^2_{F, i}$, it holds that
$\pmb{b} \in \PREF^{\ast}_{F, i}$, that is, there exists $\pmb{x} \in \mathcal{S}^{\ast}$ such that $f^{\ast}_i(\pmb{x}) \succeq \pmb{b}$.

For $i \in \{0, 1\}$ and $\pmb{b} \in \mathcal{C}^3$ such that $b_1b_2 \in \PREF^2_{F, i}$, 
Table \ref{tab:complete} shows an example of $\pmb{x} \in \mathcal{S}^{\ast}$ such that $f^{\ast}_i(\pmb{x}) \succeq \pmb{b}$.
For example, we have $f^{\ast}_0(\mathrm{ca}) \succeq 011$ and $f^{\ast}_1(\mathrm{ba}) \succeq 111$.
Note that $b_1b_2 \in \PREF^2_{F, i}$ does not hold for $(i, \pmb{b}) = (1, 000)$ and $(i, \pmb{b}) = (1, 001)$.
\end{example}

\begin{proof}[Proof of Theorem \ref{thm:complete}]

We prove by contradiction assuming that there exist $p \in \kernel_F$ and $\pmb{b} = b_1b_2\ldots b_l \in \mathcal{C}^{\geq k}$ such that 
\begin{equation}
\label{eq:lusqxp3wrnhn}
\pmb{b} \not\in \PREF^{\ast}_{F, p}, \quad b_1b_2 \ldots b_k \in \PREF^k_{F, p}.
\end{equation}
Without loss of generality, we assume $p = |F|-1$ and $\pmb{b}$ is the shortest sequence satisfying (\ref{eq:lusqxp3wrnhn}).
Because we have $l > k$ by (\ref{eq:lusqxp3wrnhn}), we have $\pref(\pmb{b}) \succeq b_1b_2 \ldots b_k \in \PREF^k_{F, |F|-1}$.
Since $\pmb{b}$ is the shortest sequence satisfying (\ref{eq:lusqxp3wrnhn}), it must hold that $\pref(\pmb{b}) \in \PREF^{\ast}_{F, |F|-1}$.
Hence, by $F \in \mathscr{F}_{\ext}$ and Lemma \ref{lem:pref-inc} (i), we have $\pmb{d} = d_1d_2\ldots d_l \coloneqq \pref(\pmb{b})\bar{b_l} \in \PREF^{\ast}_{F, |F|-1}$,
where $\bar{c}$ denotes the negation of $c \in \mathcal{C}$, that is, $\bar{0} \coloneqq 1$ and $\bar{1} \coloneqq 0$.
Namely, we have
\begin{equation}
\label{eq:2cz6211gg810}
\pmb{d} \in \PREF^{\ast}_{F, |F|-1}, \quad \pref(\pmb{d})\bar{d_l} = \pmb{b} \not\in \PREF^{\ast}_{F, |F|-1}.
\end{equation}

We state the key idea of the proof as follows.
By (\ref{eq:2cz6211gg810}), whenever the decoder reads a prefix $\pref(\pmb{d})$ of the codeword sequence, the decoder can know that the following bit is $d_l$ without reading it. 
Hence, the bit $d_l$ gives no information and is unnecessary for the $k$-bit delay decodability of the mapping $f^{\ast}_{|F|-1}$.
We consider obtaining another code-tuple $F'' \in \mathscr{F}_{\reg} \cap \mathscr{F}_{\ext} \cap \mathscr{F}_{k\hdec}$ such that $L(F'') < L(F)$ by removing this redundant bit, which leads to a contradiction to $F \in \mathscr{F}_{k\hopt}$ as desired.
However, naive removing a bit may impair the $k$-bit delay decodability of the other mappings $f^{\ast}_i$ for $i \in [|F|-1]$.
Accordingly, we first define a code-tuple $F'$ which is essentially equivalent to $F$ by adding some duplicates of the code tables to $F$.
Then by making changes to the replicated code tables instead of the original code tables, we obtain the desired $F''$ without affecting the $k$-bit delay decodability of $f^{\ast}_i$ for $i \in [|F|-1]$.

We define the code-tuple $F'$ as follows.
Put $L \coloneqq |F|(|\pmb{d}|+1)$ and $M \coloneqq |\mathcal{S}^{\leq L}|$.
We number all the sequences of $\mathcal{S}^{\leq L}$ as $\pmb{z}^{(0)}, \pmb{z}^{(1)}, \pmb{z}^{(2)}, \ldots, \pmb{z}^{(M-1)}$ in any order but $\pmb{z}^{(0)} \coloneqq \lambda$.
For $\pmb{z}' \in \mathcal{S}^{\leq L}$, we define $\langle \pmb{z}' \rangle \coloneqq |F|-1 + t$, where $t$ is the integer such that $\pmb{z}^{(t)} = \pmb{z}'$.
Note that $\langle \lambda \rangle = |F|-1$ since $\pmb{z}^{(0)} = \lambda$.
We define the code-tuple $F' \in \mathscr{F}^{(|F|-1+M)}$ consisting of $f'_0, f'_1, \ldots, f'_{|F|-1}, f'_{\langle \pmb{z}^{(1)}\rangle}, f'_{\langle \pmb{z}^{(2)}\rangle}, \ldots, f'_{\langle \pmb{z}^{(M-1)}\rangle}$ and 
$\trans'_0, \trans'_1, \ldots, \trans'_{|F|-1}, \trans'_{\langle \pmb{z}^{(1)}\rangle}, \trans'_{\langle \pmb{z}^{(2)}\rangle}, \allowbreak \ldots, \trans'_{\langle \pmb{z}^{(M-1)}\rangle}$ as
\begin{align}
f'_i(s) &=
\begin{cases}
f_{\trans^{\ast}_{\langle \lambda \rangle}(\pmb{z})}(s) &\,\,\text{if}\,\,i = \langle \pmb{z}\rangle \,\,\text{for some}\,\, \pmb{z} \in \mathcal{S}^{\leq L},\\
f_i(s) &\,\,\text{otherwise},\\
\end{cases} \label{eq:r3lzkx96owxn}\\
\trans'_i(s) &=
\begin{cases}
\langle \pmb{z}s \rangle &\,\,\text{if}\,\,i = \langle \pmb{z}\rangle \,\,\text{for some}\,\, \pmb{z} \in \mathcal{S}^{\leq L-1},\\
\trans^{\ast}_{\langle \lambda \rangle}(\pmb{z}s) &\,\,\text{if}\,\,i = \langle \pmb{z}\rangle \,\,\text{for some}\,\, \pmb{z} \in \mathcal{S}^{L},\\
\trans_i(s) &\,\,\text{otherwise}
\end{cases} \label{eq:je4ogqd03chp}
\end{align}
for $i \in [F']$ and $s \in \mathcal{S}$.
Then $F'$ satisfies the following Lemma \ref{lem:fdot}.
See Appendix \ref{subsec:proof-fdot} for the proof of Lemma \ref{lem:fdot}.

\begin{lemma}
\label{lem:fdot}
For any $\pmb{z} \in \mathcal{S}^{\leq L}$, the following statements (i) and (ii) hold.
\begin{itemize}
\item[(i)] $\trans'^{\ast}_{\langle \lambda \rangle}(\pmb{z}) = \langle \pmb{z} \rangle$.
\item[(ii)] $\langle \pmb{z} \rangle \in \kernel_{F'}$.
\end{itemize}
\end{lemma}

Lemma \ref{lem:fdot} (i) claims that the code table in $F'$ used next after encoding $\pmb{z} \in \mathcal{S}^{\leq L}$ starting from $f'_{\langle \lambda \rangle}$ is $f'_{\langle \pmb{z} \rangle}$, which is a duplicate of the code table in $F$ used next after encoding $\pmb{z}$ starting from $f_{\langle \lambda \rangle}$.
This leads to the equivalency of $F$ and $F'$ shown next.

We confirm that $F'$ is equivalent to $F$, that is, $F' \in \mathscr{F}_{\reg} \cap \mathscr{F}_{\ext} \cap \mathscr{F}_{k\hdec}$ and $L(F') = L(F)$.
We obtain $F' \in \mathscr{F}_{\reg}$ from Lemma \ref{lem:fdot} (ii) and Lemma \ref{lem:kernel} (i).
To prove $F' \in \mathscr{F}_{\ext} \cap \mathscr{F}_{k\hdec}$ and $L(F') = L(F)$ by using  Lemma \ref{lem:duplicate}, 
we show that a mapping $\varphi: [F'] \rightarrow[F]$ defined as the following (\ref{eq:gtet1tbkkabj}) is a homomorphism:
\begin{equation}
\label{eq:gtet1tbkkabj}
\varphi(i) = \begin{cases}
i &\,\,\text{if}\,\, i \in [F],\\  
\trans^{\ast}_{\langle \lambda \rangle}(\pmb{z}) &\,\,\text{if}\,\,i = \langle \pmb{z}\rangle \,\,\text{for some}\,\, \pmb{z} \in \mathcal{S}^{\leq L}\\
\end{cases}\\
\end{equation}
for $i \in [F']$.
The case $i = |F| - 1 = \langle \lambda \rangle$ applies to both of the first and second cases of (\ref{eq:gtet1tbkkabj}).
However, this case is consistent since $\trans^{\ast}_{\langle \lambda \rangle}(\pmb{z}) = \trans^{\ast}_{\langle \lambda \rangle}(\lambda) = \langle \lambda \rangle = i$.
We see that $\varphi$ satisfies (\ref{eq:phi1}) directly from (\ref{eq:r3lzkx96owxn}) and (\ref{eq:gtet1tbkkabj}).
We confirm that $\varphi$ satisfies also (\ref{eq:phi2}) as follows:
\begin{align*}
\varphi(\trans'_i(s))
&\overset{(\mathrm{A})}{=} \begin{cases}
\varphi(\langle \pmb{z}s \rangle) &\,\,\text{if}\,\,i = \langle \pmb{z}\rangle \,\,\text{for some}\,\, \pmb{z} \in \mathcal{S}^{\leq L-1},\\
\varphi(\trans^{\ast}_{\langle \lambda \rangle}(\pmb{z}s)) &\,\,\text{if}\,\,i = \langle \pmb{z}\rangle \,\,\text{for some}\,\, \pmb{z} \in \mathcal{S}^{L},\\
\varphi(\trans_i(s)) &\,\,\text{otherwise},\\
\end{cases}\\
&\overset{(\mathrm{B})}{=} \begin{cases}
\trans^{\ast}_{\langle \lambda \rangle}(\pmb{z}s) &\,\,\text{if}\,\,i = \langle \pmb{z}\rangle \,\,\text{for some}\,\, \pmb{z} \in \mathcal{S}^{\leq L-1},\\
\trans^{\ast}_{\langle \lambda \rangle}(\pmb{z}s) &\,\,\text{if}\,\,i = \langle \pmb{z}\rangle \,\,\text{for some}\,\, \pmb{z} \in \mathcal{S}^{L},\\
\trans_i(s) &\,\,\text{otherwise},\\
\end{cases}\\
&\overset{(\mathrm{C})}{=} \begin{cases}
\trans_{\trans^{\ast}_{\langle \lambda \rangle}(\pmb{z})}(s) &\,\,\text{if}\,\,i = \langle \pmb{z}\rangle \,\,\text{for some}\,\, \pmb{z} \in \mathcal{S}^{\leq L},\\
\trans_i(s) &\,\,\text{otherwise},\\
\end{cases}\\
&\overset{(\mathrm{D})}{=} \trans_{\varphi(i)}(s),
\end{align*}
where
(A) follows from (\ref{eq:je4ogqd03chp}),
(B) follows from (\ref{eq:gtet1tbkkabj}),
(C) follows from Lemma \ref{lem:f_T} (ii),
and (D) follows from (\ref{eq:gtet1tbkkabj}).
Hence, by Lemma \ref{lem:duplicate} (iv)--(vi), we obtain $F' \in \mathscr{F}_{\ext} \cap \mathscr{F}_{k\hdec}$ and $L(F') = L(F)$.

Now, we define a code-tuple $F'' \in \mathscr{F}^{(|F'|)}$ as
\begin{equation}
\label{eq:fddot}
f''_{i}(s) =
\begin{cases}
f'^{\ast}_{\langle \lambda \rangle}(\pmb{z})^{-1}\pref(\pmb{d})\pmb{d}^{-1}(f'^{\ast}_{\langle \lambda \rangle}(\pmb{z}s))\\
\quad \,\,\text{if}\,\,i = \langle \pmb{z}\rangle \,\,\text{and}\,\, f'^{\ast}_{\langle \lambda \rangle}(\pmb{z}) \prec \pmb{d} \preceq f'^{\ast}_{\langle \lambda \rangle}(\pmb{z}s) \quad \,\,\text{for some}\,\, \pmb{z} \in \mathcal{S}^{\leq L},\\
f'_i(s) \quad\quad\quad\quad \,\,\text{otherwise},
\end{cases}
\end{equation}
\begin{equation}
\label{eq:g0yqx2cwxfd1}
\trans''_i(s) = \trans'_i(s)
\end{equation}
for $i \in [F'']$ and $s \in \mathcal{S}$.

Intuitively, (\ref{eq:fddot}) means that $F''$ is obtained by removing the bit $d_l$ from codeword sequences of $F'$ such that $f'^{\ast}_{\langle \lambda \rangle}(\pmb{z}) \succeq \pmb{d}$.

Then $F''$ satisfies the following Lemma \ref{lem:fddot}.
See Appendix \ref{subsec:proof-fddot} for the proof of Lemma \ref{lem:fddot}.

\begin{lemma}
The following statements (i)--(iii) hold.
\label{lem:fddot}
\begin{itemize}
\item[(i)] For any $\pmb{z} \in \mathcal{S}^{\leq L}$ and $\pmb{x} \in \mathcal{S}^{\leq L-|\pmb{z}|}$,
we have
\begin{equation}
\label{eq:i5igfy04wlhe}
f''^{\ast}_{\langle \pmb{z} \rangle}(\pmb{x}) =
\begin{cases}
f'^{\ast}_{\langle \lambda \rangle}(\pmb{z})^{-1}\pref(\pmb{d})\pmb{d}^{-1}(f'^{\ast}_{\langle \lambda \rangle}(\pmb{z}\pmb{x}))
\quad \,\,\text{if}\,\, f'^{\ast}_{\langle \lambda \rangle}(\pmb{z}) \prec \pmb{d} \preceq f'^{\ast}_{\langle \lambda \rangle}(\pmb{z}\pmb{x}),\\
f'^{\ast}_{\langle \pmb{z} \rangle}(\pmb{x}) \quad\quad\quad\quad \,\,\text{otherwise}.\\
\end{cases}
\end{equation}

\item[(ii)] For any $\pmb{z} \in \mathcal{S}^{\leq L}$ and $s, s' \in \mathcal{S}$, if $f''_{\langle \pmb{z} \rangle}(s) \prec f''_{\langle \pmb{z} \rangle}(s')$, then $f'_{\langle \pmb{z} \rangle}(s) \prec f'_{\langle \pmb{z} \rangle}(s')$.

\item[(iii)] For any $\pmb{x} \in \mathcal{S}^{\geq L}$, we have
$|f^{\ast}_{\langle \lambda \rangle}(\pmb{x})| = |f'^{\ast}_{\langle \lambda \rangle}(\pmb{x})| \geq |\pmb{d}|+1$ and $|f''^{\ast}_{\langle \lambda \rangle}(\pmb{x})| \geq |\pmb{d}|$.
\end{itemize}
\end{lemma}

We show that $F'' \in \mathscr{F}_{\reg} \cap \mathscr{F}_{\ext} \cap \mathscr{F}_{k\hdec}$ and $L(F'') < L(F')$ ($= L(F)$ as shown above),
which conflicts with $F \in \mathscr{F}_{k\hopt}$ and completes the proof of Theorem \ref{thm:complete}.

(Proof of $F'' \in \mathscr{F}_{\reg}$): From $F' \in \mathscr{F}_{\reg}$ and (\ref{eq:g0yqx2cwxfd1}).

(Proof of $F'' \in \mathscr{F}_{\ext}$):
Choose $j \in [F'']$ arbitrarily.
Since $\langle \lambda \rangle \in \kernel_{F'} = \kernel_{F''}$ by Lemma \ref{lem:fdot} (ii) and (\ref{eq:g0yqx2cwxfd1}), there exists $\pmb{x} \in \mathcal{S}^{\ast}$ such that
\begin{equation}
\label{eq:98q5npcu7fk8}
\trans''^{\ast}_j(\pmb{x}) = \langle \lambda \rangle.
\end{equation}
Also, we can choose $\pmb{x}' \in \mathcal{S}^{L}$ such that
\begin{equation}
\label{eq:wzkt55cvmj5z}
f'^{\ast}_{\langle \lambda \rangle}(\pmb{x}') \succeq \pmb{d}
\end{equation}
by Lemma \ref{lem:fddot} (iii).
We have
\begin{align*}
|f''^{\ast}_j(\pmb{x}\pmb{x}')|
&\overset{(\mathrm{A})}{=} |f''^{\ast}_j(\pmb{x})| + |f''^{\ast}_{\trans''^{\ast}_j(\pmb{x})}(\pmb{x}')|\\
&\geq |f''^{\ast}_{\trans''^{\ast}_j(\pmb{x})}(\pmb{x}')|\\
&\overset{(\mathrm{B})}{=} |f''^{\ast}_{\langle \lambda \rangle}(\pmb{x}')|\\
&\overset{(\mathrm{C})}{=} |f'^{\ast}_{\langle \lambda \rangle}(\lambda)^{-1}\pref(\pmb{d})\pmb{d}^{-1}f'^{\ast}_{\langle \lambda \rangle}(\pmb{x}')|\\
&= |f'^{\ast}_{\langle \lambda \rangle}(\pmb{x}')|-1\\
&\overset{(\mathrm{D})}{\geq} |\pmb{d}|\\
&\geq 1,
\end{align*}
where
(A) follows from Lemma \ref{lem:f_T} (i),
(B) follows from (\ref{eq:98q5npcu7fk8}),
(C) follows from (\ref{eq:wzkt55cvmj5z}) and the first case of (\ref{eq:i5igfy04wlhe}),
and (D) follows from Lemma \ref{lem:fddot} (iii).
Hence, by (\ref{eq:pref3}), $\PREF^1_{F'', j} \neq \emptyset$ holds for any $j \in [F'']$, which leads to $F'' \in \mathscr{F}_{\ext}$ as desired.

(Proof of $L(F'') < L(F')$):
For any $i \in [F'']$ and $s \in \mathcal{S}$, we have $|f''_i(s)| \leq |f'_i(s)|$ by (\ref{eq:fddot}).
Hence, for any $i \in [F'']$, we have
\begin{equation}
\label{eq:x6x5hn8zqhk2}
\pi_i(F')L_i(F'') \leq \pi_i(F')L_i(F').
\end{equation}

By Lemma \ref{lem:fddot} (iii), we can choose $\pmb{x} = x_1x_2\ldots x_L \in \mathcal{S}^{L}$ such that $f'^{\ast}_{\langle \lambda \rangle}(\pmb{x}) \succeq \pmb{d}$.
Since $f'^{\ast}_{\langle \lambda \rangle}(\lambda) \prec \pmb{d} \preceq f'^{\ast}_{\langle \lambda \rangle}(\pmb{x})$,
there exists exactly one integer $r$ such that
\begin{equation}
\label{eq:3enm1lws9tq0}
f'^{\ast}_{\langle \lambda \rangle}(x_1x_2\ldots x_{r-1}) \prec \pmb{d} \preceq f'^{\ast}_{\langle \lambda \rangle}(x_1x_2\ldots x_r),
\end{equation}
which leads to
\begin{align}
|f''_{\langle \pmb{z}\rangle}(x_r)| 
&\overset{(\mathrm{A})}{=} |f'^{\ast}_{\langle \lambda \rangle}(\pmb{z})^{-1}\pref(\pmb{d})\pmb{d}^{-1}(f'^{\ast}_{\langle \lambda \rangle}(\pmb{z}x_r))|
= |f'_{\langle \pmb{z}\rangle}(x_r)| - 1
< |f'_{\langle \pmb{z} \rangle}(x_r)|, \label{eq:elk2jzqnt3mr}
\end{align}
where $\pmb{z}\coloneqq x_1x_2\ldots x_{r-1}$, and
(A) follows from (\ref{eq:3enm1lws9tq0}) and the first case of (\ref{eq:fddot}).
This leads to
\begin{equation}
\label{eq:0vstqkzkfhsg}
\pi_{\langle \pmb{z} \rangle}(F') L_{\langle \pmb{z} \rangle}(F'') < \pi_{\langle \pmb{z} \rangle}(F') L_{\langle \pmb{z} \rangle}(F')
\end{equation}
because $\pi_{\langle \pmb{z} \rangle}(F') > 0$ by Lemma \ref{lem:fdot} (ii) and Lemma \ref{lem:kernel} (ii).

Hence, we have
\begin{align*}
L(F'')
&= \sum_{i \in [F'']} \pi_i(F'') L_i(F'')\\
&\overset{(\mathrm{A})}{=} \sum_{i \in [F'']} \pi_i(F') L_i(F'')\\
&= \sum_{i \in [F'']\setminus \{\langle \pmb{z} \rangle\}} \pi_i(F') L_i(F'') + \pi_{\langle \pmb{z} \rangle}(F') L_{\langle \pmb{z} \rangle}(F'')\\
&\overset{(\mathrm{B})}{\leq} \sum_{i \in [F'']\setminus \{\langle \pmb{z} \rangle\}} \pi_i(F') L_i(F') + \pi_{\langle \pmb{z} \rangle}(F') L_{\langle \pmb{z} \rangle}(F'')\\
&\overset{(\mathrm{C})}{<} \sum_{i \in [F'']\setminus \{\langle \pmb{z} \rangle\}} \pi_i(F') L_i(F') + \pi_{\langle \pmb{z} \rangle}(F') L_{\langle \pmb{z} \rangle}(F')\\
&= \sum_{i \in [F']} \pi_i(F') L_i(F')\\
&= L(F')
\end{align*}
as desired,
where 
(A) follows from (\ref{eq:g0yqx2cwxfd1}),
(B) follows from (\ref{eq:x6x5hn8zqhk2}),
and (C) follows from (\ref{eq:0vstqkzkfhsg}).

(Proof of $F'' \in \mathscr{F}_{k\hdec}$):
To prove $F'' \in \mathscr{F}_{k\hdec}$, we use the following Lemma \ref{lem:fddot2}, where 
$\mathcal{J} \coloneqq  ([F'] \setminus \langle \lambda \rangle) \cup \{\langle \pmb{z} \rangle : \pmb{z} \in \mathcal{S}^L\}
= [F'] \setminus \{\langle \pmb{z} \rangle : \pmb{z} \in \mathcal{S}^{\leq L-1}\}$.
See Appendix \ref{subsec:proof-fddot2} for the proof of Lemma \ref{lem:fddot2}.

\begin{lemma}
The following statements (i)--(iii) hold.
\label{lem:fddot2}
\begin{itemize}
\item[(i)] For any $\pmb{x} \in \mathcal{S}^{\ast}$ and $\pmb{c} \in \mathcal{C}^{\leq k}$, if $f''^{\ast}_{\langle \lambda \rangle}(\pmb{x}) \succeq \pmb{c}$, then $f'^{\ast}_{\langle \lambda \rangle}(\pmb{x}) \succeq \pmb{c}$.
Therefore, we have $\PREF^k_{F', \langle \lambda \rangle} \supseteq \PREF^k_{F'', \langle \lambda \rangle}$ by (\ref{eq:pref3}).
\item[(ii)] For any $i \in \mathcal{J}$ and $s \in \mathcal{S}$, we have $f''_i(s) = f'_i(s)$.
\item[(iii)] For any $i \in \mathcal{J}$ and $\pmb{b} \in \mathcal{C}^{\ast}$, we have $\PREF^k_{F'', i}(\pmb{b}) \subseteq \PREF^k_{F', i}(\pmb{b})$ and $\bar{\PREF}^k_{F'', i}(\pmb{b}) \subseteq \bar{\PREF}^k_{F', i}(\pmb{b})$.
\end{itemize}
\end{lemma}

Also, for $\pmb{z} \in \mathcal{S}^{\ast}$, we define a mapping $\psi_{\pmb{z}} : \mathcal{C}^{\ast} \rightarrow \mathcal{C}^{\ast}$ as
\begin{equation}
\label{eq:psi-1}
\psi_{\pmb{z}}(\pmb{b}) = 
\begin{cases}
f'^{\ast}_{\langle \lambda \rangle}(\pmb{z})^{-1}\pmb{d}\pref(\pmb{d})^{-1}(f'^{\ast}_{\langle \lambda \rangle}(\pmb{z})\pmb{b})
\quad \,\,\text{if}\,\, f'^{\ast}_{\langle \lambda \rangle}(\pmb{z}) \preceq \pref(\pmb{d}) \prec f'^{\ast}_{\langle \lambda \rangle}(\pmb{z})\pmb{b},\\
\pmb{b} \quad\quad\quad\quad \,\,\text{otherwise}\\
\end{cases}
\end{equation}
for $\pmb{b} \in \mathcal{C}^{\ast}$.
Then $\psi_{\pmb{z}}$ satisfies the following Lemma \ref{lem:psi-1}.

\begin{lemma}
The following statements (i)--(iii) hold.
\label{lem:psi-1}
\begin{enumerate}[(i)]
\item For any $\pmb{z} \in \mathcal{S}^{\ast}$ and $\pmb{b}, \pmb{b}' \in \mathcal{C}^{\ast}$, if $\pmb{b} \preceq \pmb{b}'$, then $\psi_{\pmb{z}}(\pmb{b}) \preceq \psi_{\pmb{z}}(\pmb{b}')$.
\item For any $\pmb{z} \in \mathcal{S}^{\leq L}$, $\pmb{x} \in \mathcal{S}^{\leq L - |\pmb{z}|}$, and $\pmb{c} \in \mathcal{C}^{\ast}$, we have
\begin{align}
\psi_{\pmb{z}}(f''^{\ast}_{\langle \pmb{z} \rangle} (\pmb{x}) \pmb{c})
&= \begin{cases}
\pref(f'^{\ast}_{\langle \pmb{z} \rangle} (\pmb{x}))
\quad \,\,\text{if}\,\, f'^{\ast}_{\langle \lambda \rangle}(\pmb{z}) \prec f'^{\ast}_{\langle \lambda \rangle}(\pmb{zx}) = \pmb{d}, \pmb{c} = \lambda,\\
f'^{\ast}_{\langle \pmb{z} \rangle} (\pmb{x}) \psi_{\pmb{zx}}(\pmb{c})  \,\,\text{otherwise}.\\
\end{cases} \label{eq:gebqpcvgz9hc}
\end{align}
\item For any $\pmb{z} \in \mathcal{S}^{L}$ and $\pmb{b} \in \mathcal{C}^{\ast}$, we have $\psi_{\pmb{z}}(\pmb{b}) = \pmb{b}$.
\end{enumerate}
\end{lemma}

See Appendix \ref{subsec:proof-psi-1} for the proof of Lemma \ref{lem:psi-1}.

By Lemma \ref{lem:psi-1} (ii) with $\pmb{c} = \lambda$, it holds that $\psi_{\pmb{z}}(f''^{\ast}_{\langle \pmb{z} \rangle} (\pmb{x})) = f'^{\ast}_{\langle \pmb{z} \rangle} (\pmb{x})$ in most cases.
Thus, we can intuitively interpret the mapping $\psi_{\pmb{z}}$ as a kind of an inverse transformation of (\ref{eq:i5igfy04wlhe}).
We prove $k$-bit delay decodability of $F''$ later by attributing it to $k$-bit delay decodability of $F'$ using $\psi_{\pmb{z}}$.

Now we prove $F'' \in \mathscr{F}_{k\hdec}$.
We first show that $F''$ satisfies Definition \ref{def:k-bitdelay} (i).
Namely, we show that $\PREF^k_{F'', \trans'_i(s)} \cap \bar{\PREF}^k_{F'', i}(f''_i(s)) = \emptyset$ for any $i \in [F'']$ and $s \in \mathcal{S}$ dividing into the following two cases: the case $i \in \mathcal{J}$ and the case $i \in [F''] \setminus \mathcal{J}$.

\begin{itemize}
\item The case $i \in \mathcal{J}$:
Then for any $i \in \mathcal{J}$ and $s \in \mathcal{S}$, we have
\begin{align*}
\PREF^k_{F'', \trans''_i(s)} \cap \bar{\PREF}^k_{F'', i}(f''_i(s))
&\overset{(\mathrm{A})}{\subseteq} \PREF^k_{F', \trans''_i(s)} \cap \bar{\PREF}^k_{F', i}(f''_i(s))\\
&\overset{(\mathrm{B})}{=} \PREF^k_{F', \trans'_i(s)} \cap \bar{\PREF}^k_{F', i}(f'_i(s))\\
&\overset{(\mathrm{C})}{=} \emptyset,
\end{align*}
where
(A) follows from Lemma \ref{lem:fddot2} (i) (iii) since $\trans''_i(s) \in [F]$,
(B) follows from Lemma \ref{lem:fddot2} (ii) and (\ref{eq:g0yqx2cwxfd1}),
and (C) follows from $F' \in \mathscr{F}_{k\hdec}$.

\item The case $i \in [F''] \setminus \mathcal{J}$:
We prove by contradiction assuming that the exist $\pmb{z} \in \mathcal{S}^{\leq L-1}$, $s \in \mathcal{S}$, and $\pmb{c} \in \bar{\PREF}^k_{F'', \langle\pmb{z}\rangle}(f''_{\langle\pmb{z}\rangle}(s)) \cap \PREF^k_{F'', \langle\pmb{z}s\rangle}$.
By $\pmb{c} \in \bar{\PREF}^k_{F'', \langle\pmb{z}\rangle}(f''_{\langle\pmb{z}\rangle}(s))$ and (\ref{eq:pref2}), there exist $\pmb{x} \in \mathcal{S}^{L-|\pmb{z}|}$ and $\pmb{y} \in \mathcal{S}^{\ast}$ such that
\begin{equation}
\label{eq:o6q4wim5xzak}
f''^{\ast}_{\langle \pmb{z} \rangle}(\pmb{x}\pmb{y}) \succeq f''_{\langle \pmb{z} \rangle}(s) \pmb{c}
\end{equation}
and
\begin{equation}
\label{eq:cd94d67co5is}
f''_{\langle \pmb{z} \rangle}(x_1) \succ f''_{\langle \pmb{z} \rangle}(s).
\end{equation}
By Lemma \ref{lem:F_ext}, we may assume
\begin{equation}
\label{eq:hznwfkjgtb3c}
|f''^{\ast}_{\langle \pmb{zx} \rangle}(\pmb{y})| \geq \max\{k, 1\}.
\end{equation}
By (\ref{eq:cd94d67co5is}) and Lemma \ref{lem:fddot} (ii), we obtain
\begin{equation}
\label{eq:97u9hprr3k3s}
f'_{\langle \pmb{z} \rangle}(x_1) \succ f'_{\langle \pmb{z} \rangle}(s).
\end{equation}
This shows that $f'_{\langle \pmb{z} \rangle}$ is not prefix-free, which conflicts with $F' \in \mathscr{F}_{k\hdec}$ in the case $k = 0$ by Lemma \ref{lem:0dec-prefixfree}.
Thus, we consider the case $k \geq 1$, that is,
\begin{equation}
\label{eq:ge6refzynw83}
\pmb{c} \neq \lambda.
\end{equation}
Equation (\ref{eq:o6q4wim5xzak}) leads to
\begin{align}
f''^{\ast}_{\langle \pmb{z} \rangle}(\pmb{x}\pmb{y}) \succeq f''_{\langle \pmb{z} \rangle}(s) \pmb{c}
&\overset{(\mathrm{A})}{\implies} \psi_{\pmb{z}}(f''^{\ast}_{\langle \pmb{z} \rangle}(\pmb{x}\pmb{y}))
\succeq \psi_{\pmb{z}}(f''_{\langle \pmb{z} \rangle}(s) \pmb{c}) \nonumber\\
&\overset{(\mathrm{B})}{\iff} \psi_{\pmb{z}}(f''^{\ast}_{\langle \pmb{z} \rangle}(\pmb{x}) f''^{\ast}_{\langle \pmb{z}\pmb{x} \rangle}(\pmb{y}))
\succeq \psi_{\pmb{z}}(f''_{\langle \pmb{z} \rangle}(s) \pmb{c}) \nonumber\\
&\overset{(\mathrm{C})}{\iff} f'^{\ast}_{\langle \pmb{z} \rangle}(\pmb{x}) \psi_{\pmb{z}\pmb{x}}(f''^{\ast}_{\langle \pmb{z}\pmb{x} \rangle}(\pmb{y}))
\succeq f'_{\langle \pmb{z} \rangle}(s)\psi_{\pmb{z}s}(\pmb{c}) \nonumber\\
&\overset{(\mathrm{D})}{\iff} f'^{\ast}_{\langle \pmb{z} \rangle}(\pmb{x}) f''^{\ast}_{\langle \pmb{z}\pmb{x} \rangle}(\pmb{y})
\succeq f'_{\langle \pmb{z} \rangle}(s)\psi_{\pmb{z}s}(\pmb{c}) \label{eq:ub3q36542x96},
\end{align}
where
(A) follows from Lemma \ref{lem:psi-1} (i),
(B) follows from Lemma \ref{lem:f_T} (i) and Lemma \ref{lem:fdot} (i),
(C) follows from (\ref{eq:hznwfkjgtb3c}), (\ref{eq:ge6refzynw83}), and the second case of (\ref{eq:gebqpcvgz9hc}),
and (D) follows from Lemma \ref{lem:psi-1} (iii) and $|\pmb{zx}| = L$.

Now, for $\pmb{b} \in \mathcal{C}^{\geq k}$, let $[\pmb{b}]_k$ denote the prefix of length $k$ of $\pmb{b}$.
Then by (\ref{eq:97u9hprr3k3s}) and (\ref{eq:ub3q36542x96}), we have
\begin{equation}
\label{eq:0hab0zqds0ty}
f'^{\ast}_{\langle \pmb{z} \rangle}(\pmb{x}) [f''^{\ast}_{\langle \pmb{z}\pmb{x} \rangle}(\pmb{y})]_k
\succeq f'_{\langle \pmb{z} \rangle}(s)[\psi_{\pmb{z}s}(\pmb{c})]_k.
\end{equation}
Also, we have
\begin{equation}
[f''^{\ast}_{\langle \pmb{z}\pmb{x} \rangle}(\pmb{y})]_k \in \PREF^k_{F'', \langle \pmb{zx} \rangle} \overset{(\mathrm{A})}{\subseteq} \PREF^k_{F', \langle \pmb{zx} \rangle},
\end{equation}
where (A) follows from Lemma \ref{lem:fddot2} (iii) and $\langle \pmb{zx} \rangle \in \mathcal{S}^L \subseteq \mathcal{J}$.
Hence, by (\ref{eq:pref3}) there exists $\pmb{y}' \in \mathcal{S}^{\ast}$ such that
$f'^{\ast}_{\langle \pmb{z}\pmb{x} \rangle}(\pmb{y}') \succeq [f''^{\ast}_{\langle \pmb{z}\pmb{x} \rangle}(\pmb{y})]_k$, which leads to 
\begin{align}
\label{eq:vihwlnbhtk4n}
f'^{\ast}_{\langle \pmb{z} \rangle}(\pmb{x}\pmb{y}')
&= f'^{\ast}_{\langle \pmb{z} \rangle}(\pmb{x}) f'^{\ast}_{\langle \pmb{z}\pmb{x} \rangle}(\pmb{y}')
\succeq f'^{\ast}_{\langle \pmb{z} \rangle}(\pmb{x}) [f''^{\ast}_{\langle \pmb{z}\pmb{x} \rangle}(\pmb{y})]_k 
\overset{(\mathrm{A})}{\succeq} f'_{\langle \pmb{z} \rangle}(s) [\psi_{\pmb{z}s}(\pmb{c})]_k,
\end{align}
where (A) follows from (\ref{eq:0hab0zqds0ty}).
Equations (\ref{eq:97u9hprr3k3s}) and (\ref{eq:vihwlnbhtk4n}) show
\begin{equation}
\label{eq:xkhgxisgau0l}
[\psi_{\pmb{z}s}(\pmb{c})]_k \in \bar{\PREF}^k_{F', \langle \pmb{z} \rangle}(f'_{\langle \pmb{z} \rangle}(s))
\end{equation}
by (\ref{eq:pref2}).

On the other hand, by $\pmb{c} \in \PREF^k_{F'', \langle\pmb{z}s\rangle}$ and (\ref{eq:pref3}), there exist $\pmb{x} \in \mathcal{S}^{L-|\pmb{z}s|}$ and $\pmb{y} \in \mathcal{S}^{\ast}$ such that
\begin{equation}
\label{eq:rtf3glg7thoe}
f''^{\ast}_{\langle \pmb{z}s \rangle}(\pmb{x}\pmb{y}) \succeq \pmb{c}.
\end{equation}
By Lemma \ref{lem:F_ext}, we may assume
\begin{equation}
\label{eq:kdnf706jdq2d}
|f''^{\ast}_{\langle \pmb{z}s\pmb{x} \rangle}(\pmb{y})| \geq k \geq 1.
\end{equation}
We have
\begin{align*}
f'^{\ast}_{\langle \pmb{z}s \rangle}(\pmb{x}) f''^{\ast}_{\langle \pmb{z}s\pmb{x} \rangle}(\pmb{y})
&\overset{(\mathrm{A})}{=} f'^{\ast}_{\langle \pmb{z}s \rangle}(\pmb{x}) \psi_{\pmb{z}s\pmb{x}}(f''^{\ast}_{\langle \pmb{z}s\pmb{x} \rangle}(\pmb{y}))
\overset{(\mathrm{B})}{=}  \psi_{\pmb{z}s}(f''^{\ast}_{\langle \pmb{z}s \rangle}(\pmb{x}\pmb{y}))
\overset{(\mathrm{C})}{\succeq} \psi_{\pmb{z}s}(\pmb{c}),
\end{align*}
where
(A) follows from Lemma \ref{lem:psi-1} (iii) and $|\pmb{z}s\pmb{x}| = L$,
(B) follows from  (\ref{eq:kdnf706jdq2d}) and the second case of (\ref{eq:gebqpcvgz9hc}),
and (C) follows from (\ref{eq:rtf3glg7thoe}) and Lemma \ref{lem:psi-1} (i).

Hence, we have
\begin{equation}
\label{eq:8zjo9028udgh}
f'^{\ast}_{\langle \pmb{z}s \rangle}(\pmb{x}) [f''^{\ast}_{\langle \pmb{z}s\pmb{x} \rangle}(\pmb{y})]_k
\succeq [\psi_{\pmb{z}s}(\pmb{c})]_k.
\end{equation}
Also, we have
\begin{equation}
[f''^{\ast}_{\langle \pmb{z}s\pmb{x} \rangle}(\pmb{y})]_k \in \PREF^k_{F'', \langle \pmb{z}s\pmb{x} \rangle} \overset{(\mathrm{A})}{\subseteq} \PREF^k_{F', \langle \pmb{z}s\pmb{x} \rangle},
\end{equation}
where (A) follows from Lemma \ref{lem:fddot2} (iii) and $\langle \pmb{z}s\pmb{x} \rangle \in \mathcal{S}^L \subseteq \mathcal{J}$.
Hence, there exists $\pmb{y}' \in \mathcal{S}^{\ast}$ such that
$f'^{\ast}_{\langle \pmb{z}\pmb{x} \rangle}(\pmb{y}') \succeq [f''^{\ast}_{\langle \pmb{z}s\pmb{x} \rangle}(\pmb{y})]_k$,
which leads to 
\begin{align}
\label{eq:5nw09tdlaulb}
f'^{\ast}_{\langle \pmb{z}s \rangle}(\pmb{x}\pmb{y}')
&= f'^{\ast}_{\langle \pmb{z}s \rangle}(\pmb{x}) f'^{\ast}_{\langle \pmb{z}\pmb{x}s \rangle}(\pmb{y}')
\succeq f'^{\ast}_{\langle \pmb{z}s \rangle}(\pmb{x}) [f''^{\ast}_{\langle \pmb{z}s\pmb{x} \rangle}(\pmb{y})]_k 
\overset{(\mathrm{A})}{\succeq} [\psi_{\pmb{z}s}(\pmb{c})]_k,
\end{align}
where (A) follows from (\ref{eq:8zjo9028udgh}).
This shows
\begin{equation}
\label{eq:hu05xayiw2ed}
[\psi_{\pmb{z}s}(\pmb{c})]_k \in \PREF^k_{F', \langle \pmb{z}s \rangle}
\end{equation}
by (\ref{eq:pref3}).
By (\ref{eq:xkhgxisgau0l}) and (\ref{eq:hu05xayiw2ed}),
the code-tuple $F'$ does not satisfy Definition \ref{def:k-bitdelay} (i), which conflicts with $F' \in \mathscr{F}_{k\hdec}$.
\end{itemize}
Consequently, $F''$ satisfies Definition \ref{def:k-bitdelay} (i).

Next, we show that $F''$ satisfies Definition \ref{def:k-bitdelay} (ii).
Namely, we show that for any $i \in [F'']$ and $s, s' \in \mathcal{S}$ such that $s \neq s'$ and $f''_i(s) = f''_i(s')$, we have $\PREF^k_{F', \trans'_i(s)} \cap \PREF^k_{F', \trans'_i(s')} = \emptyset$.
We prove for the following two cases: the case $i \in \mathcal{J}$ and the case $i \in [F''] \setminus \mathcal{J}$.

\begin{itemize}
\item The case $i \in \mathcal{J}$:
Then for any $i \in \mathcal{J}$ and $s, s' \in \mathcal{S}$ such that $s \neq s'$ and $f''_i(s) = f''_i(s')$, we have
\begin{equation}
\label{eq:rx2efdxu2i5d}
f'_i(s) = f'_i(s')
\end{equation}
by Lemma \ref{lem:fddot2} (ii), and we have
\begin{align*}
\PREF^k_{F'', \trans''_i(s)} \cap \PREF^k_{F'', \trans''_i(s')}
&\overset{(\mathrm{A})}{\subseteq} \PREF^k_{F', \trans''_i(s)} \cap \PREF^k_{F', \trans''_i(s')}
\overset{(\mathrm{B})}{=} \PREF^k_{F', \trans'_i(s)} \cap \PREF^k_{F', \trans'_i(s')} \overset{(\mathrm{C})}{=} \emptyset,
\end{align*}
where
(A) follows from Lemma \ref{lem:fddot2} (i) (iii) since $\trans''_i(s), \allowbreak \trans''_i(s') \in [F]$,
(B) follows from (\ref{eq:g0yqx2cwxfd1}),
and (C) follows from $F' \in \mathscr{F}_{k\hdec}$ and (\ref{eq:rx2efdxu2i5d}).

\item The case $i \in [F''] \setminus \mathcal{J}$:
We prove by contradiction assuming that there exists $\pmb{z} \in \mathcal{S}^{\leq L-1}, s, s' \in \mathcal{S}$, and $\pmb{c} \in \PREF^k_{F'', \langle\pmb{z}s\rangle} \cap \PREF^k_{F'', \langle\pmb{z}s'\rangle}$ such that $s \neq s'$ and
\begin{equation}
\label{eq:iofwmooucr64}
f''_{\langle \pmb{z} \rangle}(s) = f''_{\langle \pmb{z} \rangle}(s').
\end{equation}
By the similar way to derive (\ref{eq:hu05xayiw2ed}), we obtain
\begin{equation}
\label{eq:ifte7279emhl}
[\psi_{\pmb{z}s}(\pmb{c})]_k \in \PREF^k_{F', \langle\pmb{z}s\rangle}
\end{equation}
 from $\pmb{c} \in \PREF^k_{F'', \langle\pmb{z}s\rangle}$.
By (\ref{eq:iofwmooucr64}) and Lemma \ref{lem:psi-1} (i),
we have
\begin{equation}
\label{eq:zq7ajj39x777}
\psi_{\langle \pmb{z} \rangle}(f''_{\langle \pmb{z} \rangle}(s)) = \psi_{\langle \pmb{z} \rangle}(f''_{\langle \pmb{z} \rangle}(s')).
\end{equation}
By Lemma \ref{lem:psi-1} (ii), exactly one of $f'_{\langle \pmb{z} \rangle}(s) = f'_{\langle \pmb{z} \rangle}(s')$, $f'_{\langle \pmb{z} \rangle}(s) \prec f'_{\langle \pmb{z} \rangle}(s')$, and $f'_{\langle \pmb{z} \rangle}(s) \succ f'_{\langle \pmb{z} \rangle}(s')$ holds.
Therefore, $f'_{\langle \pmb{z} \rangle}$ is not prefix-free, which conflicts with $F' \in \mathscr{F}_{k\hdec}$ in the case $k = 0$ by Lemma \ref{lem:0dec-prefixfree}.
We consider the case $k \geq 1$, that is,
\begin{equation}
\label{eq:capkno0i5kee}
\pmb{c} \neq \lambda.
\end{equation}
We consider the following two cases separately: the case $f'_{\langle \pmb{z} \rangle}(s) = f'_{\langle \pmb{z} \rangle}(s')$ and the case $f'_{\langle \pmb{z} \rangle}(s) \prec f'_{\langle \pmb{z} \rangle}(s')$.
Note that we may exclude the case $f'_{\langle \pmb{z} \rangle}(s) \succ f'_{\langle \pmb{z} \rangle}(s')$ by symmetry.

\begin{itemize}
\item The case $f'_{\langle \pmb{z} \rangle}(s) = f'_{\langle \pmb{z} \rangle}(s')$:
By (\ref{eq:psi-1}), we have $\psi_{\pmb{z}s}(\pmb{c}) = \psi_{\pmb{z}s'}(\pmb{c})$ and thus
\begin{equation}
\label{eq:m2hj94a4b003}
[\psi_{\pmb{z}s}(\pmb{c})]_k = [\psi_{\pmb{z}s'}(\pmb{c})]_k \overset{(\mathrm{A})}{\in} \PREF^k_{F', \langle\pmb{z}s'\rangle},
\end{equation}
where (A) is obtained from $\pmb{c} \in \PREF^k_{F'', \langle\pmb{z}s'\rangle}$ by the similar way to derive (\ref{eq:hu05xayiw2ed}).

By (\ref{eq:ifte7279emhl}), (\ref{eq:m2hj94a4b003}), and $f'_{\langle \pmb{z} \rangle}(s) = f'_{\langle \pmb{z} \rangle}(s')$,
the code-tuple $F'$ does not satisfy Definition \ref{def:k-bitdelay} (ii), which conflicts with $F' \in \mathscr{F}_{k\hdec}$.

\item The case $f'_{\langle \pmb{z} \rangle}(s) \prec f'_{\langle \pmb{z} \rangle}(s')$:
Then by (\ref{eq:zq7ajj39x777}) and Lemma \ref{lem:psi-1} (ii),
it must hold that
\begin{equation}
\label{eq:df0g61poqdfr}
f'^{\ast}_{\langle \lambda \rangle}(\pmb{z}) \prec f'^{\ast}_{\langle \lambda \rangle}(\pmb{z}s') = \pmb{d}
\end{equation}
and
\begin{equation}
\label{eq:tv9pahtdqyls}
f'_{\langle \pmb{z} \rangle}(s) = \pref(f'_{\langle \pmb{z} \rangle}(s')).
\end{equation}
Thus, we have
\begin{align}
f'_{\langle \pmb{z} \rangle}(s)d_l
&\overset{(\mathrm{A})}{=} \pref(f'_{\langle \pmb{z} \rangle}(s'))d_l \nonumber\\
&= f'^{\ast}_{\langle \lambda \rangle}(\pmb{z})^{-1} f'^{\ast}_{\langle \lambda \rangle}(\pmb{z}) \pref(f'_{\langle \pmb{z} \rangle}(s'))d_l\nonumber\\
&\overset{(\mathrm{B})}{=} f'^{\ast}_{\langle \lambda \rangle}(\pmb{z})^{-1} \pref(f'^{\ast}_{\langle \lambda \rangle}(\pmb{z}s'))d_l\nonumber\\
&\overset{(\mathrm{C})}{=} f'^{\ast}_{\langle \lambda \rangle}(\pmb{z})^{-1} \pref(\pmb{d})d_l\nonumber\\
&= f'^{\ast}_{\langle \lambda \rangle}(\pmb{z})^{-1} \pmb{d}\nonumber\\
&\overset{(\mathrm{D})}{=} f'^{\ast}_{\langle \lambda \rangle}(\pmb{z})^{-1} f'^{\ast}_{\langle \lambda \rangle}(\pmb{z}s')\nonumber\\
&\overset{(\mathrm{E})}{=} f'^{\ast}_{\langle \lambda \rangle}(\pmb{z})^{-1} f'^{\ast}_{\langle \lambda \rangle}(\pmb{z})f'_{\langle \pmb{z} \rangle}(s')\nonumber\\
&= f'_{\langle \pmb{z} \rangle}(s'), \label{eq:gweuyhey56ha}
\end{align}
where
(A) follows from (\ref{eq:tv9pahtdqyls}),
(B) follows from Lemma \ref{lem:f_T} (i) and Lemma \ref{lem:fdot} (i),
(C) follows from (\ref{eq:df0g61poqdfr}),
(D) follows from (\ref{eq:df0g61poqdfr}),
and (E) follows from Lemma \ref{lem:f_T} (i)  and Lemma \ref{lem:fdot} (i).

Also, we have
\begin{align}
\pref(\pmb{d}) 
&\overset{(\mathrm{A})}{=}  \pref(f'^{\ast}_{\langle \lambda \rangle}(\pmb{z}s')) \nonumber\\
&= \pref(f'^{\ast}_{\langle \lambda \rangle}(\pmb{z}) f'_{\langle \pmb{z} \rangle}(s')) \nonumber\\
&\overset{(\mathrm{B})}{=} \pref(f'^{\ast}_{\langle \lambda \rangle}(\pmb{z}) f'_{\langle \pmb{z} \rangle}(s)d_l) \nonumber\\
&= f'^{\ast}_{\langle \lambda \rangle}(\pmb{z}) f'_{\langle \pmb{z} \rangle}(s) \nonumber\\
&= f'^{\ast}_{\langle \lambda \rangle}(\pmb{z}s), \label{eq:5rn1davq8bs5}
\end{align}
where (A) follows from (\ref{eq:df0g61poqdfr}),
and (B) follows from (\ref{eq:gweuyhey56ha}).

By $\pmb{c} \in \PREF^k_{F'', \langle\pmb{z}s'\rangle}$ and (\ref{eq:pref3}), there exist $\pmb{x} \in \mathcal{S}^{L-|\pmb{z}s'|}$ and $\pmb{y} \in \mathcal{S}^{\ast}$ such that
\begin{equation}
\label{eq:d51ydrv24rkh}
f''^{\ast}_{\langle \pmb{z}s' \rangle}(\pmb{x}\pmb{y}) \succeq \pmb{c}.
\end{equation}
By Lemma \ref{lem:F_ext}, we may assume
\begin{equation}
\label{eq:1kpudy9rq6i8}
|f''^{\ast}_{\langle \pmb{z}s'\pmb{x} \rangle}(\pmb{y})| \geq k \geq 1.
\end{equation}
We have
\begin{align*}
\lefteqn{f'_{\langle \pmb{z} \rangle}(s') f'^{\ast}_{\langle \pmb{z}s' \rangle}(\pmb{x}) f''^{\ast}_{\langle \pmb{z}s'\pmb{x} \rangle}(\pmb{y})}\nonumber\\
&\overset{(\mathrm{A})}{=} f'_{\langle \pmb{z} \rangle}(s') \psi_{\pmb{z}s'}(f''^{\ast}_{\langle \pmb{z}s' \rangle}(\pmb{x}) f''^{\ast}_{\langle \pmb{z}s'\pmb{x} \rangle}(\pmb{y}))\\
&\overset{(\mathrm{B})}{=} f'_{\langle \pmb{z} \rangle}(s') \psi_{\pmb{z}s'}(f''^{\ast}_{\langle \pmb{z}s' \rangle}(\pmb{x}\pmb{y}))\\
&\overset{(\mathrm{C})}{\succeq} f'_{\langle \pmb{z} \rangle}(s') \psi_{\pmb{z}s'}(\pmb{c})\\
&\overset{(\mathrm{D})}{=} f'_{\langle \pmb{z} \rangle}(s)d_l \psi_{\pmb{z}s'}(\pmb{c})\\
&\overset{(\mathrm{E})}{=} f'_{\langle \pmb{z} \rangle}(s)d_l \pmb{c}\\
&= f'_{\langle \pmb{z} \rangle}(s) \pref(\pmb{d})^{-1}\pmb{d}\pref(\pmb{d})^{-1}(\pref(\pmb{d})\pmb{c})\\
&\overset{(\mathrm{F})}{=} f'_{\langle \pmb{z} \rangle}(s) f'^{\ast}_{\langle \lambda \rangle}(\pmb{z}s)^{-1}\pmb{d}\pref(\pmb{d})^{-1}(f'^{\ast}_{\langle \lambda \rangle}(\pmb{z}s)\pmb{c})\\
&\overset{(\mathrm{G})}{=} f'_{\langle \pmb{z} \rangle}(s) \psi_{\pmb{z}s}(\pmb{c}),
\end{align*}
where
(A) follows from (\ref{eq:1kpudy9rq6i8}) and the second case of (\ref{eq:gebqpcvgz9hc}),
(B) follows from Lemma \ref{lem:f_T} (i) and Lemma \ref{lem:fdot} (i),
(C) follows from (\ref{eq:d51ydrv24rkh}) and Lemma \ref{lem:psi-1} (i),
(D) follows from (\ref{eq:gweuyhey56ha}),
(E) follows from the second case of (\ref{eq:psi-1}) because $f'^{\ast}_{\langle \lambda \rangle}(\pmb{z}s') \preceq \pref(\pmb{d})$ does not hold by (\ref{eq:df0g61poqdfr}),
(F) follows from (\ref{eq:5rn1davq8bs5}),
and (G) follows from the first case of (\ref{eq:psi-1}) because $f'^{\ast}_{\langle \lambda \rangle}(\pmb{z}s) = \pref(f'^{\ast}_{\langle \lambda \rangle}(\pmb{z}s')) = \pref(\pmb{d}) \prec f'^{\ast}_{\langle \lambda \rangle}(\pmb{z}s)\pmb{c}$ by (\ref{eq:df0g61poqdfr}), (\ref{eq:tv9pahtdqyls}), and (\ref{eq:capkno0i5kee}).

Hence, by $f'_{\langle \pmb{z} \rangle}(s) \prec f'_{\langle \pmb{z} \rangle}(s')$, we have
\begin{equation}
\label{eq:bd695t2710fv}
f'_{\langle \pmb{z} \rangle}(s') f'^{\ast}_{\langle \pmb{z}s' \rangle}(\pmb{x}) [f''^{\ast}_{\langle \pmb{z}s'\pmb{x} \rangle}(\pmb{y})]_k
\succeq f'_{\langle \pmb{z} \rangle}(s) [\psi_{\pmb{z}s}(\pmb{c})]_k.
\end{equation}
Also, we have
\begin{equation}
[f''^{\ast}_{\langle \pmb{z}s'\pmb{x} \rangle}(\pmb{y})]_k \in \PREF^k_{F'', \langle \pmb{z}s'\pmb{x} \rangle} \overset{(\mathrm{A})}{\subseteq} \PREF^k_{F', \langle \pmb{z}s'\pmb{x} \rangle},
\end{equation}
where (A) follows from Lemma \ref{lem:fddot2} (iii) and $\langle \pmb{z}s'\pmb{x} \rangle \in \mathcal{S}^L \subseteq \mathcal{J}$.
Hence, there exists $\pmb{y}' \in \mathcal{S}^{\ast}$ such that
$f'^{\ast}_{\langle \pmb{z}s'\pmb{x} \rangle}(\pmb{y}') \succeq [f''^{\ast}_{\langle \pmb{z}s'\pmb{x} \rangle}(\pmb{y})]_k$,
which leads to 
\begin{align}
\label{eq:7tc5umq31cyj}
f'^{\ast}_{\langle \pmb{z} \rangle}(s'\pmb{x}\pmb{y}')
&= f'_{\langle \pmb{z} \rangle}(s') f'^{\ast}_{\langle \pmb{z}s' \rangle}(\pmb{x}) f'^{\ast}_{\langle \pmb{z}s'\pmb{x} \rangle}(\pmb{y}')
 \nonumber\\
&\succeq f'_{\langle \pmb{z} \rangle}(s') f'^{\ast}_{\langle \pmb{z}s' \rangle}(\pmb{x}) [f'^{\ast}_{\langle \pmb{z}s'\pmb{x} \rangle}(\pmb{y}')]_k
 \nonumber\\
&\overset{(\mathrm{A})}{\succeq} f'_{\langle \pmb{z} \rangle}(s) [\psi_{\pmb{z}s}(\pmb{c})]_k,
\end{align}
where (A) follows from (\ref{eq:bd695t2710fv}).
The assumption that $f'_{\langle \pmb{z} \rangle}(s) \prec f'_{\langle \pmb{z} \rangle}(s')$ and (\ref{eq:7tc5umq31cyj}) shows that
\begin{equation}
\label{eq:jv1udzuc1qdc}
[\psi_{\pmb{z}s}(\pmb{c})]_k \in \bar{\PREF}^k_{F', \langle\pmb{z}\rangle}(f'_{\langle \pmb{z} \rangle}(s))
\end{equation}
by (\ref{eq:pref2}).
By (\ref{eq:ifte7279emhl}) and (\ref{eq:jv1udzuc1qdc}), the code-tuple $F'$ does not satisfy Definition \ref{def:k-bitdelay} (i), which conflicts with $F' \in \mathscr{F}_{k\hdec}$.
\end{itemize}
\end{itemize}
Consequently, $F''$ satisfies Definition \ref{def:k-bitdelay} (ii).
\end{proof}

\section{Conclusion}
\label{sec:conclusion}

This paper discussed the general properties of $k$-bit delay decodable code-tuples for $k \geq 0$ and proved two main theorems.
Theorem \ref{thm:differ} guarantees that it suffices to consider only irreducible code-tuples with at most $2^{(2^k)}$ code tables to achieve the optimal average codeword length.
Theorem \ref{thm:complete} is a generalization of the necessary condition of Huffman codes that every internal node in the code-tree has two child nodes.
Both theorems enable us to limit the scope of code-tuples to be considered when discussing optimal $k$-bit delay decodable code-tuples.

\appendix

\section{Equivalency of the Definitions of a $k$-bit Delay Decodable Code-tuple}
\label{sec:equiv}

We confirm that Definition \ref{def:k-bitdelay} in this paper is equivalent to the definition of a $k$-bit delay decodable code-tuple in \cite{Hashimoto2022}.
We first introduce the following Definition \ref{def:pos-neg}.

\begin{definition}
\label{def:pos-neg}
Let $F(f, \trans) \in \mathscr{F}$ and $i \in [F]$.
\begin{itemize}
\item[(i)] A pair $(\pmb{x}, \pmb{c}) \in \mathcal{S}^{\ast} \times \mathcal{C}^{\ast}$ is said to be \emph{$f^{\ast}_i$-positive} if
for any $\pmb{x}^{\prime} \in \mathcal{S}^{\ast}$, if $f^{\ast}_i(\pmb{x})\pmb{c} \preceq f^{\ast}_i(\pmb{x}^{\prime})$, 
then $\pmb{x} \preceq \pmb{x}^{\prime}$.
\item[(ii)] A pair $(\pmb{x}, \pmb{c}) \in \mathcal{S}^{\ast} \times \mathcal{C}^{\ast}$ is said to be \emph{$f^{\ast}_i$-negative} if
for any $\pmb{x}^{\prime} \in \mathcal{S}^{\ast}$, if $f^{\ast}_i(\pmb{x})\pmb{c} \preceq f^{\ast}_i(\pmb{x}^{\prime})$, 
then $\pmb{x} \not\preceq \pmb{x}^{\prime}$.
\end{itemize}
\end{definition}

Then the definition of a $k$-bit delay decodable code-tuple in \cite{Hashimoto2022} is stated as the following Definition  \ref{def:k-bitdelay0}.

 \begin{definition}
  \label{def:k-bitdelay0}
Let $k \geq 0$ be an integer. A code-tuple $F$ is said to be \emph{$k$-bit delay decodable} if
 for any $i \in [F]$ and $(\pmb{x}, \pmb{c}) \in \mathcal{S}^{\ast} \times \mathcal{C}^k$, the pair $(\pmb{x}, \pmb{c})$ is {$f^{\ast}_i$-positive} or {$f^{\ast}_i$-negative}.
 \end{definition}
 
We show that the following conditions (a) and (b) are equivalent for any $F(f, \trans) \in \mathscr{F}$.
\begin{itemize}
\item[(a)] For any $i \in [F]$ and $(\pmb{x}, \pmb{c}) \in \mathcal{S}^{\ast} \times \mathcal{C}^k$, the pair $(\pmb{x}, \pmb{c})$ is $f^{\ast}_i$-positive or $f^{\ast}_i$-negative.
\item[(b)] The code-tuple $F$ satisfies Definition \ref{def:k-bitdelay} (i) and (ii).
\end{itemize}

((a) $\implies$ (b)): We show the contraposition.
Assume that (b) does not hold.
We consider the following two cases separately: the case where Definition \ref{def:k-bitdelay} (i) is false and the case where Definition \ref{def:k-bitdelay} (ii) is false.

\begin{itemize}
\item The case where Definition \ref{def:k-bitdelay} (i) is false:
Then there exist $i \in [F]$, $s\in \mathcal{S}$, and $\pmb{c} \in \PREF^k_{F, \trans_i(s)} \cap \bar{\PREF}^k_{F, i}(f_i(s))$.
By (\ref{eq:pref2}) and (\ref{eq:pref3}), there exist $\pmb{x} = x_1x_2\ldots, x_n \in \mathcal{S}^\ast$ and $\pmb{x}' = x'_1x'_2\ldots x'_{n'} \in \mathcal{S}^{+}$ such that
\begin{equation}
\label{eq:bqqsx722vk74}
f^{\ast}_{\trans_i(s)}(\pmb{x}) \succeq \pmb{c},
\end{equation}
\begin{equation}
\label{eq:456ve1a9snhp}
f^{\ast}_i(\pmb{x}') \succeq f_i(s)\pmb{c},
\end{equation}
 \begin{equation}
\label{eq:94h7r770efbn}
 f_i(x'_1) \succ f_i(s).
 \end{equation}
We have
\begin{equation}
\label{eq:t6dfks2mngfi}
f^{\ast}_i(s\pmb{x}) \overset{(\mathrm{A})}{=} f_i(s)f^{\ast}_{\trans_i(s)}(\pmb{x}) \overset{(\mathrm{B})}{\succeq} f_i(s)\pmb{c},
\end{equation}
where
(A) follows from (\ref{eq:fstar}),
and (B) follows from (\ref{eq:bqqsx722vk74}).
By (\ref{eq:t6dfks2mngfi}) and $s \preceq s\pmb{x}$, the pair $(s, \pmb{c})$ is not $f^{\ast}_i$-negative.
On the other hand, since $s \neq x'_1$ by (\ref{eq:94h7r770efbn}), we have $s \not\preceq \pmb{x}'$.
Hence, by (\ref{eq:456ve1a9snhp}), the pair $(s, \pmb{c})$ is not $f^{\ast}_i$-positive.
Since the pair $(s, \pmb{c})$ is neither $f^{\ast}_i$-positive nor $f^{\ast}_i$-negative, the condition (a) does not hold.

\item The case where Definition \ref{def:k-bitdelay} (ii) is false:
Then there exist $i \in [F]$, $s, s' \in \mathcal{S}$, and $\pmb{c} \in \PREF^k_{F, \trans_i(s)} \cap \PREF^k_{F, \trans_i(s')}$ such that $s \neq s'$ and
\begin{equation}
\label{eq:m95lm043j0dq}
f_i(s) = f_i(s').
\end{equation}
By (\ref{eq:pref3}), there exist $\pmb{x}, \pmb{x}' \in \mathcal{S}^{\ast}$ such that
\begin{equation}
\label{eq:g6b5snkgshq3}
f^{\ast}_{\trans_i(s)}(\pmb{x}) \succeq \pmb{c}
\end{equation}
and
\begin{equation}
\label{eq:r8wdebklb7gd}
f^{\ast}_{\trans_i(s')}(\pmb{x}') \succeq \pmb{c}.
\end{equation}
Thus, we have
\begin{equation}
\label{eq:gk1gudp2s9xq}
f^{\ast}_i(s\pmb{x}) \overset{(\mathrm{A})}{=} f_i(s)f^{\ast}_{\trans_i(s)}(\pmb{x}) \overset{(\mathrm{B})}{\succeq} f_i(s)\pmb{c}
\end{equation}
and
\begin{align}
f^{\ast}_i(s'\pmb{x}')
&\overset{(\mathrm{C})}{=} f_i(s')f^{\ast}_{\trans_i(s')}(\pmb{x}') 
\overset{(\mathrm{D})}{=} f_i(s)f^{\ast}_{\trans_i(s')}(\pmb{x}')
\overset{(\mathrm{E})}{\succeq} f_i(s)\pmb{c}, \label{eq:c3fystg10ttz}
\end{align}
where
(A) follows from (\ref{eq:fstar}),
(B) follows from (\ref{eq:g6b5snkgshq3}),
(C) follows from (\ref{eq:fstar}),
(D) follows from (\ref{eq:m95lm043j0dq}),
and (E) follows from (\ref{eq:r8wdebklb7gd}).
By (\ref{eq:gk1gudp2s9xq}) and $s \preceq s\pmb{x}$, the pair $(s, \pmb{c})$ is not $f^{\ast}_i$-negative.
On the other hand, by $s \not\preceq s'\pmb{x}$ and (\ref{eq:c3fystg10ttz}), the pair $(s, \pmb{c})$ is not $f^{\ast}_i$-positive.
Since the pair $(s, \pmb{c})$ is neither $f^{\ast}_i$-positive nor $f^{\ast}_i$-negative, the condition (a) does not hold.
\end{itemize}

((b) $\implies$ (a)): 
We show the contraposition.
Assume that (a) does not hold. Then there exist $i \in [F]$ and $(\pmb{x}, \pmb{c}) \in \mathcal{S}^{\ast} \times \mathcal{C}^k$ such that $(\pmb{x}, \pmb{c})$ is neither $f^{\ast}_i$-positive nor $f^{\ast}_i$-negative.
Thus, there exist $\pmb{x}', \pmb{x}'' \in \mathcal{S}^{\ast}$ such that
\begin{equation}
\label{eq:sgw0eotiv293}
f^{\ast}_i(\pmb{x})\pmb{c} \preceq f^{\ast}_i(\pmb{x}'),
\end{equation}
\begin{equation}
\label{eq:wfkfx6eosfry}
f^{\ast}_i(\pmb{x})\pmb{c} \preceq f^{\ast}_i(\pmb{x}''),
\end{equation}
\begin{equation}
\label{eq:0y2hqitfrs0g}
\pmb{x} \preceq \pmb{x}',
\end{equation}
\begin{equation}
\label{eq:mjqmn1nueyib}
\pmb{x} \not\preceq \pmb{x}''.
\end{equation}

We consider the following two cases separately: the case $\pmb{x} \succeq \pmb{x}''$ and the case $\pmb{x} \not\succeq \pmb{x}''$.
\begin{itemize}
\item The case $\pmb{x} \succeq \pmb{x}''$:
By Lemma \ref{lem:f_T} (iii), we have
\begin{equation}
\label{eq:c61y8syqkodc}
f^{\ast}_i(\pmb{x}) \succeq f^{\ast}_i(\pmb{x}'').
\end{equation}
Hence, by (\ref{eq:wfkfx6eosfry}), it must hold that $\pmb{c} = \lambda$.
Namely, only $k = 0$ is possible now.

Since (\ref{eq:mjqmn1nueyib}) and $\pmb{x} \succeq \pmb{x}''$ lead to $\pmb{x} \succ \pmb{x}''$, 
there exists $\pmb{u} = u_1u_2\ldots u_n \in \mathcal{S}^{+}$ such that $\pmb{x} = \pmb{x}''\pmb{u}$.
Defining $j \coloneqq \trans^{\ast}_i(\pmb{x}'')$, we have
\begin{align}
f^{\ast}_i(\pmb{x})
&\overset{(\mathrm{A})}{=}  f^{\ast}_i(\pmb{x}'')f^{\ast}_j(\pmb{u}) 
\overset{(\mathrm{B})}{=} f^{\ast}_i(\pmb{x})f^{\ast}_j(\pmb{u}) 
\overset{(\mathrm{C})}{=} f^{\ast}_i(\pmb{x})f^{\ast}_j(u_1)f^{\ast}_{\trans_j(u_1)}(\suff(\pmb{u})), \label{eq:jubkcr32zgcc}
\end{align}
where 
(A) follows from Lemma \ref{lem:f_T} (i),
(B) follows because we have $f^{\ast}_i(\pmb{x}) \preceq  f^{\ast}_i(\pmb{x}'')$ by (\ref{eq:wfkfx6eosfry}) and we have $f^{\ast}_i(\pmb{x}) \succeq  f^{\ast}_i(\pmb{x}'')$ by (\ref{eq:c61y8syqkodc}),
and (C) follows from (\ref{eq:fstar}).
Comparing both sides of (\ref{eq:jubkcr32zgcc}), we obtain
\begin{equation}
\label{eq:r3hvtkqz4af3}
f_j(u_1) = \lambda
\end{equation}
and $f^{\ast}_{\trans_j(u_1)}(\suff(\pmb{u})) = \lambda$.

We now show that (b) does not hold dividing into two cases by whether $f_j$ is injective.
\begin{itemize}
\item If $f_j$ is not injective, then $F$ does not satisfy Definition \ref{def:k-bitdelay} (ii) by $k = 0$ and Lemma \ref{lem:0dec-prefixfree}.
\item If $f_j$ is injective, then there exists $s \in \mathcal{S}$ such that $f_j(s) \succ \lambda$ by $\sigma \geq 2$, which leads to
\begin{equation}
\label{eq:oljo9b0jd0fx}
\bar{\PREF}^0_{F, j} \neq \emptyset
\end{equation}
by (\ref{eq:pref2}).
We see that $F$ does not satisfy Definition \ref{def:k-bitdelay} (i) because 
\begin{align*}
\PREF^0_{F, \trans_j(u_1)} \cap \bar{\PREF}^0_{F, j}(f_j(u_1))
&\overset{(\mathrm{A})}{=} \PREF^0_{F, \trans_j(u_1)} \cap \bar{\PREF}^0_{F, j}\\
&\overset{(\mathrm{B})}{=} \{\lambda\} \cap \bar{\PREF}^0_{F, j}\\
&\overset{(\mathrm{C})}{=} \{\lambda\} \cap \{\lambda\}\\
&= \{\lambda\}\\
&\neq \emptyset,
\end{align*}
where
(A) follows from (\ref{eq:r3hvtkqz4af3}),
(B) follows from (\ref{eq:pref3}),
and (C) follows from (\ref{eq:oljo9b0jd0fx}).
\end{itemize}

\item The case $\pmb{x} \not\succeq \pmb{x}''$:
By (\ref{eq:mjqmn1nueyib}) and $\pmb{x} \not\succeq \pmb{x}''$, 
there exist $\pmb{z} = z_1z_2\ldots z_{n} \in \mathcal{S}^{+}$ and $\pmb{z}'' = z''_1z''_2\ldots z''_{n''} \in \mathcal{S}^{+}$ such that
\begin{align}
\pmb{x} &= \pmb{y}\pmb{z}, \label{eq:wvf7fojs80yn}\\
\pmb{x}'' &= \pmb{y}\pmb{z}'', \label{eq:4bb6m52ayl25}\\
z_1 &\neq z''_1 \label{eq:xxxjqkzwgkel},
\end{align}
where $\pmb{y}$ is the longest common prefix of $\pmb{x}$ and $\pmb{x}''$.
Also, by (\ref{eq:0y2hqitfrs0g}), there exists $\pmb{w} \in \mathcal{S}^{\ast}$ such that
\begin{equation}
\pmb{x}' = \pmb{xw}.
\end{equation}
Defining $\pmb{z}' = z'_1z'_2\ldots z'_{n'} \coloneqq \pmb{zw}$, we have
\begin{equation}
\label{eq:q0qp077d2kx6}
\pmb{x}' = \pmb{xw} = \pmb{yzw} = \pmb{yz}',
\end{equation}
\begin{equation}
\label{eq:rgqayx9ks4yg}
z_1 = z'_1.
\end{equation}

Then defining $j \coloneqq \trans^{\ast}_i(\pmb{y})$, we have 
\begin{align}
&\lefteqn{f^{\ast}_i(\pmb{y}) f_j(z'_1) f^{\ast}_{\trans_j(z'_1)}(\suff(\pmb{z}'))}\nonumber\\
&\overset{(\mathrm{A})}{=} f^{\ast}_i(\pmb{y}) f^{\ast}_j(\pmb{z}')
\overset{(\mathrm{B})}{=} f^{\ast}_i(\pmb{y} \pmb{z}')
\overset{(\mathrm{C})}{=} f^{\ast}_i(\pmb{x}')
\overset{(\mathrm{D})}{\succeq} f^{\ast}_i(\pmb{x})\pmb{c}, \label{eq:qqxgu76jwgm2}
\end{align}
where
(A) follows from (\ref{eq:fstar}),
(B) follows from Lemma \ref{lem:f_T} (i),
(C) follows from (\ref{eq:q0qp077d2kx6}),
and (D) follows from (\ref{eq:sgw0eotiv293}).
Similarly, by (\ref{eq:wfkfx6eosfry}) and (\ref{eq:4bb6m52ayl25}), we have
\begin{align}
&\lefteqn{f^{\ast}_i(\pmb{y}) f_j(z''_1) f^{\ast}_{\trans_j(z''_1)}(\suff(\pmb{z}''))}\nonumber\\
&=  f^{\ast}_i(\pmb{y}) f^{\ast}_j(\pmb{z}'') 
=  f^{\ast}_i(\pmb{y} \pmb{z}'')
= f^{\ast}_i(\pmb{x}'')
\succeq f^{\ast}_i(\pmb{x})\pmb{c}. \label{eq:bm0j7xt71xqb}
\end{align}

Also, we have
\begin{align}
f^{\ast}_i(\pmb{x})\pmb{c}
&\overset{(\mathrm{A})}{=} f^{\ast}_i(\pmb{y}\pmb{z}) \pmb{c} 
\overset{(\mathrm{B})}{=} f^{\ast}_i(\pmb{y})f^{\ast}_j(\pmb{z}) \pmb{c} 
\overset{(\mathrm{C})}{=} f^{\ast}_i(\pmb{y})f_j(z_1) f^{\ast}_{\trans_j(z_1)}(\suff(\pmb{z})) \pmb{c}, \label{eq:xh2hmm1vpzvn}
\end{align}
where
(A) follows from (\ref{eq:wvf7fojs80yn}),
(B) follows from Lemma \ref{lem:f_T} (i),
and (C) follows from (\ref{eq:fstar}).

Thus, we have
\begin{align}
f_j(z'_1) f^{\ast}_{\trans_j(z'_1)}(\suff(\pmb{z}'))
&\overset{(\mathrm{A})}{\succeq} f_j(z_1) f^{\ast}_{\trans_j(z_1)}(\suff(\pmb{z})) \pmb{c} \nonumber\\
&\overset{(\mathrm{B})}{=}  f_j(z'_1) f^{\ast}_{\trans_j(z'_1)}(\suff(\pmb{z})) \pmb{c}  \nonumber\\
&\succeq f_j(z'_1) \pmb{c}', \label{eq:ywtm3nqdxg92}
\end{align}
where $\pmb{c}' \in \mathcal{C}^k$ is defined as the prefix of length $k$ of $f^{\ast}_{\trans_j(z'_1)}(\suff(\pmb{z})) \pmb{c}$, and
(A) follows from (\ref{eq:qqxgu76jwgm2}) and (\ref{eq:xh2hmm1vpzvn}),
and (B) follows from (\ref{eq:rgqayx9ks4yg}).
Similarly, we have
\begin{align}
f_j(z''_1) f^{\ast}_{\trans_j(z''_1)}(\suff(\pmb{z}''))
&\overset{(\mathrm{A})}{\succeq} f_j(z_1) f^{\ast}_{\trans_j(z_1)}(\suff(\pmb{z})) \pmb{c} \nonumber\\
&\overset{(\mathrm{B})}{=}  f_j(z'_1) f^{\ast}_{\trans_j(z'_1)}(\suff(\pmb{z})) \pmb{c} \nonumber\\
&\succeq f_j(z'_1) \pmb{c}', \label{eq:l1p7sy6pk3qp}
\end{align}
where
(A) follows from (\ref{eq:bm0j7xt71xqb}) and (\ref{eq:xh2hmm1vpzvn}),
and (B) follows from (\ref{eq:rgqayx9ks4yg}).
By (\ref{eq:ywtm3nqdxg92}), we have $f^{\ast}_{\trans_j(z'_1)}(\suff(\pmb{z}')) \succeq \pmb{c}'$,
which leads to
\begin{equation}
\label{eq:2e5xkiwgd6wd}
\pmb{c}' \in \PREF^k_{F, \trans_j(z'_1)}
\end{equation}
by (\ref{eq:pref3}).

By (\ref{eq:l1p7sy6pk3qp}), at least one of $f_j(z'_1) \preceq f_j(z''_1)$ and $f_j(z'_1) \succeq f_j(z''_1)$ holds.
We may assume $f_j(z'_1) \preceq f_j(z''_1)$ by symmetry.
We consider the following two cases separately: the case $f_j(z'_1) \prec f_j(z''_1)$ and the case $f_j(z'_1) = f_j(z''_1)$.
\begin{itemize}
\item the case $f_j(z'_1) \prec f_j(z''_1)$:
We have
\begin{equation}
\label{eq:icbm2ucps4rc}
f^{\ast}_j(\pmb{z}'') \overset{(\mathrm{A})}{=} f_j(z''_1)f^{\ast}_{\trans_j(z''_1)}(\suff(\pmb{z}'')) \overset{(\mathrm{B})}{\succeq} f_j(z'_1) \pmb{c}',
\end{equation}
where
(A) follows from (\ref{eq:fstar}),
and (B) follows from (\ref{eq:l1p7sy6pk3qp}).
By (\ref{eq:icbm2ucps4rc}) and $f_j(z'_1) \prec f_j(z''_1)$, we obtain
\begin{equation}
\label{eq:n2w0ykhkz0ut}
\pmb{c}' \in \bar{\PREF}^k_{F, j}(f_j(z'_1))
\end{equation}
by (\ref{eq:pref2}).
By (\ref{eq:2e5xkiwgd6wd}) and (\ref{eq:n2w0ykhkz0ut}), the code-tuple $F$ does not satisfy Definition \ref{def:k-bitdelay} (i).

\item the case $f_j(z'_1) = f_j(z''_1)$: 
We have 
\begin{align*}
f_j(z'_1) f^{\ast}_{\trans_j(z''_1)}(\suff(\pmb{z}''))
&\overset{(\mathrm{A})}{=} f_j(z''_1) f^{\ast}_{\trans_j(z''_1)}(\suff(\pmb{z}''))
\overset{(\mathrm{B})}{\succeq} f_j(z'_1)\pmb{c}',
\end{align*}
where
(A) follows from $f_j(z'_1) = f_j(z''_1)$,
and (B) follows from (\ref{eq:l1p7sy6pk3qp}).
This shows $f^{\ast}_{\trans_j(z''_1)}(\suff(\pmb{z}'')) \succeq \pmb{c}'$, which leads to 
\begin{equation}
\label{eq:xyo1q6tzfxox}
\pmb{c}' \in \PREF^k_{F, \trans_j(z''_1)}
\end{equation}
by (\ref{eq:pref3}).
By $f_j(z'_1) = f_j(z''_1)$, (\ref{eq:xxxjqkzwgkel}), (\ref{eq:2e5xkiwgd6wd}), and (\ref{eq:xyo1q6tzfxox}),
the code-tuple $F$ does not satisfy Definition \ref{def:k-bitdelay} (ii).
\end{itemize}

\end{itemize}

\section{Proof of the Existence of an Optimal Code-tuple}
\label{sec:proof-goodsetTopt}

For $m \in \{1, 2, \ldots, M \coloneqq 2^{(2^k)}\}$, the number of possible tuples $(\trans_0, \trans_1, \ldots, \trans_{m-1})$ (i.e., a tuple of $m$ mappings from $\mathcal{S}$ to $[m]$) is $m^{\sigma m}$, in particular, finite.
Hence, the number of possible vectors $\pmb{\pi}(F') = (\pi_0(F'), \pi_1(F'), \ldots, \pi_{m-1}(F'))$ of a code-tuple $F' \in \mathscr{F}'$ is also finite (cf. Remark \ref{rem:transitionprobability}), where
\begin{equation}
\mathscr{F}' \coloneqq \{F' \in \mathscr{F}_{\irr} \cap \mathscr{F}_{\ext} \cap \mathscr{F}_{k\hdec}: |F'| \leq M\}.
\end{equation}
Therefore, $\mathcal{D} \coloneqq \{\pi_i(F') : F' \in \mathscr{F}', i \in [F']\}$ is a finite set and has the minimum value $\delta \coloneqq \min \mathcal{D}$.
Note that $\delta > 0$ holds since $\pi_i(F') > 0$ for any $F' \in \mathscr{F}'$ and $i \in [F']$ by $\mathscr{F}' \subseteq \mathscr{F}_{\irr}$ and Lemma \ref{lem:kernel} (ii).

Now, we define
\begin{equation}
\mathscr{F}'' \coloneqq \{F'(f', \trans') \in \mathscr{F}': \sum_{i \in [F'], s \in \mathcal{S}} |f'_i(s)| \leq \frac{l}{\delta\nu}\},
\end{equation}
where $l \coloneqq \lceil \log_2 \sigma \rceil$ and $\nu \coloneqq \min_{s \in \mathcal{S}} \mu(s)$.
Note that
\begin{equation}
\label{eq:f0uoc788q7d5}
0 < \nu \leq 1 / \sigma.
\end{equation}
Then $\mathscr{F}''$ is not empty because $\tilde{F}(\tilde{f}_0, \tilde{\trans}_0) \in \mathscr{F}^{(1)}$ defined as the following (\ref{eq:u5w1u1d7szub}) is in $\mathscr{F}''$:
\begin{align}
\label{eq:u5w1u1d7szub}
\tilde{f}_0(s_r) = b(r), \quad \tilde{\trans}_0(s_r) = 0
\end{align}
for $r = 0, 1, 2, \ldots, \sigma-1$, where $\mathcal{S} = \{s_0, s_1, \ldots, s_{\sigma-1}\}$ and $b(r)$ denotes the binary representation of length $l$ of the integer $r$.
In fact, we obtain $\tilde{F} \in \mathscr{F}''$ by
\begin{equation}
\sum_{i \in [\tilde{F}], s \in \mathcal{S}} |\tilde{f}_i(s)| = \sum_{s \in \mathcal{S}} |\tilde{f}_0(s)| =  \sigma l
\overset{(\mathrm{A})}{\leq} \frac{l}{\nu} \overset{(\mathrm{B})}{\leq} \frac{l}{\delta\nu},
\end{equation}
where
(A) follows from (\ref{eq:f0uoc788q7d5}),
and (B) follows from $0 < \delta \leq 1$.
Since $\mathscr{F}''$ is a non-empty and finite set, there exists $F^{\ast} \in \mathscr{F}''$ such that
\begin{equation}
\label{eq:4w134pp591j2}
L(F^{\ast}) = \min_{F'' \in \mathscr{F}''} L(F'').
\end{equation}
To complete the proof, it suffices to show that $L(F^{\ast}) \leq L(F)$ for any $F  \in \mathscr{F}_{\reg} \cap \mathscr{F}_{\ext} \cap \mathscr{F}_{k\hdec}$.

First, we can see that $L(F^{\ast}) \leq L(F')$ for any $F'  \in \mathscr{F}'$ because
for any $F'(f', \trans') \in \mathscr{F}' \setminus \mathscr{F}''$, we have
\begin{align*}
L(F') &= \sum_{i \in [F']}\pi_i(F') L_i(F')\\
&= \sum_{i \in [F']}\pi_i(F') \sum_{s \in \mathcal{S}} \mu(s)|f'_i(s)|\\
&\overset{(\mathrm{A})}{\geq} \delta\nu \sum_{i \in [F'], s \in \mathcal{S}} |f'_i(s)|\\
&\overset{(\mathrm{B})}{>} \delta\nu \cdot \frac{l}{\delta\nu}\\
&= l\\
&= L(\tilde{F})\\
&\overset{(\mathrm{C})}{\geq} L(F^{\ast}),
\end{align*}
where
(A) follows from the definitions of $\delta$ and $\nu$,
(B) follows from $F' \not\in \mathscr{F}''$,
and (C) follows from (\ref{eq:4w134pp591j2}).
Hence, we have
\begin{equation}
\label{eq:ce0kvxm7dbb4}
L(F^{\ast}) = \min_{F' \in \mathscr{F}'} L(F').
\end{equation}

By Theorem \ref{thm:differ}, for any $F \in \mathscr{F}_{\reg} \cap \mathscr{F}_{\ext} \cap \mathscr{F}_{k\hdec}$,
there exists $F' \in \mathscr{F}_{\irr} \cap \mathscr{F}_{\ext} \cap \mathscr{F}_{k\hdec}$ such that $L(F') \leq L(F)$ and $|\prefset^k_{F'}| = |F'|$.
Then we have $F' \in \mathscr{F}'$ because
\begin{equation}
|F'| = |\prefset^k_{F'}| \leq |\mathcal{P}({\mathcal{C}^k})| = 2^{(2^k)} = M,
\end{equation}
where $\mathcal{P}({\mathcal{C}^k})$ denotes the power set of $\mathcal{C}^k$.
Therefore, for any $F \in \mathscr{F}_{\reg} \cap \mathscr{F}_{\ext} \cap \mathscr{F}_{k\hdec}$, we have
\begin{equation}
L(F) \geq L(F') \overset{(\mathrm{A})}{\geq} L(F^{\ast})
\end{equation}
as desired, where
(A) follows from (\ref{eq:ce0kvxm7dbb4}).

\section{Proofs of Lemmas}
\label{sec:proofs}

\subsection{Proof of Lemma \ref{lem:longest}}
\label{sec:proof-longest}

To prove Lemma \ref{lem:longest}, we first show the following Lemma \ref{lem:pref-sum}.
\begin{lemma}
\label{lem:pref-sum}
For any integer $k \geq 0$, $F(f, \trans) \in \mathscr{F}$ and $i \in [F]$, the following statements (i) and (ii) hold.
\begin{itemize}
\item[(i)] $\PREF^k_{F, i} \supseteq \bar{\PREF}^k_{F, i}$.
\item[(ii)] For any $s \in \mathcal{S}$ such that $f_i(s) = \lambda$, we have $\PREF^k_{F, i} \supseteq \PREF^k_{F, \trans_i(s)}$.
\end{itemize}
\end{lemma}

\begin{proof}[Proof of Lemma \ref{lem:pref-sum}]
(Proof of (i)): Directly from (\ref{eq:pref1}) and (\ref{eq:pref2}).

(Proof of (ii)):
Choose $\pmb{c} \in \PREF^k_{F, \trans_i(s)}$ arbitrarily.
Then there exists $\pmb{x} \in \mathcal{S}^{\ast}$ such that
\begin{equation}
\label{eq:epkmwg4abc6w}
f^{\ast}_{\trans_i(s)}(\pmb{x}) \succeq \pmb{c}
\end{equation}
by (\ref{eq:pref3}).
We have 
\begin{equation}
f^{\ast}_i(s\pmb{x})
\overset{(\mathrm{A})}{=} f_i(s)f^{\ast}_{\trans_i(s)}(\pmb{x})
\overset{(\mathrm{B})}{=} f^{\ast}_{\trans_i(s)}(\pmb{x})
\overset{(\mathrm{C})}{\succeq} \pmb{c},
\end{equation}
where
(A) follows from Lemma \ref{lem:f_T} (i),
(B) follows from the assumption,
and (C) follows from (\ref{eq:epkmwg4abc6w}).
This leads to $\pmb{c} \in \PREF^k_{F, i}$.
\end{proof}

\begin{proof}[Proof of Lemma \ref{lem:longest}]
It suffices to show that $|f^{\ast}_i(\pmb{x})| \geq 1$ holds for any $i \in [F]$ and $\pmb{x} \in \mathcal{S}^{|F|}$.
We prove by contradiction assuming that there exist $i \in [F]$ and $\pmb{x} = x_1x_2\ldots x_{|F|} \in \mathcal{S}^{|F|}$ such that $f^{\ast}_i(\pmb{x}) = \lambda$.
Then by pigeonhole principle, we can choose integers $p, q$ such that $0 \leq p < q \leq |F|$ and
\begin{equation}
\label{eq:g2gg9zkmkjej}
\trans^{\ast}_i(x_1x_2\ldots x_p) = \trans^{\ast}_i(x_1x_2\ldots x_q) \eqqcolon j.
\end{equation}
We have
\begin{align}
\trans^{\ast}_j(x_{p+1}x_{p+2}\ldots x_q)
&\overset{(\mathrm{A})}{=} \trans^{\ast}_{\trans^{\ast}_i(x_1x_2\ldots x_p)}(x_{p+1}x_{p+2}\ldots x_q)
\overset{(\mathrm{B})}{=} \trans^{\ast}_i(x_1x_2\ldots x_q) \overset{(\mathrm{C})}{=} j,\label{eq:ghzv2yj7kboz}
\end{align}
where
(A) follows from (\ref{eq:g2gg9zkmkjej}),
(B) follows from Lemma \ref{lem:f_T} (ii),
and (C) follows from (\ref{eq:g2gg9zkmkjej}).
Thus, we obtain
\begin{align}
\PREF^k_{F, \trans_j(x_{p+1})}
&\overset{(\mathrm{A})}{\supseteq} \PREF^k_{F, \trans^{\ast}_j(x_{p+1}x_{p+2})} \nonumber\\
&\overset{(\mathrm{A})}{\supseteq}  \cdots \nonumber\\
&\overset{(\mathrm{A})}{\supseteq}  \PREF^k_{F, \trans^{\ast}_j(x_{p+1}x_{p+2}\ldots x_q)} \nonumber\\
&\overset{(\mathrm{B})}{=} \PREF^k_{F, j}, \label{eq:fmey1lo6g0mb}
\end{align}
where
(A)s follow from Lemma \ref{lem:pref-sum} (ii) and $f^{\ast}_i(\pmb{x}) = \lambda$,
and (B) follows from (\ref{eq:ghzv2yj7kboz}).

We consider the following two cases separately: the case $\bar{\PREF}^k_{F, j} \neq \emptyset$ and the case $\bar{\PREF}^k_{F, j} = \emptyset$.
\begin{itemize}
\item The case $\bar{\PREF}^k_{F, j} \neq \emptyset$: 
We have 
\begin{align*}
\PREF^k_{F, \trans_j(x_{p+1})} \cap \bar{\PREF}^k_{F, j}
&\overset{(\mathrm{A})}{\supseteq} \PREF^k_{F, j}  \cap \bar{\PREF}^k_{F, j}
\overset{(\mathrm{B})}{\supseteq} \bar{\PREF}^k_{F, j} \cap \bar{\PREF}^k_{F, j}
= \bar{\PREF}^k_{F, j}
\overset{(\mathrm{C})}{\neq} \emptyset,
\end{align*}
where
 (A) follows from (\ref{eq:fmey1lo6g0mb}),
 (B) follows from Lemma \ref{lem:pref-sum} (i),
and  (C) follows from the assumption.
Therefore, $F$ does not satisfy Definition \ref{def:k-bitdelay} (i), which conflicts with $F \in \mathscr{F}_{k\hdec}$.
\item  The case $\bar{\PREF}^k_{F, j} = \emptyset$: By Corollary \ref{cor:pref-inc} (ii), we have $\bar{\PREF}^0_{F, j} = \emptyset$.
Hence, by (\ref{eq:pref2}), there is no symbol $s' \in \mathcal{S}$ such that $f_j(s') \succ \lambda$.
Therefore, by $\sigma \geq 2$, there exists $s \in \mathcal{S}$ such that $s \neq x_{p+1}$ and $f_j(s) = \lambda = f_j(x_{p+1})$.
We have
\begin{align*}
\PREF^k_{F, \trans_j(x_{p+1})} \cap \PREF^k_{F, \trans_j(s)}
&\overset{(\mathrm{A})}{\supseteq} \PREF^k_{F, j}  \cap \PREF^k_{F, \trans_j(s)}\\
&\overset{(\mathrm{B})}{\supseteq} \PREF^k_{F, \trans_j(s)} \cap \PREF^k_{F, \trans_j(s)}\\
&= \PREF^k_{F, \trans_j(s)}  \overset{(\mathrm{C})}{\neq} \emptyset,
\end{align*}
where
(A) follows from (\ref{eq:fmey1lo6g0mb}),
(B) follows from Lemma \ref{lem:pref-sum} (ii),
and (C) follows from $F \in \mathscr{F}_{\ext}$ and Corollary \ref{cor:pref-inc} (i).
Therefore, $F$ does not satisfy Definition \ref{def:k-bitdelay} (ii), which conflicts with $F \in \mathscr{F}_{k\hdec}$.
\end{itemize}
\end{proof}

\subsection{Proof of Lemma \ref{lem:stationary}}
\label{subsec:proof-stationary}

In preparation for the proof, we introduce the following Definition \ref{def:closed} and Lemma \ref{lem:kernel02}.

\begin{definition}
\label{def:closed}
Let $F(f, \trans) \in \mathscr{F}$.
A set $\mathcal{I} \subseteq [F]$ is said to be \emph{closed} if 
for any $i \in \mathcal{I}$ and $s \in \mathcal{S}$, it holds that $\trans_i(s) \in \mathcal{I}$.
\end{definition}

\begin{lemma}
\label{lem:kernel02}
For any $F \in \mathscr{F}$ and $\pmb{x} = (x_0, x_1, \ldots, \allowbreak x_{|F|-1}) \in \mathbb{R}^{|F|}$,
if
\begin{equation}
\label{eq:7gvchdcogtn4}
\pmb{x}Q(F) = \pmb{x},
\end{equation}
then both of $\mathcal{I}_+ \coloneqq \{i \in [F] : x_i > 0\}$ and $\mathcal{I}_- \coloneqq \{i \in [F] : x_i < 0\}$ are closed.
\end{lemma}

\begin{proof}[Proof of Lemma \ref{lem:kernel02}]
By symmetry, it suffices to prove only that $\mathcal{I}_+$ is closed.
We have
\begin{align}
\lefteqn{\sum_{i \in \mathcal{I}_+} \sum_{j \in \mathcal{I}_+} x_j Q_{j, i}(F) + \sum_{i \in \mathcal{I}_+} \sum_{j \in [F] \setminus \mathcal{I}_+} x_j Q_{j, i}(F)}\nonumber\\
&= \sum_{i \in \mathcal{I}_+} \sum_{j \in [F]} x_jQ_{j, i}(F)\nonumber\\
&\overset{(\mathrm{A})}{=} \sum_{i \in \mathcal{I}_+} x_i \nonumber\\
&\overset{(\mathrm{B})}{=} \sum_{i \in \mathcal{I}_+} x_i \sum_{j \in [F]} Q_{i, j}(F) \nonumber\\
&= \sum_{i \in \mathcal{I}_+} \sum_{j \in [F]}  x_i Q_{i, j}(F) \nonumber\\
&= \sum_{i \in \mathcal{I}_+} \sum_{j \in \mathcal{I}_+}  x_i Q_{i, j}(F) + \sum_{i \in \mathcal{I}_+} \sum_{j \in [F] \setminus \mathcal{I}_+}  x_i Q_{i, j}(F) \nonumber\\
&\overset{(\mathrm{C})}{=} \sum_{i \in \mathcal{I}_+} \sum_{j \in \mathcal{I}_+}  x_j Q_{j, i}(F) + \sum_{i \in \mathcal{I}_+} \sum_{j \in [F] \setminus \mathcal{I}_+}  x_i Q_{i, j}(F), \label{eq:ztvib4xnlpoi}
\end{align}
where
(A) follows since from (\ref{eq:7gvchdcogtn4}),
(B) follows from $\sum_{j \in [F]} Q_{i, j}(F) = 1$ for any $i \in [F]$,
and (C) is obtained by exchanging the roles of $i$ and $j$ in the first term.
Therefore, we have 
\begin{align*}
\label{eq:stationary7}
0 &\overset{(\mathrm{A})}{\geq} \sum_{i \in \mathcal{I}_+} \sum_{j \in [F] \setminus \mathcal{I}_+} x_j Q_{j, i}(F)
\overset{(\mathrm{B})}{=} \sum_{i \in \mathcal{I}_+} \sum_{j \in [F] \setminus \mathcal{I}_+}  x_i Q_{i, j}(F)
\overset{(\mathrm{C})}{\geq} 0,
\end{align*}
where
(A) follows since $x_j \leq 0$ for any $j \in [F] \setminus \mathcal{I}_+$,
(B) is obtained by eliminating the first terms from the leftmost and rightmost sides of (\ref{eq:ztvib4xnlpoi}),
and (C) follows since $x_i > 0$ for any $i \in \mathcal{I}_+$.
This shows
\begin{equation*}
\sum_{i \in \mathcal{I}_+} \sum_{j \in [F] \setminus \mathcal{I}_+}  x_i Q_{i, j}(F) = 0.
\end{equation*}
Since $x_i > 0$ holds for any $i \in \mathcal{I}_+$, it must hold that $Q_{i, j}(F) = 0$ for any $i \in \mathcal{I}_+$ and $j \in [F] \setminus \mathcal{I}_+$.
This implies that for any $i \in \mathcal{I}_+$ and $s \in \mathcal{S}$, we have $\trans_i(s) \in \mathcal{I}_+$;
that is, $\mathcal{I}_{+}$ is closed as desired.
\end{proof}

\begin{proof}[Proof of Lemma \ref{lem:stationary}]
Equation (\ref{eq:stationary1}) can be rewritten as
\begin{equation}
\label{eq:weaw3ce74e22}
\pmb{\pi}A = \pmb{0},
\end{equation}
where $A = (A_{i, j}) \coloneqq Q(F) - E$ and $E$ is the identity matrix.
We have $\det A = 0$ because the sum of each row of $A$ equals $0$:
for any $i \in [F]$, we have
\begin{align*}
\sum_{j \in [F]} A_{i, j} 
&= \sum_{j \in [F]} (Q_{i, j}(F) - \delta_{ij})\\
&= \sum_{j \in [F]} Q_{i, j}(F) - \sum_{j \in [F]} \delta_{ij}\\
&= \sum_{j \in [F]} Q_{i, j}(F) - 1\\
&= \sum_{j \in [F]}\sum_{s \in \mathcal{S}, \trans_i(s) = j} \mu(s) - 1\\
&= \sum_{s \in \mathcal{S}} \mu(s) - 1\\
&= 0,
\end{align*}
where $\delta_{ij}$ denotes Kronecker delta.
Thus, the dimension of the null space of $A$ is greater than or equal to $1$.
In particular, Equation (\ref{eq:weaw3ce74e22}), which is equivalent to (\ref{eq:stationary1}), has a non-trivial solution $\pmb{\pi} \neq \pmb{0}$.
We choose such $\pmb{\pi} = (\pi_0, \pi_1, \ldots, \pi_{|F|-1}) \neq \pmb{0}$.
Then both of $\mathcal{I}_+ \coloneqq \{i \in [F] : \pi_i > 0\}$ and $\mathcal{I}_- \coloneqq \{i \in [F] : \pi_i < 0\}$ are closed by Lemma \ref{lem:kernel02}.
Hence, we have
\begin{align}
\label{eq:xbtsa8mbdrio}
{}^{\forall}i \in \mathcal{I}_+; {}^{\forall} j \in [F] \setminus \mathcal{I}_{+}; & Q_{i, j}(F) = 0, \\
\label{eq:equmcbqpefoo}
{}^{\forall}i \in \mathcal{I}_-; {}^{\forall} j \in [F] \setminus \mathcal{I}_{-}; & Q_{i, j}(F) = 0.
\end{align}

Since $\pmb{\pi} \neq \pmb{0}$, we have $\sum_{i \in [F]} |\pi_i| > 0$ and thus we can define $\pmb{\pi}' = (\pi'_0, \pi'_1, \ldots, \pi'_{|F|-1}) \in \mathbb{R}^{|F|}$ as
\begin{equation}
\label{eq:rj3ye3n6tg30}
\pi'_i = \frac{|\pi_i|}{\sum_{i \in [F]} |\pi_i|}
\end{equation}
for $i \in [F]$.
This vector $\pmb{\pi}'$ is a desired stationary distribution of $F$.
In fact, by the definition, $\pmb{\pi}'$ clearly satisfies (\ref{eq:stationary2}) and $\pi_i' \geq 0$ for any $i \in [F]$.
Also, we can see that $\pmb{\pi}'$ satisfies (\ref{eq:stationary1}) because
for any $j \in [F]$, we have
\begin{align*}
\lefteqn{\Big(\sum_{i \in [F]} |\pi_i|\Big) \Big(\sum_{i \in [F]} \pi'_i Q_{i, j}(F) \Big)}\\
&\overset{(\mathrm{A})}{=} \sum_{i \in [F]}  |\pi_i|Q_{i, j}(F)\\
&=  \sum_{i \in \mathcal{I}_+}  \pi_i Q_{i, j}(F)- \sum_{i \in \mathcal{I}_-} \pi_i Q_{i, j}(F) \\
&\overset{(\mathrm{B})}{=}  \begin{cases}
\sum_{i \in \mathcal{I}_+} \pi_i Q_{i, j}(F)  &\,\,\text{if}\,\, j \in \mathcal{I}_+,\\
- \sum_{i \in \mathcal{I}_-} \pi_i Q_{i, j}(F)  &\,\,\text{if}\,\, j \in \mathcal{I}_-,\\
0 &\,\,\text{otherwise},\\
\end{cases}\\
&\overset{(\mathrm{C})}{=}  \begin{cases}
\sum_{i \in \mathcal{I}_+} \pi_iQ_{i, j}(F)  + \sum_{i \in \mathcal{I}_-} \pi_iQ_{i, j}(F)  &\,\,\text{if}\,\, j \in \mathcal{I}_+,\\
-\sum_{i \in \mathcal{I}_+} \pi_iQ_{i, j}(F)  - \sum_{i \in \mathcal{I}_-} \pi_iQ_{i, j}(F)  &\,\,\text{if}\,\, j \in \mathcal{I}_-,\\
0 &\,\,\text{otherwise},\\
\end{cases}\\
&=  \begin{cases}
\sum_{i \in [F]} \pi_iQ_{i, j}(F)  &\,\,\text{if}\,\, j \in \mathcal{I}_+,\\
-\sum_{i \in [F]} \pi_iQ_{i, j}(F) &\,\,\text{if}\,\, j \in \mathcal{I}_-,\\
0 &\,\,\text{otherwise},\\
\end{cases}\\
&\overset{(\mathrm{D})}{=}  \begin{cases}
\pi_j &\,\,\text{if}\,\, j \in \mathcal{I}_+,\\
- \pi_j &\,\,\text{if}\,\, j \in \mathcal{I}_-,\\
0 &\,\,\text{otherwise},\\
\end{cases}\\
&= |\pi_j|\\
&\overset{(\mathrm{E})}{=} \Big(\sum_{i \in [F]} |\pi_i|\Big) \pi'_j,
\end{align*}
where
(A) follows from (\ref{eq:rj3ye3n6tg30}),
(B) follows from (\ref{eq:xbtsa8mbdrio}) and (\ref{eq:equmcbqpefoo}),
(C) follows from (\ref{eq:xbtsa8mbdrio}) and (\ref{eq:equmcbqpefoo}),
(D) follows since $\pmb{\pi}$ is a stationary distribution of $F$,
and (E) follows from (\ref{eq:rj3ye3n6tg30}).
\end{proof}

\subsection{Proof of Lemma \ref{lem:duplicate}}
\label{subsec:proof-duplicate}

\begin{proof}[Proof of Lemma \ref{lem:duplicate}]
(Proof of (i)):
We first show that $f'^{\ast}_i(\pmb{x}) = f^{\ast}_{\varphi(i)}(\pmb{x})$ for any $i \in [F']$ and $\pmb{x} \in \mathcal{S}^{\ast}$ by induction for $|\pmb{x}|$.
For the base case $|\pmb{x}| = 0$, we have $f'^{\ast}_i(\lambda) = \lambda = f^{\ast}_{\varphi(i)}(\lambda)$ by (\ref{eq:fstar}).
We consider the induction step for $|\pmb{x}| \geq 1$.
We have
\begin{align*}
f'^{\ast}_i(\pmb{x})
&\overset{(\mathrm{A})}{=} f'_i(x_1)f'^{\ast}_{\trans'_i(x_1)}(\suff(\pmb{x}))\\
&\overset{(\mathrm{B})}{=} f_{\varphi(i)}(x_1)f'^{\ast}_{\trans'_i(x_1)}(\suff(\pmb{x}))\\
&\overset{(\mathrm{C})}{=} f_{\varphi(i)}(x_1)f^{\ast}_{\varphi(\trans'_i(x_1))}(\suff(\pmb{x}))\\
&\overset{(\mathrm{D})}{=} f_{\varphi(i)}(x_1)f^{\ast}_{\trans_{\varphi(i)}(x_1)}(\suff(\pmb{x}))\\
&\overset{(\mathrm{E})}{=} f^{\ast}_{\varphi(i)}(\pmb{x})
\end{align*}
as desired, where
(A) follows from (\ref{eq:fstar}),
(B) follows from (\ref{eq:phi1}),
(C) follows from the induction hypothesis,
(D) follows from (\ref{eq:phi2}),
and (E) follows from (\ref{eq:fstar}).

Next, we show that $\varphi(\trans'^{\ast}_i(\pmb{x})) = \trans^{\ast}_{\varphi(i)}(\pmb{x})$ for any $i \in [F']$ and $\pmb{x} \in \mathcal{S}^{\ast}$ by induction for $|\pmb{x}|$.
For the base case $|\pmb{x}| = 0$, we have $\varphi(\trans'^{\ast}_i(\lambda)) = \varphi(i) = \trans^{\ast}_{\varphi(i)}(\lambda)$ by (\ref{eq:tstar}).
We consider the induction step for $|\pmb{x}| \geq 1$.
We have
\begin{align*}
\varphi(\trans'^{\ast}_i(\pmb{x}))
&\overset{(\mathrm{A})}{=} \varphi(\trans'^{\ast}_{\trans'_i(x_1)}(\suff(\pmb{x})))
\overset{(\mathrm{B})}{=} \trans^{\ast}_{\varphi(\trans'_i(x_1))}(\suff(\pmb{x}))\\
&\overset{(\mathrm{C})}{=} \trans^{\ast}_{\trans_{\varphi(i)}(x_1)}(\suff(\pmb{x}))
\overset{(\mathrm{D})}{=} \trans^{\ast}_{\varphi(i)}(\pmb{x})
\end{align*}
as desired, where
(A) follows from (\ref{eq:tstar}),
(B) follows from the induction hypothesis,
(C) follows from (\ref{eq:phi2}),
and (D) follows from (\ref{eq:tstar}).

(Proof of (ii)):
For any $\pmb{c} \in \mathcal{C}^{\ast}$, we have
\begin{align*}
\pmb{c} \in \PREF^{\ast}_{F', i}(\pmb{b})
&\overset{(\mathrm{A})}{\iff} {}^{\exists}\pmb{x} \in \mathcal{S}^{+}; (f'^{\ast}_i(\pmb{x}) \succeq \pmb{b}\pmb{c}, f'_i(x_1) \succeq \pmb{b})\\
&\overset{(\mathrm{B})}{\iff} {}^{\exists}\pmb{x} \in \mathcal{S}^{+}; (f^{\ast}_{\varphi(i)}(\pmb{x}) \succeq \pmb{b}\pmb{c}, f_{\varphi(i)}(x_1) \succeq \pmb{b})\\
&\overset{(\mathrm{C})}{\iff} \pmb{c} \in \PREF^{\ast}_{F, \varphi(i)}(\pmb{b}),
\end{align*}
where
(A) follows from (\ref{eq:pref1}),
(B) follows from (i) of this lemma,
and (C) follows from (\ref{eq:pref1}).
This shows that $\PREF^{\ast}_{F', i}(\pmb{b}) = \PREF^{\ast}_{F, \varphi(i)}(\pmb{b})$.
We can prove $\bar{\PREF}^{\ast}_{F', i}(\pmb{b}) = \bar{\PREF}^{\ast}_{F, \varphi(i)}(\pmb{b})$ by the same way using (\ref{eq:pref2}).

(Proof of (iii)):
For any $i' \in [F']$ and $j \in [F]$, we have
\begin{align}
\sum_{j' \in \mathcal{A}_j} Q_{i', j'}(F')
&\overset{(\mathrm{A})}{=} \sum_{j' \in \mathcal{A}_j} \sum_{\substack{s \in \mathcal{S}\\ \trans'_{i'}(s) = j'}} \mu(s)
= \sum_{\substack{s \in \mathcal{S}\\ \trans'_{i'}(s) \in \mathcal{A}_j}} \mu(s)\nonumber\\
&\overset{(\mathrm{B})}{=} \sum_{\substack{s \in \mathcal{S}\\ \varphi(\trans'_{i'}(s)) = j}} \mu(s)
\overset{(\mathrm{C})}{=} \sum_{\substack{s \in \mathcal{S}\\ \trans_{\varphi(i')}(s) = j}} \mu(s)\nonumber\\
&\overset{(\mathrm{D})}{=} Q_{\varphi(i'), j}(F)
= Q_{i, j}(F)  \label{eq:np8qbt6u0hqz},
\end{align}
where $i \coloneqq \varphi(i')$ and
(A) follows from (\ref{eq:9x9htdrx1001}),
(B) follows from (\ref{eq:twqab1injvre}),
(C) follows from (\ref{eq:phi2}),
and (D) follows from (\ref{eq:9x9htdrx1001}).

Thus, for any $j \in [F]$, we have
\begin{align}
\pi_j &= \sum_{j' \in \mathcal{A}_j} \pi'_{j'}\nonumber\\
&\overset{(\mathrm{A})}{=} \sum_{j' \in \mathcal{A}_j} \sum_{i' \in [F']} \pi'_{i'}Q_{i', j'}(F')\nonumber\\
&=\sum_{j' \in \mathcal{A}_j} \sum_{i \in [F]} \sum_{i' \in \mathcal{A}_i} \pi'_{i'}Q_{i', j'}(F')\nonumber\\
&= \sum_{i \in [F]} \sum_{i' \in \mathcal{A}_i} \pi'_{i'} \sum_{j' \in \mathcal{A}_j} Q_{i', j'}(F')\nonumber\\
&\overset{(\mathrm{B})}{=} \sum_{i \in [F]} \sum_{i' \in \mathcal{A}_i} \pi'_{i'} Q_{i, j}(F)\nonumber\\
&= \sum_{i \in [F]} Q_{i, j}(F) \sum_{i' \in \mathcal{A}_i} \pi'_{i'} \nonumber\\
&= \sum_{i \in [F]} Q_{i, j}(F) \pi_i \label{eq:kjjrie6kps13},
\end{align}
where
(A) follows since $\pmb{\pi}'$ satisfies (\ref{eq:stationary1}),
and (B) follows from (\ref{eq:np8qbt6u0hqz}) and $i' \in \mathcal{A}_i$.

Also, we have
\begin{align}
\label{eq:p925jvkoor03}
\sum_{i \in [F]} \pi_i
= \sum_{i \in [F]} \sum_{i' \in \mathcal{A}_i} \pi'_{i'}
= \sum_{i' \in [F']} \pi'_{i'}
&\overset{(\mathrm{A})}{=} 1,
\end{align}
where (A) follows since $\pmb{\pi}'$ satisfies (\ref{eq:stationary2}).
By (\ref{eq:kjjrie6kps13}) and (\ref{eq:p925jvkoor03}), $\pmb{\pi}$ is a stationary distribution of $F$.

(Proof of (iv)):
We have
\begin{align*}
F \in \mathscr{F}_{\ext}
&\iff {}^{\forall}i \in [F]; \PREF^1_{F, i} \neq \emptyset\\
&\implies {}^{\forall}i' \in [F']; \PREF^1_{F, \varphi(i')} \neq \emptyset\\
&\overset{(\mathrm{A})}{\iff} {}^{\forall}i' \in [F']; \PREF^1_{F', i'} \neq \emptyset\\
&\iff F' \in \mathscr{F}_{\ext},
\end{align*}
where (A) follows from (ii) of this lemma.

(Proof of (v)):
By $F, F' \in \mathscr{F}_{\reg}$, the code-tuples $F$ and $F'$ have the unique stationary distributions $\pmb{\pi}(F)$ and $\pmb{\pi}(F')$, respectively.
By (iii) of this lemma, we have
\begin{equation}
\label{eq:npipchud5ezq}
{}^{\forall} j \in [F]; \pi_j(F) = \sum_{j' \in \mathcal{A}_j} \pi_{j'}(F'),
\end{equation}
where 
\begin{equation}
\label{eq:kr0b6tszizg4}
\mathcal{A}_i \coloneqq \{i' \in [F'] : \varphi(i') = i\}
\end{equation}
for $i \in [F]$.
Therefore, we have
\begin{align*}
L(F') &= \sum_{i' \in [F']} \pi_{i'}(F') L_{i'}(F')\\
&= \sum_{i \in [F]} \sum_{i' \in \mathcal{A}_i} \pi_{i'}(F') L_{i'}(F')\\
&\overset{(\mathrm{A})}{=} \sum_{i \in [F]} \sum_{i' \in \mathcal{A}_i} \pi_{i'}(F') L_{\varphi(i')}(F)\\
&\overset{(\mathrm{B})}{=} \sum_{i \in [F]} \sum_{i' \in \mathcal{A}_i} \pi_{i'}(F') L_i(F)\\
&= \sum_{i \in [F]}  L_i(F) \sum_{i' \in \mathcal{A}_i} \pi_{i'}(F')\\
&\overset{(\mathrm{C})}{=} \sum_{i \in [F]} \pi_i(F) L_i(F)\\
&=  L(F)
\end{align*}
as desired, where
(A) follows from (\ref{eq:phi1}) (cf. Remark \ref{rem:transitionprobability}),
(B) follows from (\ref{eq:kr0b6tszizg4}) and $i' \in \mathcal{A}_i$,
and (C) follows from (\ref{eq:npipchud5ezq}).

(Proof of (vi)):
For any $i \in [F']$ and $s \in \mathcal{S}$, we have
\begin{align*}
\PREF^k_{F', \trans'_i(s)} \cap \bar{\PREF}^k_{F', i}(f'_i(s))
&\overset{(\mathrm{A})}{=} \PREF^k_{F, \varphi(\trans'_i(s))} \cap \bar{\PREF}^k_{F, \varphi(i)}(f'_i(s))\\
&\overset{(\mathrm{B})}{=} \PREF^k_{F, \trans_{\varphi(i)}(s)} \cap \bar{\PREF}^k_{F, \varphi(i)}(f_{\varphi(i)}(s))\\
&\overset{(\mathrm{C})}{=} \emptyset,
\end{align*}
where
(A) follows from (ii) of this lemma,
(B) follows from (\ref{eq:phi1}) and (\ref{eq:phi2}),
and (C) follows from $F \in \mathscr{F}_{k\hdec}$.
Namely, $F'$ satisfies Definition \ref{def:k-bitdelay} (i).

Choose $i \in [F']$ and $s, s' \in \mathcal{S}$ such that $s \neq s'$ and $f'_i(s) = f'_i(s')$ arbitrarily.
Then by (\ref{eq:phi1}), we have
\begin{equation}
\label{eq:ezqwi02upw8m}
f_{\varphi(i)}(s) = f'_i(s) = f'_i(s') = f_{\varphi(i)}(s').
\end{equation}
Thus, we obtain
\begin{align*}
\PREF^k_{F', \trans'_i(s)} \cap \PREF^k_{F', \trans'_i(s')}
&\overset{(\mathrm{A})}{=} \PREF^k_{F, \varphi(\trans'_i(s))} \cap \PREF^k_{F, \varphi(\trans'_i(s'))}
\overset{(\mathrm{B})}{=} \PREF^k_{F, \trans_{\varphi(i)}(s)} \cap \PREF^k_{F, \trans_{\varphi(i)}(s')}
\overset{(\mathrm{C})}{=} \emptyset,
\end{align*}
where
(A) follows from (ii) of this lemma,
(B) follows from (\ref{eq:phi2}),
(C) follows from (\ref{eq:ezqwi02upw8m}) and $F \in \mathscr{F}_{k\hdec}$.
Namely, $F'$ satisfies Definition \ref{def:k-bitdelay} (ii).
\end{proof}

\subsection{Proof of Lemma \ref{lem:kernel}}
\label{subsec:proof-kernel}

To prove Lemma \ref{lem:kernel}, we first prove the following Lemmas \ref{lem:closed2}--\ref{lem:kernel0}.
Lemmas \ref{lem:closed2} and \ref{lem:closed} relate to closed sets defined in Appendix \ref{subsec:proof-stationary}.

\begin{lemma}
\label{lem:closed2}
For any $F \in \mathscr{F}$, the following statements (i) and (ii) hold.
\begin{itemize}
\item[(i)] $\kernel_F$ is closed.
\item[(ii)] For any non-empty closed set $\mathcal{I} \subseteq [F]$, we have $\kernel_F \subseteq \mathcal{I}$.
\end{itemize}
\end{lemma}

\begin{proof}[Proof of Lemma \ref{lem:closed2}]
(Proof of (i)):
Choose $i \in \kernel_F$ and $s \in \mathcal{S}$ arbitrarily.
For any $j \in [F]$, there exists $\pmb{x} \in \mathcal{S}^{\ast}$ such that $\trans^{\ast}_j(\pmb{x}) = i$, which leads to
\begin{equation}
\trans^{\ast}_j(\pmb{x}s) \overset{(\mathrm{A})}{=} \trans_{\trans^{\ast}_j(\pmb{x})}(s) = \trans_i(s),
\end{equation}
where (A) follows from Lemma \ref{lem:f_T} (ii).
This shows $\trans_i(s) \in \kernel_F$.

(Proof of (ii)):
Choose $i \in \kernel_F$ arbitrarily.
We prove $i \in \mathcal{I}$ by contradiction assuming the contrary $i \not \in \mathcal{I}$.
Since $\mathcal{I} \neq \emptyset$, we can choose $j \in \mathcal{I}$.
By $i \in \kernel_F$, there exists $\pmb{x} = x_1x_2\ldots x_n \in \mathcal{S}^{\ast}$ such that $\trans^{\ast}_j(\pmb{x}) = i$.
We define $i_l \coloneqq \trans^{\ast}_j(x_1x_2\ldots x_l)$ for $l = 0, 1, 2, \ldots, n$.
Since $i_0 = \trans^{\ast}_j(\lambda) = j \in \mathcal{I}$ and $i_n = \trans^{\ast}_j(\pmb{x}) = i \not \in \mathcal{I}$,
there exists an integer $0 \leq l < n$ such that $i_l \in \mathcal{I}$ and $i_{l+1} = \trans_{i_l}(x_{l+1}) \not\in \mathcal{I}$.
This conflicts with that $\mathcal{I}$ is closed.
\end{proof}

\begin{lemma}
\label{lem:closed}
For any $F \in \mathscr{F}$ and non-empty closed set $\mathcal{I} \subseteq [F]$, the following statements (i) and (ii) hold.
\begin{itemize}
\item[(i)] There exist $F' \in \mathscr{F}^{(|\mathcal{I}|)}$ and an injective homomorphism $\varphi : [F'] \rightarrow [F]$ from $F'$ to $F$ such that $\mathcal{I} = \varphi([F']) \coloneqq \{\varphi(i) : i \in [F']\}$.
\item[(ii)] There exists a stationary distribution $\pmb{\pi} = (\pi_0, \allowbreak \pi_1, \ldots, \pi_{|F|-1})$ of $F$ such that $\pi_i = 0$ for any $i \in [F] \setminus \mathcal{I}$.
\end{itemize}
\end{lemma}

\begin{proof}[Proof of Lemma \ref{lem:closed}]
(Proof of (i)):
Suppose $\mathcal{I} = \{i_0, i_1, \ldots, i_{m-1}\}$, where $i_0 < i_1 < \cdots < i_{m-1}$ and $m = |\mathcal{I}|$.
We define a mapping $\varphi : [m] \rightarrow [F]$ as $\varphi(j) = i_j$ for $j \in [m]$.
Since $\varphi$ is injective and $\varphi([m]) = \mathcal{I}$, we can consider the inverse mapping $\varphi^{-1} : \mathcal{I} \rightarrow [m]$, which maps $\varphi(i)$ to $i$ for any $i \in [m]$.
Also, we define $F'(f', \trans') \in \mathscr{F}^{(m)}$ as
\begin{align}
f'_i(s) &= f_{\varphi(i)}(s), \label{eq:8dpb9trty1gg}\\
\trans'_i(s) &= \varphi^{-1}(\trans_{\varphi(i)}(s)) \label{eq:3l3b00el4n7f}
\end{align}
for $i \in [F']$ and $s \in \mathcal{S}$.
Since $\mathcal{I}$ is closed, we have $\trans_{\varphi(i)}(s) \in \mathcal{I}$ and thus $\trans'_i(s) = \varphi^{-1}(\trans_{\varphi(i)}(s)) \in [m] = [F']$;
that is, $F'$ is indeed well-defined.
We can see that $\varphi$ is a homomorphism from $F'$ to $F$ directly from (\ref{eq:8dpb9trty1gg}) and (\ref{eq:3l3b00el4n7f}).

(Proof of (ii)):
By (i) of this lemma, there exist $F' \in \mathscr{F}$ and an injective homomorphism $\varphi : [F'] \rightarrow [F]$  from $F'$ to $F$ such that
\begin{equation}
\label{eq:cjoobaqo5l8x}
\varphi([F']) = \mathcal{I}.
\end{equation}
By Lemma \ref{lem:stationary}, we can choose a stationary distribution $\pmb{\pi}'$ of $F'$.
By Lemma \ref{lem:duplicate} (iii), the vector $\pmb{\pi} \in \mathbb{R}^{|F|}$ defined as (\ref{eq:3720vpgjeuph}) is a stationary distribution of $F$.
This vector $\pmb{\pi}$ is a desired stationary distribution because $\mathcal{A}_i = \{i' \in [F'] : \varphi(i') = i\} = \emptyset$ holds for any $i \in [F] \setminus \mathcal{I}$ by (\ref{eq:cjoobaqo5l8x}).
\end{proof}

\begin{lemma}
\label{lem:kernel0}
For any $F \in \mathscr{F}$,  If $\kernel_{F} = \emptyset$, then there exist $p, q \in [F]$ such that $\mathcal{I}_p \cap \mathcal{I}_q = \emptyset$,
where $\mathcal{I}_i \coloneqq \{\trans^{\ast}_i(\pmb{x}) : \pmb{x} \in \mathcal{S}^{\ast}\}$ for $i \in [F]$.
\end{lemma}

\begin{proof}[Proof of Lemma \ref{lem:kernel0}]
We first show that for any $i, j \in [F]$, we have
\begin{equation}
\label{eq:u1kenn7zebgw}
j \in \mathcal{I}_i \implies \mathcal{I}_j \subseteq \mathcal{I}_i.
\end{equation}
Assume $j \in \mathcal{I}_i$ and choose $p \in \mathcal{I}_j$ arbitrarily.
Then there exists $\pmb{x} \in \mathcal{S}^{\ast}$ such that $\trans^{\ast}_j(\pmb{x}) = p$.
Also, by $j \in \mathcal{I}_i$, there exists $\pmb{y} \in \mathcal{S}^{\ast}$ such that $\trans^{\ast}_i(\pmb{y}) = j$.
Therefore, we have
\begin{equation*}
\trans^{\ast}_i(\pmb{yx})
\overset{(\mathrm{A})}{=} \trans^{\ast}_{\trans^{\ast}_i(\pmb{y})}(\pmb{x})
= \trans^{\ast}_j(\pmb{x})
= p,
\end{equation*}
where (A) follows from Lemma \ref{lem:f_T} (ii).
This leads to $p \in \mathcal{I}_i$ and thus we obtain (\ref{eq:u1kenn7zebgw}).

Now, we prove Lemma \ref{lem:kernel0} by proving its contraposition.
Namely, we show $\kernel_{F} \neq \emptyset$ assuming that
\begin{equation}
\label{eq:ovsyl7xqcns0}
{}^{\forall}i, j \in [F]; \mathcal{I}_i \cap \mathcal{I}_j \neq \emptyset.
\end{equation}
We can see that
\begin{equation*}
\kernel_F = \bigcap_{i \in [F]} \mathcal{I}_i
\end{equation*}
because for any $j \in [F]$, it holds that
\begin{align*}
j \in \bigcap_{i \in [F]} \mathcal{I}_i
&\iff {}^{\forall} i \in [F]; j \in \mathcal{I}_i \\
&\iff {}^{\forall} i \in [F]; {}^{\exists} \pmb{x} \in \mathcal{S}^{\ast}; \trans^{\ast}_i(\pmb{x}) = j\\
&\iff j \in \kernel_F.
\end{align*}
Thus, to show $\kernel_{F} \neq \emptyset$, it suffices to show that \begin{equation}
\label{eq:wpjhflge5ppg}
\bigcap_{i \in [r]} \mathcal{I}_i \neq \emptyset
\end{equation}
for any $r = 1, 2, \ldots, |F|$ since the case $r = |F|$ gives the desired result.

We prove (\ref{eq:wpjhflge5ppg}) by induction for $r$.
The base case $r = 1$ is trivial since $\mathcal{I}_0 \owns 0$.
We consider the induction step for $r \geq 2$.
By the induction hypothesis, we have $\bigcap_{i \in [r-1]} \mathcal{I}_i \neq \emptyset$.
Therefore, we can choose $j \in [F]$ such that $j \in \mathcal{I}_i$ for any $i \in [r-1]$.
By (\ref{eq:u1kenn7zebgw}), we have $\mathcal{I}_j \subseteq \mathcal{I}_i$ for any $i \in [r-1]$ and thus
\begin{equation}
\label{eq:1ymkzi1xamyn}
\mathcal{I}_j \subseteq \bigcap_{i \in [r-1]} \mathcal{I}_i.
\end{equation}
Hence, we obtain
\begin{equation*}
\bigcap_{i \in [r]} \mathcal{I}_i
= \Big( \bigcap_{i \in [r-1]} \mathcal{I}_i \Big) \cap \mathcal{I}_{r-1}
\overset{(\mathrm{A})}{\supseteq} \mathcal{I}_j \cap \mathcal{I}_{r-1}
\overset{(\mathrm{B})}{\neq} \emptyset
\end{equation*}
as desired, where
(A) follows from (\ref{eq:1ymkzi1xamyn}),
and (B) follows from (\ref{eq:ovsyl7xqcns0}).
\end{proof}

\begin{proof}[Proof of Lemma \ref{lem:kernel}]
(Proof of (i)):
(Necessity)
We assume $\kernel_{F} = \emptyset$ and show that $F$ has two distinct stationary distributions.
By Lemma \ref{lem:kernel0}, we can choose $p, q \in [F]$ such that
\begin{equation}
\label{eq:e31o1qy8dxd3}
\mathcal{I}_p \cap \mathcal{I}_q = \emptyset.
\end{equation}
We can see that $\mathcal{I}_p$ is not empty since $\mathcal{I}_p \owns p$ and also see that $\mathcal{I}_p$ is closed because for any $i \in \mathcal{I}_p$, we have
\begin{equation}
\{\trans_i(s) : s \in \mathcal{S}\} \subseteq \{\trans^{\ast}_i(\pmb{x}) : \pmb{x} \in \mathcal{S}^{\ast}\}
= \mathcal{I}_i \overset{(\mathrm{A})}{\subseteq} \mathcal{I}_p,
\end{equation}
where (A) follows from (\ref{eq:u1kenn7zebgw}).
By the same argument, also $\mathcal{I}_q$ is a non-empty closed set.
Therefore, by Lemma \ref{lem:closed} (ii), there exist stationary distributions $\pmb{\pi} = (\pi_0, \pi_1, \ldots, \pi_{|F|-1})$ and $\pmb{\pi}' = (\pi'_0, \pi'_1, \ldots, \pi'_{|F|-1})$ of $F$ such that
\begin{equation}
\label{eq:xe6uq39oql3h}
{}^{\forall} i \in [F] \setminus \mathcal{I}_p; \pi_i = 0
\end{equation}
and
\begin{equation}
\label{eq:s2l44qnbjz4j}
{}^{\forall} i \in [F] \setminus \mathcal{I}_q; \pi'_i = 0.
\end{equation}
Since $\pmb{\pi}$ satisfies (\ref{eq:stationary2}), we have $\pi_j > 0$ for some $j \in [F]$.
By (\ref{eq:xe6uq39oql3h}) and (\ref{eq:e31o1qy8dxd3}), it must hold that $j \in \mathcal{I}_p \subseteq [F] \setminus \mathcal{I}_q$.
Hence, we obtain $\pi'_j = 0 < \pi_j$ by (\ref{eq:s2l44qnbjz4j}).
This shows $\pmb{\pi} \neq \pmb{\pi}'$.
Therefore, we conclude that $F$ has two distinct stationary distributions as desired.

(Sufficiency)
We prove $\kernel_{F} = \emptyset$ assuming that there exist two distinct stationary distributions $\pmb{\pi}  = (\pi_0, \pi_1, \ldots, \pi_{|F|-1})$ and $\pmb{\pi}'  = (\pi'_0, \pi'_1, \ldots, \pi'_{|F|-1})$ of $F$.
Then $\pmb{x} = (x_0, x_1, \ldots, x_{|F|-1}) \coloneqq \pmb{\pi} - \pmb{\pi}' \neq \pmb{0}$ satisfies
 \begin{equation}
 \label{eq:stationary3}
\pmb{x}Q(F) = \pmb{\pi}Q(F) - \pmb{\pi}'Q(F) \overset{(\mathrm{A})}{=} \pmb{\pi} - \pmb{\pi}' = \pmb{x},
\end{equation}
\begin{equation}
\label{eq:stationary4}
\sum_{i \in [F]} x_i =  \sum_{i \in [F]} \pi_i - \sum_{i \in [F]}  \pi'_i \overset{(\mathrm{B})}{=}  1 - 1 = 0,
\end{equation}
where
(A) follows from (\ref{eq:stationary1}),
and (B) follows from (\ref{eq:stationary2}).
 Thus, by $\pmb{x} \neq \pmb{0}$ and (\ref{eq:stationary4}), both of $\mathcal{I}_+ \coloneqq \{i \in [F] : x_i > 0\}$ and $\mathcal{I}_- \coloneqq \{i \in [F] : x_i < 0\}$ are non-empty sets.
Also, both of $\mathcal{I}_+$ and $\mathcal{I}_-$ are closed by (\ref{eq:stationary3}) and Lemma \ref{lem:kernel02} stated in Appendix \ref{subsec:proof-stationary}.
Therefore, by Lemma \ref{lem:closed2} (ii), we obtain $\kernel_F \subseteq \mathcal{I}_+$ and $\kernel_F \subseteq \mathcal{I}_-$, which conclude
$\kernel_F \subseteq \mathcal{I}_+ \cap \mathcal{I}_- = \emptyset$ as desired.

(Proof of (ii)):
We show $\kernel_{F} = \mathcal{I}_+ \coloneqq \{i \in [F] : \pi_i(F) > 0\}$.

($\kernel_{F} \subseteq \mathcal{I}_+$) By (\ref{eq:stationary2}), the set $\mathcal{I}_+$ is not empty.
Also, by (\ref{eq:stationary1}) and Lemma \ref{lem:kernel02} stated in Appendix \ref{subsec:proof-stationary}, the set $\mathcal{I}_+$ is closed.
Hence, we obtain $\kernel_{F} \subseteq \mathcal{I}_+$ by Lemma \ref{lem:closed2} (ii).

($\kernel_{F} \supseteq \mathcal{I}_+$) Since $\kernel_F$ is closed by Lemma \ref{lem:closed2} (i), we see from Lemma \ref{lem:closed} (ii) that the unique stationary distribution $\pmb{\pi}(F)$ satisfies $\pi_i(F) = 0$ for any $i \in [F] \setminus \kernel_F$.
Therefore, we obtain $\kernel_F \supseteq \mathcal{I}_+$.
\end{proof}

\subsection{Proof of Lemma \ref{lem:irr-part}}
\label{subsec:proof-irr-part}

The proof of Lemma \ref{lem:irr-part} relies on Lemmas \ref{lem:closed2} and \ref{lem:closed} stated in Appendix \ref{subsec:proof-kernel}.

\begin{proof}[Proof of Lemma \ref{lem:irr-part}]
Since $\kernel_F$ is closed by Lemma \ref{lem:closed2} (i),
we see from Lemma \ref{lem:closed} (i) that there exist $\bar{F}(\bar{f}, \bar{\trans}) \in \mathscr{F}$ and an injective homomorphism $\varphi : [\bar{F}] \rightarrow [F]$ from $F'$ to $F$ such that $\varphi([\bar{F}]) = \kernel_F$.
Now, it suffices to show $\bar{F} \in \mathscr{F}_{\irr}$.

For any $i, j \in [\bar{F}]$,
there exists $\pmb{x} \in \mathcal{S}^{\ast}$ such that
\begin{equation}
\label{eq:a5yuz7ru8p08}
\trans^{\ast}_{\varphi(i)}(\pmb{x}) = \varphi(j)
\end{equation}
by $\varphi(j) \in \varphi([\bar{F}]) = \kernel_{F}$.
Thus, for any $i, j \in [F]$, we have
\begin{align*}
\bar{\trans}^{\ast}_i(\pmb{x})
&= \varphi^{-1}(\varphi(\bar{\trans}^{\ast}_i(\pmb{x})))
\overset{(\mathrm{A})}{=} \varphi^{-1}(\trans^{\ast}_{\varphi(i)}(\pmb{x}))
\overset{(\mathrm{B})}{=} \varphi^{-1}(\varphi(j))
= j,
\end{align*}
where
(A) follows from Lemma \ref{lem:duplicate} (i),
and (B) follows from (\ref{eq:a5yuz7ru8p08}).
Therefore, $\bar{F} \in \mathscr{F}_{\irr}$ holds.
\end{proof}

\subsection{Proof of Lemma \ref{lem:pref-unchanged}}
\label{subsec:proof-pref-unchanged}

\begin{proof}[Proof of Lemma \ref{lem:pref-unchanged}]
(Proof of (i)):
We prove only $\PREF^k_{F, i}(\pmb{b}) \supseteq \PREF^k_{F', i}(\pmb{b})$ for any $i \in [F]$ and $\pmb{b} \in \mathcal{C}^{\ast}$ because we can prove $\PREF^k_{F, i}(\pmb{b}) \subseteq \PREF^k_{F', i}(\pmb{b})$, $\bar{\PREF}^k_{F, i}(\pmb{b}) \supseteq \bar{\PREF}^k_{F', i}(\pmb{b})$, and $\bar{\PREF}^k_{F, i}(\pmb{b}) \subseteq \bar{\PREF}^k_{F', i}(\pmb{b})$ in the similar way.
To prove $\PREF^k_{F, i}(\pmb{b}) \supseteq \PREF^k_{F', i}(\pmb{b})$, it suffices to prove that for any $(i, \pmb{x}, \pmb{b}, \pmb{c}) \in [F] \times \mathcal{S}^{+} \times \mathcal{C}^{\ast} \times \mathcal{C}^{\leq k}$, we have
\begin{align}
(f'^{\ast}_i(\pmb{x}) \succeq \pmb{b}\pmb{c}, f'_i(x_1) \succeq \pmb{b})
&\implies {}^{\exists}\pmb{x}' \in \mathcal{S}^{+}; (f^{\ast}_i(\pmb{x}') \succeq \pmb{b}\pmb{c}, f_i(x'_1) \succeq \pmb{b}) \label{eq:uaoj2b5ueybr}
\end{align}
because this shows that for any $i \in [F']$, $\pmb{b} \in \mathcal{C}^{\ast}$, and $\pmb{c} \in \mathcal{C}^k$, we have
\begin{align*}
\pmb{c} \in \PREF^k_{F', i}(\pmb{b})
&\overset{(\mathrm{A})}{\iff} {}^{\exists}\pmb{x} \in \mathcal{S}^{+}; (f'^{\ast}_i(\pmb{x}) \succeq \pmb{b}\pmb{c}, f_i(x_1) \succeq \pmb{b})\\
&\overset{(\mathrm{B})}{\implies} {}^{\exists}\pmb{x}' \in \mathcal{S}^{+}; (f^{\ast}_i(\pmb{x}') \succeq \pmb{b}\pmb{c}, f_i(x'_1) \succeq \pmb{b})\\
&\overset{(\mathrm{C})}{\iff} \pmb{c} \in \PREF^k_{F, i}(\pmb{b})
\end{align*}
as desired, where
(A) follows from (\ref{eq:pref1}),
(B) follows from (\ref{eq:uaoj2b5ueybr}),
and (C) follows from (\ref{eq:pref1}).

Choose $(i, \pmb{x}, \pmb{b}, \pmb{c}) \in [F] \times \mathcal{S}^{+} \times \mathcal{C}^{\ast} \times \mathcal{C}^{\leq k}$ arbitrarily and 
assume
\begin{equation}
\label{eq:fukjlkhx8fs9k}
f'^{\ast}_i(\pmb{x}) \succeq \pmb{b}\pmb{c}
\end{equation}
and
\begin{equation}
\label{eq:1yy0qhed18vf}
f'_i(x_1) \succeq \pmb{b}.
\end{equation}
Then we have
\begin{equation}
\label{eq:vl4krmm8ndjo}
f_i(x_1) \overset{(\mathrm{A})}{=} f'_i(x_1) \overset{(\mathrm{B})}{\succeq} \pmb{b},
\end{equation}
where (A) follows from the assumption (a) of this lemma,
and (B) follows from (\ref{eq:1yy0qhed18vf}).

We prove (\ref{eq:uaoj2b5ueybr}) by induction for $|\pmb{x}|$.
For the base case $|\pmb{x}| = 1$, we have
\begin{equation}
\label{eq:vvilz1z5mg2s}
f^{\ast}_i(\pmb{x}) = f_i(x_1)
\overset{(\mathrm{A})} = f'_i(x_1)
= f'^{\ast}_i(\pmb{x})
\overset{(\mathrm{B})}{\succeq}
\pmb{b}\pmb{c},
\end{equation}
where
(A) follows from the assumption (a) of this lemma,
and (B) follows from (\ref{eq:fukjlkhx8fs9k}).
By (\ref{eq:vvilz1z5mg2s}) and (\ref{eq:vl4krmm8ndjo}), the claim (\ref{eq:uaoj2b5ueybr}) holds for the base case $|\pmb{x}| = 1$.

We consider the induction step for $|\pmb{x}| \geq 2$.
We have
\begin{align}
f_i(x_1)f'^{\ast}_{\trans'_i(x_1)}(\suff(\pmb{x}))
&\overset{(\mathrm{A})}{=} f'_i(x_1)f'^{\ast}_{\trans'_i(x_1)}(\suff(\pmb{x}))
\overset{(\mathrm{B})}{=} f'^{\ast}_i(\pmb{x}) \overset{(\mathrm{C})}{\succeq} \pmb{b}\pmb{c}, \label{eq:k3v4hbuwbkgpo}
\end{align}
where
(A) follows from the assumption (a) of this lemma,
(B) follows from (\ref{eq:fstar}),
and (C) follows from (\ref{eq:fukjlkhx8fs9k}).
Therefore, $f_i(x_1) \succeq \pmb{b}\pmb{c}$ or $f_i(x_1) \prec \pmb{b}\pmb{c}$ holds.
In the case $f_i(x_1) \succeq \pmb{b}\pmb{c}$, clearly $\pmb{x}' \coloneqq x_1$ satisfies $f^{\ast}_i(\pmb{x}') \succeq \pmb{b}\pmb{c}$ and $f_i(x'_1) = f_i(x_1) \succeq \pmb{b}$ by (\ref{eq:vl4krmm8ndjo}) as desired.
Thus, now we assume $f_i(x_1) \prec \pmb{b}\pmb{c}$.
Then we have 
\begin{align}
|f_i(x_1)^{-1}\pmb{b}\pmb{c}|
&= - |f_i(x_1)| + |\pmb{b}| + |\pmb{c}|
\overset{(\mathrm{A})}{=} - |f'_i(x_1)| + |\pmb{b}| + |\pmb{c}|
\overset{(\mathrm{B})}{\leq}  |\pmb{c}| \leq k, \label{eq:qgsn4s28gp0y}
\end{align}
where
(A) follows from the assumption (a) of this lemma,
and (B) follows from (\ref{eq:1yy0qhed18vf}).
By (\ref{eq:k3v4hbuwbkgpo}), we have
\begin{equation}
\label{eq:s5c41waomab3}
f'^{\ast}_{\trans'_i(x_1)}(\suff(\pmb{x})) \succeq f_i(x_1)^{-1}\pmb{b}\pmb{c}.
\end{equation}
By (\ref{eq:qgsn4s28gp0y}) and (\ref{eq:s5c41waomab3}), we can apply the induction hypothesis to $(\trans'_i(x_1), \suff(\pmb{x}), \allowbreak \lambda, f_i(x_1)^{-1}\pmb{b}\pmb{c})$.
Hence, there exists $\pmb{x}' \in \mathcal{S}^{\ast}$ such that
$f^{\ast}_{\trans'_i(x_1)}(\pmb{x}') \succeq f_i(x_1)^{-1}\pmb{b}\pmb{c}$,
which leads to $f_i(x_1)^{-1}\pmb{b}\pmb{c} \in \PREF^{k'}_{F, \trans'_i(x_1)}$ by (\ref{eq:pref3}), where $k' \coloneqq |f_i(x_1)^{-1}\pmb{b}\pmb{c}|$.
By Lemma \ref{lem:pref-inc} (i), 
there exists $\pmb{c}' \in \mathcal{C}^{k-k'}$ such that
\begin{equation}
f_i(x_1)^{-1}\pmb{b}\pmb{c}\pmb{c}' \in \PREF^k_{F, \trans'_i(x_1)} \overset{(\mathrm{A})}{=} \PREF^k_{F, \trans_i(x_1)},
\end{equation}
where (A) follows from the assumption (b) of this lemma.
By (\ref{eq:pref3}), there exists $\pmb{x}'' \in \mathcal{S}^{\ast}$ such that 
\begin{equation}
\label{eq:o15zob8zuiw9}
f^{\ast}_{\trans_i(x_1)}(\pmb{x}'') \succeq f_i(x_1)^{-1}\pmb{b}\pmb{c}\pmb{c}' \succeq f_i(x_1)^{-1}\pmb{b}\pmb{c}.
\end{equation}
Thus, we have
\begin{align}
f^{\ast}_i(x_1\pmb{x}'')
&\overset{(\mathrm{A})}{=} f_i(x_1)f^{\ast}_{\trans_i(x_1)}(\pmb{x}'') \nonumber\\
&\overset{(\mathrm{B})}{\succeq} f_i(x_1)f_i(x_1)^{-1}\pmb{b}\pmb{c} \nonumber\\
&= \pmb{b}\pmb{c}, \label{eq:eeyn9kbb4v0o}
\end{align}
where
(A) follows from (\ref{eq:fstar}),
and (B) follows from (\ref{eq:o15zob8zuiw9}).
The induction is completed by (\ref{eq:vl4krmm8ndjo}) and (\ref{eq:eeyn9kbb4v0o}).

(Proof of (ii)): 
We have
\begin{align*}
F \in \mathscr{F}_{\ext}
&\iff {}^{\forall}i \in [F]; \PREF^1_{F, i} \neq \emptyset\\
&\overset{(\mathrm{A})}{\iff} {}^{\forall}i \in [F']; \PREF^1_{F', i} \neq \emptyset\\
&\iff F' \in \mathscr{F}_{\ext},
\end{align*}
where (A) follows from (i) of this lemma.

(Proof of (iii)): 
For any $i \in [F']$ and $s \in \mathcal{S}$, we have
\begin{align*}
\PREF^k_{F', \trans'_i(s)} \cap \bar{\PREF}^k_{F', i}(f'_i(s))
&\overset{(\mathrm{A})}{=} \PREF^k_{F, \trans'_i(s)} \cap \bar{\PREF}^k_{F, i}(f'_i(s))
\overset{(\mathrm{B})}{=} \PREF^k_{F, \trans_i(s)} \cap \bar{\PREF}^k_{F, i}(f_i(s)) 
\overset{(\mathrm{C})}{=} \emptyset,
\end{align*}
where
(A) follows from (i) of this lemma,
(B) follows from the assumptions (a) and (b),
and (C) follows from $F \in \mathscr{F}_{k\hdec}$.
Namely, $F'$ satisfies Definition \ref{def:k-bitdelay} (i).

For any $i \in [F']$ and $s, s' \in \mathcal{S}$ such that $s \neq s'$ and $f'_i(s) = f'_i(s')$, we have
\begin{equation}
\label{eq:81yuxd8qy20q}
f_i(s) = f_i(s')
\end{equation}
by the assumption (a), and we have
\begin{align*}
\PREF^k_{F', \trans'_i(s)} \cap \PREF^k_{F', \trans'_i(s')}
&\overset{(\mathrm{A})}{\subseteq} \PREF^k_{F, \trans'_i(s)} \cap \PREF^k_{F, \trans'_i(s')}
\overset{(\mathrm{B})}{=} \PREF^k_{F, \trans_i(s)} \cap \PREF^k_{F, \trans_i(s')} 
\overset{(\mathrm{C})}{=} \emptyset,
\end{align*}
where
(A) follows from (i) of this lemma,
(B) follows from the assumptions (b),
and (C) follows from $F \in \mathscr{F}_{k\hdec}$ and (\ref{eq:81yuxd8qy20q}).
Namely, $F'$ satisfies Definition \ref{def:k-bitdelay} (ii).
\end{proof}

\subsection{Proof of Lemma \ref{lem:chooseone}}
\label{subsec:proof-chooseone}

\begin{proof}[Proof of Lemma \ref{lem:chooseone}]
We show $\kernel_{F'} \owns p$ since this implies $F' \in \mathscr{F}_{\reg}$ by Lemma \ref{lem:kernel} (i).
Namely, we show that for any $j \in [F']$, there exists $\pmb{x} \in \mathcal{S}^{\ast}$ such that $\trans'^{\ast}_j(\pmb{x}) = p$.

For $j = p$, the sequence $\pmb{x} \coloneqq \lambda$ satisfies $\trans'^{\ast}_{j}(\pmb{x}) = p$ by (\ref{eq:tstar}).
Thus, we now consider the case $j \neq p$. Choose $j \in [F'] \setminus \{p\}$ arbitrarily. Since $p \in \kernel_{F}$ by $F \in \mathscr{F}_{\irr}$, there exists $\pmb{x} = x_1x_2\ldots x_n \in \mathcal{S}^{+}$ such that $\trans^{\ast}_j(\pmb{x}) = p$. Let $r \geq 1$ be the minimum positive integer such that
\begin{equation}
\label{eq:nf8031j99by9}
\trans^{\ast}_j(x_1x_2\ldots x_r) \in \mathcal{I}.
\end{equation}
Note that there exists such an integer $r \leq n$ since $\trans^{\ast}_j(\pmb{x}) = \trans^{\ast}_j(x_1x_2\ldots x_n) = p \in \mathcal{I}$.
We show that 
\begin{equation}
\label{eq:3aykcjky7gi4}
\trans'^{\ast}_j(x_1x_2\ldots x_{r'}) = \trans^{\ast}_j(x_1x_2\ldots x_{r'})
\end{equation}
for any $r' = 1, 2, \ldots, r-1$ by induction for $r'$.
For the base case $r' = 0$, we have $\trans'^{\ast}_j(\lambda) = j = \trans^{\ast}_j(\lambda)$ by (\ref{eq:tstar}).
We consider the induction step for $r' \geq 1$. We have
\begin{align*}
\trans'^{\ast}_j(x_1x_2\ldots x_{r'})
&\overset{(\mathrm{A})}{=} \trans'_{\trans'^{\ast}_j(x_1x_2\ldots x_{r'-1})}(x_{r'})\\
&\overset{(\mathrm{B})}{=} \trans'_{\trans^{\ast}_j(x_1x_2\ldots x_{r'-1})}(x_{r'})\\
&\overset{(\mathrm{C})}{=} \trans_{\trans^{\ast}_j(x_1x_2\ldots x_{r'-1})}(x_{r'})\\
&\overset{(\mathrm{D})}{=} \trans^{\ast}_j(x_1x_2\ldots x_{r'})
\end{align*}
as desired, where
(A) follows from Lemma \ref{lem:f_T} (ii),
(B) follows from the induction hypothesis,
(C) is obtained by applying the second case of (\ref{eq:hbm12zjixhzy}) since $\trans_{\trans^{\ast}_j(x_1x_2\ldots x_{r'-1})}(x_{r'}) = \trans^{\ast}_j(x_1x_2\ldots x_{r'}) \not\in \mathcal{I}$ by $r' \leq r-1$ and the minimality of $r$, 
and (D) follows from Lemma \ref{lem:f_T} (ii).

Thus, we obtain
\begin{align*}
\trans'^{\ast}_j(x_1x_2\ldots x_r)
&\overset{(\mathrm{A})}{=} \trans'_{\trans'^{\ast}_j(x_1x_2\ldots x_{r-1})}(x_r)
\overset{(\mathrm{B})}{=}\trans'_{\trans^{\ast}_j(x_1x_2\ldots x_{r-1})}(x_r) 
\overset{(\mathrm{C})}{=} p
\end{align*}
as desired, where
(A) follows from Lemma \ref{lem:f_T} (ii),
(B) follows from (\ref{eq:3aykcjky7gi4}),
and (C) follows from (\ref{eq:nf8031j99by9}) and the first case of (\ref{eq:hbm12zjixhzy}).
\end{proof}

\subsection{Proof of Lemma \ref{lem:improve}}
\label{subsec:proof-improve}

\begin{proof}[Proof of Lemma \ref{lem:improve}]
Let $p \in \argmin_{i \in [F]} (h_i(F)-h_i(F'))$.
Then it holds that
\begin{equation}
\label{eq:gapmxf1qb1k5}
{}^{\forall} i \in [F]; h_i(F') - h_p(F') \leq h_i(F) - h_p(F).
\end{equation}
We have
\begin{align}
\sum_{i \in [F]}(h_i(F) - h_p(F))Q_{p, i}(F)
&= \sum_{i \in [F]}(h_i(F) - h_p(F))\sum_{\substack{s \in \mathcal{S}\\\trans_p(s) = i}}\mu(s)\nonumber\\
&= \sum_{i \in [F]}\sum_{\substack{s \in \mathcal{S}\\\trans_p(s) = i}}(h_i(F) - h_p(F))\mu(s)\nonumber\\
&= \sum_{i \in [F]}\sum_{\substack{s \in \mathcal{S}\\\trans_p(s) = i}}(h_{\trans_p(s)}(F) - h_p(F))\mu(s)\nonumber\\
&= \sum_{s \in \mathcal{S}}(h_{\trans_p(s)}(F) - h_p(F))\mu(s). \label{eq:lsjabhzl1385}
\end{align}
Similarly, we have
\begin{align}
\sum_{i \in [F]}(h_i(F) - h_p(F))Q_{p, i}(F')
&= \sum_{s \in \mathcal{S}}(h_{\trans'_p(s)}(F) - h_p(F))\mu(s). \label{eq:yuk28f95g633}
\end{align}
Hence, we obtain
\begin{align*}
L(F') &\overset{(\mathrm{A})}{=} L_p(F') + \sum_{i \in [F]}(h_i(F') - h_p(F'))Q_{p, i}(F')\\
&\overset{(\mathrm{B})}{\leq} L_p(F') + \sum_{i \in [F]}(h_i(F) - h_p(F))Q_{p, i}(F')\\
&\overset{(\mathrm{C})}{=} L_p(F') + \sum_{s \in \mathcal{S}}(h_{\trans'_p(s)}(F) - h_p(F))\mu(s)\\
&\overset{(\mathrm{D})}{\leq} L_p(F) + \sum_{s \in \mathcal{S}}(h_{\trans_p(s)}(F) - h_p(F))\mu(s)\\
&\overset{(\mathrm{E})}{=} L_p(F) + \sum_{i \in [F]}(h_i(F) - h_p(F))Q_{p, i}(F)\\
&\overset{(\mathrm{F})}{=} L(F)
\end{align*}
as desired, where
(A) follows from (\ref{eq:potential}),
(B) follows from (\ref{eq:gapmxf1qb1k5}),
(C) follows from (\ref{eq:yuk28f95g633}),
(D) follows from the assumptions (a) and (b) of this lemma,
(E) follows from (\ref{eq:lsjabhzl1385}),
and (F) follows from (\ref{eq:potential}).
\end{proof}

\subsection{Proof of Lemma \ref{lem:fdot}}
\label{subsec:proof-fdot}

\begin{proof}[Proof of Lemma \ref{lem:fdot}]
(Proof of (i)):
We prove by induction for $|\pmb{z}|$.
For the base case $|\pmb{z}| = 0$, we have $\trans'^{\ast}_{\langle \lambda \rangle}(\lambda) = \langle \lambda \rangle$ by (\ref{eq:tstar}).
We consider the induction step for $|\pmb{z}| \geq 1$. We have
\begin{equation}
\trans'^{\ast}_{\langle \lambda \rangle}(\pmb{z})
\overset{(\mathrm{A})}{=} \trans'_{\trans'^{\ast}_{\langle \lambda \rangle}(\pref(\pmb{z}))}(z_n)
\overset{(\mathrm{B})}{=} \trans'_{\langle \pref(\pmb{z}) \rangle}(z_n)
\overset{(\mathrm{C})}{=} \langle \pmb{z} \rangle,
\end{equation}
where
$\pmb{z} = z_1z_2\ldots z_n$ and 
(A) follows from Lemma \ref{lem:f_T} (ii),
(B) follows from the induction hypothesis,
and (C) follows from the first case of (\ref{eq:je4ogqd03chp}).

(Proof of (ii)):
It suffices to show that $\langle \lambda \rangle \in \kernel_{F'}$ because
it guarantees that for any $j \in [F']$, 
there exists $\pmb{x} \in \mathcal{S}^{\ast}$ such that $\trans'^{\ast}_j(\pmb{x}) = \langle \lambda \rangle$, which leads to that for any $\pmb{z} \in \mathcal{S}^{\leq L}$, we have
\begin{equation}
\trans'^{\ast}_j(\pmb{x}\pmb{z})
\overset{(\mathrm{A})}{=} \trans'^{\ast}_{\trans'^{\ast}_j(\pmb{x})}(\pmb{z})
= \trans'^{\ast}_{\langle \lambda \rangle}(\pmb{z})
\overset{(\mathrm{B})}{=} \langle \pmb{z} \rangle
\end{equation}
as desired, where
(A) follows from Lemma \ref{lem:f_T} (ii),
and (B) follows from (i) of this lemma.

To prove $\langle \lambda \rangle \in \kernel_{F'}$, we show that there exists $\pmb{x} \in \mathcal{S}^{\ast}$ such that $\trans'^{\ast}_j(\pmb{x}) = \langle \lambda \rangle$ for the following two cases separately: (I) the case $j \in [F]$ and (II) the case $j = [F'] \setminus [F]$.

\begin{itemize}
\item[(I)] The case $j \in [F]$:
By the assumption that $p = \langle \lambda \rangle \in \kernel_F$, there exists $\pmb{x} = x_1x_2\ldots x_{n'} \in \mathcal{S}^{\ast}$ such that $\trans^{\ast}_j(\pmb{x}) = \langle \lambda \rangle$. We choose the shortest $\pmb{x}$ among such sequences.
Then we can see $\trans'^{\ast}_j(x_1x_2\ldots x_r) = \trans^{\ast}_j(x_1x_2\ldots x_r)$
for any $r = 0, 1, 2, \ldots, n$ by induction for $r$.
For the base case $r = 0$, we have $\trans'^{\ast}_j(\lambda) = j = \trans^{\ast}_j(\lambda)$ by (\ref{eq:tstar}).
We consider the induction step for $r \geq 1$. We have
\begin{align*}
\trans'^{\ast}_j(x_1x_2\ldots x_{r})
&\overset{(\mathrm{A})}{=} \trans'_{\trans'^{\ast}_j(x_1x_2\ldots x_{r-1})}(x_{r})\\
&\overset{(\mathrm{B})}{=} \trans'_{\trans^{\ast}_j(x_1x_2\ldots x_{r-1})}(x_{r})\\
&\overset{(\mathrm{C})}{=} \trans_{\trans^{\ast}_j(x_1x_2\ldots x_{r-1})}(x_{r})\\
&\overset{(\mathrm{D})}{=} \trans^{\ast}_j(x_1x_2\ldots x_{r})
\end{align*}
as desired, where
(A) follows from (\ref{eq:tstar}),
(B) follows from the induction hypothesis,
(C) follows from the third case of (\ref{eq:je4ogqd03chp}) since $\trans^{\ast}_j(x_1x_2\ldots x_{r-1}) \allowbreak \in [F] \setminus \{\langle\lambda \rangle\}$ by the definition of $\pmb{x}$, 
and (D) follows from Lemma \ref{lem:f_T} (ii).
Therefore, we obtain $\trans'^{\ast}_j(\pmb{x}) = \trans^{\ast}_j(\pmb{x}) = \langle \lambda \rangle$ as desired.

\item[(II)] The case where $j = [F'] \setminus [F]$:
Then we have $j = \pmb{z} \rangle$ for some $\pmb{z} \in \mathcal{S}^{\leq L}$.
Choose $\pmb{z'} = z'_1z'_2 \ldots z'_{n'} \in \mathcal{S}^{L - |\pmb{z}|+1}$ arbitrarily.
We have
\begin{align*}
\trans'^{\ast}_{\langle \lambda \rangle}(\pmb{z}\pmb{z}')
&\overset{(\mathrm{A})}{=} \trans'_{\trans'^{\ast}_{\langle \lambda \rangle}(\pmb{z}\pref(\pmb{z}'))}(z'_{n'})\\
&\overset{(\mathrm{B})}{=} \trans'_{\langle \pmb{z}\pref(\pmb{z}') \rangle}(z'_{n'})
\overset{(\mathrm{C})}{=} \trans^{\ast}_{\langle \lambda \rangle}(\pmb{z}\pmb{z}')\\
&= \trans^{\ast}_{|F|-1}(\pmb{z}\pmb{z}')
\in [F],
\end{align*}
where 
(A) follows from Lemma \ref{lem:f_T} (ii),
(B) follows from (i) of this lemma and $\pmb{z}\pref(\pmb{z}') \in \mathcal{S}^{\leq L}$,
and (C) follows from the second case of (\ref{eq:je4ogqd03chp}) and $\pmb{z}\pref(\pmb{z}') \in \mathcal{S}^{L}$.
Hence, by the discussion for the case (I) above, there exists $\pmb{x}' \in \mathcal{S}^\ast$ such that $\trans'^{\ast}_{\trans'^{\ast}_{\langle \lambda \rangle}(\pmb{z}\pmb{z}')}(\pmb{x}') = \langle \lambda \rangle$.
Thus, $\pmb{x} \coloneqq \pmb{z}'\pmb{x}'$ satisfies
\begin{align*}
\trans'^{\ast}_{\langle \pmb{z} \rangle}(\pmb{x})
&= \trans'^{\ast}_{\langle \pmb{z} \rangle}(\pmb{z}'\pmb{x}')
\overset{(\mathrm{A})}{=}\trans'^{\ast}_{\trans'^{\ast}_{\langle \lambda \rangle}(\pmb{z})}(\pmb{z}'\pmb{x}')
\overset{(\mathrm{B})}{=} \trans'^{\ast}_{\langle \lambda \rangle}(\pmb{z}\pmb{z}'\pmb{x}')
\overset{(\mathrm{C})}{=} \trans'^{\ast}_{\trans'^{\ast}_{\langle \lambda \rangle}(\pmb{z}\pmb{z}')}(\pmb{x}')
= \langle \lambda \rangle,
\end{align*}
where
(A) follows from (i) of this lemma,
(B) follows from Lemma \ref{lem:f_T} (ii),
and (C) follows from Lemma \ref{lem:f_T} (ii).
\end{itemize}
\end{proof}

\subsection{Proof of Lemma \ref{lem:fddot}}
\label{subsec:proof-fddot}

\begin{proof}[Proof of Lemma \ref{lem:fddot}]
(Proof of (i)):
We prove by the induction for $|\pmb{x}|$.
For the base case $|\pmb{x}| = 0$, we have $f''^{\ast}_{\langle \pmb{z} \rangle}(\lambda) = \lambda = f'^{\ast}_{\langle \pmb{z} \rangle}(\lambda)$ by (\ref{eq:fstar}).
We consider the induction step for $|\pmb{x}| \geq 1$ choosing $\pmb{z} \in \mathcal{S}^{\leq L}$ arbitrarily and dividing into the following two cases:
the case $f'^{\ast}_{\langle \lambda \rangle}(\pmb{z}) \prec \pmb{d} \preceq f'^{\ast}_{\langle \lambda \rangle}(\pmb{z}\pmb{x})$ and the other case.
\begin{itemize}
\item The case $f'^{\ast}_{\langle \lambda \rangle}(\pmb{z}) \prec \pmb{d} \preceq f'^{\ast}_{\langle \lambda \rangle}(\pmb{z}\pmb{x})$:
We consider the following two cases separately: the case $f'^{\ast}_{\langle \lambda \rangle}(\pmb{z}) \prec \pmb{d} \preceq f'^{\ast}_{\langle \lambda \rangle}(\pmb{z}x_1)$ and the case $f'^{\ast}_{\langle \lambda \rangle}(\pmb{z}x_1) \prec \pmb{d} \preceq f'^{\ast}_{\langle \lambda \rangle}(\pmb{z}\pmb{x})$.

\begin{itemize}
\item The case $f'^{\ast}_{\langle \lambda \rangle}(\pmb{z}) \prec \pmb{d} \preceq f'^{\ast}_{\langle \lambda \rangle}(\pmb{z}x_1)$:
We have
\begin{align*}
f''^{\ast}_{\langle \pmb{z} \rangle}(\pmb{x})
&\overset{(\mathrm{A})}{=} f''_{\langle \pmb{z} \rangle}(x_1)f''^{\ast}_{\langle \pmb{z}x_1 \rangle}(\suff(\pmb{x}))\\
&\overset{(\mathrm{B})}{=} f'^{\ast}_{\langle \pmb{z} \rangle}(\pmb{z})^{-1}\pref(\pmb{d})\pmb{d}^{-1}f'^{\ast}_{\langle \lambda \rangle}(\pmb{z}x_1)f''^{\ast}_{\langle \pmb{z}x_1 \rangle}(\suff(\pmb{x}))\\
&\overset{(\mathrm{C})}{=} f'^{\ast}_{\langle \pmb{z} \rangle}(\pmb{z})^{-1}\pref(\pmb{d})\pmb{d}^{-1}f'^{\ast}_{\langle \lambda \rangle}(\pmb{z}x_1)f'^{\ast}_{\langle \pmb{z}x_1 \rangle}(\suff(\pmb{x}))\\
&\overset{(\mathrm{D})}{=} f'^{\ast}_{\langle \pmb{z} \rangle}(\pmb{z})^{-1}\pref(\pmb{d})\pmb{d}^{-1}f'^{\ast}_{\langle \lambda \rangle}(\pmb{z}\pmb{x}),
\end{align*}
where
(A) follows from (\ref{eq:fstar}) and Lemma \ref{lem:fdot} (i),
(B) follows from the first case of (\ref{eq:fddot}) and $f'^{\ast}_{\langle \lambda \rangle}(\pmb{z}) \prec \pmb{d} \preceq f'^{\ast}_{\langle \lambda \rangle}(\pmb{z}x_1)$,
(C) follows from the second case of (\ref{eq:i5igfy04wlhe}) by the induction hypothesis and $f'^{\ast}_{\langle \lambda \rangle}(\pmb{z}x_1) \not\prec \pmb{d}$,
and (D) follows from (\ref{eq:fstar}).

\item The case $f'^{\ast}_{\langle \lambda \rangle}(\pmb{z}x_1) \prec \pmb{d} \preceq f'^{\ast}_{\langle \lambda \rangle}(\pmb{z}\pmb{x})$:
We have
\begin{align*}
f''^{\ast}_{\langle \pmb{z} \rangle}(\pmb{x})
&\overset{(\mathrm{A})}{=} f''_{\langle \pmb{z} \rangle}(x_1)f''^{\ast}_{\langle \pmb{z}x_1 \rangle}(\suff(\pmb{x}))\\
&\overset{(\mathrm{B})}{=} f'_{\langle \pmb{z} \rangle}(x_1)f''^{\ast}_{\langle \pmb{z}x_1 \rangle}(\suff(\pmb{x}))\\
&\overset{(\mathrm{C})}{=} f'_{\langle \pmb{z} \rangle}(x_1)f'^{\ast}_{\langle \lambda \rangle}(\pmb{z}x_1)^{-1}\pref(\pmb{d})\pmb{d}^{-1}(f'^{\ast}_{\langle \lambda \rangle}(\pmb{z}\pmb{x}))\\
&\overset{(\mathrm{D})}{=} f'^{\ast}_{\langle \pmb{z} \rangle}(\pmb{z})^{-1}\pref(\pmb{d})\pmb{d}^{-1}f'^{\ast}_{\langle \lambda \rangle}(\pmb{z}\pmb{x}),
\end{align*}where
(A) follows from (\ref{eq:fstar}) and Lemma \ref{lem:fdot} (i),
(B) follows from the second case of (\ref{eq:fddot}) since $\pmb{d} \not\preceq f'^{\ast}_{\langle \lambda \rangle}(\pmb{z}x_1)$,
(C) follows from the first case of (\ref{eq:i5igfy04wlhe}) by the induction hypothesis and $f'^{\ast}_{\langle \lambda \rangle}(\pmb{z}x_1) \prec \pmb{d} \preceq f'^{\ast}_{\langle \lambda \rangle}(\pmb{z}\pmb{x})$,
and (D) follows from (\ref{eq:fstar}).
\end{itemize}

\item The other case: We have
\begin{align*}
f''^{\ast}_{\langle \pmb{z} \rangle}(\pmb{x})
&\overset{(\mathrm{A})}{=} f''_{\langle \pmb{z} \rangle}(x_1)f''^{\ast}_{\langle \pmb{z}x_1 \rangle}(\suff(\pmb{x}))\\
&\overset{(\mathrm{B})}{=} f'_{\langle \pmb{z} \rangle}(x_1)f''^{\ast}_{\langle \pmb{z}x_1 \rangle}(\suff(\pmb{x}))\\
&\overset{(\mathrm{C})}{=} f'_{\langle \pmb{z} \rangle}(x_1)f'^{\ast}_{\langle \pmb{z}x_1 \rangle}(\suff(\pmb{x}))\\
&\overset{(\mathrm{D})}{=} f'^{\ast}_{\langle \pmb{z} \rangle}(\pmb{x}),
\end{align*}
where
(A) follows from (\ref{eq:fstar}) and Lemma \ref{lem:fdot} (i),
(B) follows from the second case of (\ref{eq:fddot}) since $f'^{\ast}_{\langle \lambda \rangle}(\pmb{z}) \prec \pmb{d} \preceq f'^{\ast}_{\langle \lambda \rangle}(\pmb{z}x_1)$ does not hold,
(C) follows from the second case of (\ref{eq:i5igfy04wlhe}) by the induction hypothesis and that $f'^{\ast}_{\langle \lambda \rangle}(\pmb{z}) \prec \pmb{d} \preceq f'^{\ast}_{\langle \lambda \rangle}(\pmb{z}x_1)$ does not hold,
and (D) follows from (\ref{eq:fstar}).
\end{itemize}

(Proof of (ii)):
Assume that
\begin{equation}
\label{eq:b4nps5s151zm}
f''_{\langle \pmb{z} \rangle}(s) \prec f''_{\langle \pmb{z} \rangle}(s').
\end{equation}
 In the case $f'^{\ast}_{\langle \lambda \rangle}(\pmb{z}) \not\prec \pmb{d}$, we have 
 \begin{equation}
 f'_{\langle \pmb{z} \rangle}(s)
\overset{(\mathrm{A})}{=} f''_{\langle \pmb{z} \rangle}(s)
\overset{(\mathrm{B})}{\prec} f''_{\langle \pmb{z} \rangle}(s')
\overset{(\mathrm{C})}{=} f'_{\langle \pmb{z} \rangle}(s')
\end{equation}
as desired, where
(A) follows from the second case of (\ref{eq:fddot}) and $f'^{\ast}_{\langle \lambda \rangle}(\pmb{z}) \not\prec \pmb{d}$,
(B) follows from (\ref{eq:b4nps5s151zm}),
and (C) follows from the second case of (\ref{eq:fddot}) and $f'^{\ast}_{\langle \lambda \rangle}(\pmb{z}) \not\prec \pmb{d}$.

We consider the case $f'^{\ast}_{\langle \lambda \rangle}(\pmb{z}) \prec \pmb{d}$ dividing into four cases by whether $\pmb{d} \preceq f'^{\ast}_{\langle \lambda \rangle}(\pmb{z}s)$ and whether $\pmb{d} \preceq f'^{\ast}_{\langle \lambda \rangle}(\pmb{z}s')$.
\begin{itemize}
\item The case $\pmb{d} \preceq f'^{\ast}_{\langle \lambda \rangle}(\pmb{z}s), \pmb{d} \preceq f'^{\ast}_{\langle \lambda \rangle}(\pmb{z}s')$:
We have 
\begin{align}
\lefteqn{f'^{\ast}_{\langle \lambda \rangle}(\pmb{z})^{-1}\pref(\pmb{d})\pmb{d}^{-1}(f'^{\ast}_{\langle \lambda \rangle}(\pmb{z})f'^{\ast}_{\langle \pmb{z} \rangle}(s))} \nonumber\\
&\overset{(\mathrm{A})}{=} f'^{\ast}_{\langle \lambda \rangle}(\pmb{z})^{-1}\pref(\pmb{d})\pmb{d}^{-1}f'^{\ast}_{\langle \lambda \rangle}(\pmb{z}s) \nonumber\\
&\overset{(\mathrm{B})}{=} f''_{\langle \pmb{z} \rangle}(s) \nonumber\\
&\overset{(\mathrm{C})}{\prec} f''_{\langle \pmb{z} \rangle}(s')  \nonumber\\
&\overset{(\mathrm{D})}{=} f'^{\ast}_{\langle \lambda \rangle}(\pmb{z})^{-1}\pref(\pmb{d})\pmb{d}^{-1}f'^{\ast}_{\langle \lambda \rangle}(\pmb{z}s') \nonumber\\
&\overset{(\mathrm{E})}{=}f'^{\ast}_{\langle \lambda \rangle}(\pmb{z})^{-1}\pref(\pmb{d})\pmb{d}^{-1}(f'^{\ast}_{\langle \lambda \rangle}(\pmb{z})f'_{\langle \pmb{z} \rangle}(s')), \label{eq:0n9mk4ye2bcs}
\end{align}
where
(A) follows from Lemma \ref{lem:f_T} (i) and Lemma \ref{lem:fdot} (i),
(B) follows from the first case of (\ref{eq:fddot}) and $\pmb{d} \preceq f'^{\ast}_{\langle \lambda \rangle}(\pmb{z}s)$, 
(C) follows from (\ref{eq:b4nps5s151zm}),
(D) follows from the first case of (\ref{eq:fddot}) and $\pmb{d} \preceq f'^{\ast}_{\langle \lambda \rangle}(\pmb{z}s')$,
and (E) follows from Lemma \ref{lem:f_T} (i) and Lemma \ref{lem:fdot} (i).
Comparing both sides of (\ref{eq:0n9mk4ye2bcs}), we obtain $f'_{\langle \pmb{z} \rangle}(s) \prec f'_{\langle \pmb{z} \rangle}(s')$ as desired.

\item The case $\pmb{d} \preceq f'^{\ast}_{\langle \lambda \rangle}(\pmb{z}s), \pmb{d} \not\preceq f'^{\ast}_{\langle \lambda \rangle}(\pmb{z}s')$:
We show that this case is impossible.
We have 
\begin{align*}
f'^{\ast}_{\langle \lambda \rangle}(\pmb{z}s')
&\overset{(\mathrm{A})}{=} f'^{\ast}_{\langle \lambda \rangle}(\pmb{z})f'_{\langle \pmb{z} \rangle}(s')\nonumber\\
&\overset{(\mathrm{B})}{=} f'^{\ast}_{\langle \lambda \rangle}(\pmb{z})f''_{\langle \pmb{z} \rangle}(s')\\
&\overset{(\mathrm{C})}{\succ} f'^{\ast}_{\langle \lambda \rangle}(\pmb{z})f''_{\langle \pmb{z} \rangle}(s)\\
&\overset{(\mathrm{D})}{=} f'^{\ast}_{\langle \lambda \rangle}(\pmb{z})f'^{\ast}_{\langle \lambda \rangle}(\pmb{z})^{-1}\pmb{d}\pref(\pmb{d})^{-1}f'^{\ast}_{\langle \lambda \rangle}(\pmb{z}s)\\
&= \pmb{d}\pref(\pmb{d})^{-1}\pmb{d}\\
&\succeq \pmb{d},
\end{align*}
where
(A) follows from Lemma \ref{lem:f_T} (i) and Lemma \ref{lem:fdot} (i),
(B) follows from the second case of (\ref{eq:fddot}) and $\pmb{d} \not\preceq f'^{\ast}_{\langle \lambda \rangle}(\pmb{z}s')$,
(C) follows from (\ref{eq:b4nps5s151zm}),
and (D) follows from the first case of (\ref{eq:fddot}) and $\pmb{d} \preceq f'^{\ast}_{\langle \lambda \rangle}(\pmb{z}s)$.
This conflicts with $\pmb{d} \not\preceq f'^{\ast}_{\langle \lambda \rangle}(\pmb{z}s')$.

\item The case $\pmb{d} \not\preceq f'^{\ast}_{\langle \lambda \rangle}(\pmb{z}s), \pmb{d} \preceq f'^{\ast}_{\langle \lambda \rangle}(\pmb{z}s')$:
 We have
\begin{align*}
f'^{\ast}_{\langle \lambda \rangle}(\pmb{z}s)
&\overset{(\mathrm{A})}{=} f'^{\ast}_{\langle \lambda \rangle}(\pmb{z})f'_{\langle \pmb{z} \rangle}(s)\\
&\overset{(\mathrm{B})}{=} f'^{\ast}_{\langle \lambda \rangle}(\pmb{z})f''_{\langle \pmb{z} \rangle}(s)\\
&\overset{(\mathrm{C})}{\prec} f'^{\ast}_{\langle \lambda \rangle}(\pmb{z})f''_{\langle \pmb{z} \rangle}(s')\\
&\overset{(\mathrm{D})}{=} f'^{\ast}_{\langle \lambda \rangle}(\pmb{z})f'^{\ast}_{\langle \lambda \rangle}(\pmb{z})^{-1}\pmb{d}\pref(\pmb{d})^{-1}f'^{\ast}_{\langle \lambda \rangle}(\pmb{z}s')\\
&= \pmb{d}\pref(\pmb{d})^{-1}f'^{\ast}_{\langle \lambda \rangle}(\pmb{z}s'),
\end{align*}
where
(A) follows from Lemma \ref{lem:f_T} (i) and Lemma \ref{lem:fdot} (i),
(B) follows from the second case of (\ref{eq:fddot}) and $\pmb{d} \not\preceq f'^{\ast}_{\langle \lambda \rangle}(\pmb{z}s)$,
(C) follows from  (\ref{eq:b4nps5s151zm}),
and (D) follows from the first case of (\ref{eq:fddot}) and $\pmb{d} \preceq f'^{\ast}_{\langle \lambda \rangle}(\pmb{z}s')$.

Therefore, we have at least one of $f'^{\ast}_{\langle \lambda \rangle}(\pmb{z}s) \prec \pmb{d}$ and $f'^{\ast}_{\langle \lambda \rangle}(\pmb{z}s) \succeq \pmb{d}$.
Since $\pmb{d} \not\preceq f'^{\ast}_{\langle \lambda \rangle}(\pmb{z}s)$, we have $f'^{\ast}_{\langle \lambda \rangle}(\pmb{z}s) \prec \pmb{d}$.
Thus, we have $f'^{\ast}_{\langle \lambda \rangle}(\pmb{z}s) \prec \pmb{d} \preceq f'^{\ast}_{\langle \lambda \rangle}(\pmb{z}s')$, which leads to $f'_{\langle \pmb{z} \rangle}(s) \prec f'_{\langle \pmb{z} \rangle}(s')$ as desired.

\item The case $\pmb{d} \not\preceq f'^{\ast}_{\langle \lambda \rangle}(\pmb{z}s), \pmb{d} \not\preceq f'^{\ast}_{\langle \lambda \rangle}(\pmb{z}s')$:
We have 
\begin{equation}
f'_{\langle \pmb{z} \rangle}(s)
\overset{(\mathrm{A})}{=} f''_{\langle \pmb{z} \rangle}(s)
\overset{(\mathrm{B})}{\prec} f''_{\langle \pmb{z} \rangle}(s')
\overset{(\mathrm{C})}{=} f'_{\langle \pmb{z} \rangle}(s')
\end{equation}
as desired, where
(A) follows from the second case of (\ref{eq:fddot}) and $\pmb{d} \not\preceq f'^{\ast}_{\langle \lambda \rangle}(\pmb{z}s)$,
(B) follows from (\ref{eq:b4nps5s151zm}),
and (C) follows from the second case of (\ref{eq:fddot}) and $\pmb{d} \not\preceq f'^{\ast}_{\langle \lambda \rangle}(\pmb{z}s')$.
\end{itemize}

(Proof of (iii)):
Choose $\pmb[x] \in \mathcal{S}^{\geq L}$ arbitrarily.
We have
\begin{align}
|f'^{\ast}_{\langle \lambda \rangle}(\pmb{x})|
\overset{(\mathrm{A})}{=} |f^{\ast}_{\langle \lambda \rangle}(\pmb{x})|
\overset{(\mathrm{B})}{\geq} \left\lfloor \frac{|\pmb{x}|}{|F|} \right\rfloor
\geq  \left\lfloor \frac{L}{|F|} \right\rfloor
\overset{(\mathrm{C})}{=} \left\lfloor \frac{|F|(|\pmb{d}|+1)}{|F|} \right\rfloor 
= |\pmb{d}|+1, \label{eq:h754dc3iyon4}
\end{align}
where
(A) follows from Lemma \ref{lem:duplicate} (i) since $\varphi$ defined in (\ref{eq:gtet1tbkkabj}) is a homomorphism from $F'$ to $F$,
(B) follows from Lemma \ref{lem:longest},
and (C) follows from the definition of $L$.
Also, we have
\begin{align*}
|f''^{\ast}_{\langle \lambda \rangle}(\pmb{x})|
&\overset{(\mathrm{A})}{\geq} \min \{|f'^{\ast}_{\langle \lambda \rangle}(\pmb{x})|, |f'^{\ast}_{\langle \lambda \rangle}(\pmb{z})^{-1}\pref(\pmb{d})\pmb{d}^{-1}(f'^{\ast}_{\langle \lambda \rangle}(\pmb{z}\pmb{x}))|\} \nonumber\\
&= \min \{ |f'^{\ast}_{\langle \pmb{z} \rangle}(\pmb{x})|, |f'^{\ast}_{\langle \pmb{z} \rangle}(\pmb{x})| - 1\} \nonumber\\
&\overset{(\mathrm{B})}{\geq} |\pmb{d}|,
\end{align*}
where
(A) follows from (i) of this lemma,
and (B) follows from (\ref{eq:h754dc3iyon4}).
\end{proof}

\subsection{Proof of Lemma \ref{lem:fddot2}}
\label{subsec:proof-fddot2}

\begin{proof}[Proof of Lemma \ref{lem:fddot2}]

(Proof of (i)):
Assume
\begin{equation}
\label{eq:ckeqzkg47tsl}
f''^{\ast}_{\langle \lambda \rangle}(\pmb{x}) \succeq \pmb{c}.
\end{equation}
We consider the following two cases separately: the case $\pmb{d} \preceq f'^{\ast}_{\langle \lambda \rangle}(\pmb{x})$ and the case $\pmb{d} \not\preceq f'^{\ast}_{\langle \lambda \rangle}(\pmb{x})$.

\begin{itemize}
\item The case $\pmb{d} \preceq f'^{\ast}_{\langle \lambda \rangle}(\pmb{x})$:
We have
\begin{align}
\label{eq:wh4n94pxvfzg}
f''^{\ast}_{\langle \lambda \rangle}(\pmb{x})
\overset{(\mathrm{A})}{=} \pref(\pmb{d}) \pmb{d}^{-1} f'^{\ast}_{\langle \lambda \rangle}(\pmb{x})
\succeq \pref(\pmb{d}),
\end{align}
where (A) follows from the first case of (\ref{eq:i5igfy04wlhe}) and $\pmb{d} \preceq f'^{\ast}_{\langle \lambda \rangle}(\pmb{x})$.
Comparing (\ref{eq:ckeqzkg47tsl}) and (\ref{eq:wh4n94pxvfzg}), we have $\pref(\pmb{d}) \succeq \pmb{c}$ since $|\pref(\pmb{d})| \geq k \geq |\pmb{c}|$.
Therefore, by $\pmb{d} \preceq f'^{\ast}_{\langle \lambda \rangle}(\pmb{x})$,
we obtain $f'^{\ast}_{\langle \lambda \rangle}(\pmb{x}) \succeq \pmb{d} \succeq \pref(\pmb{d}) \succeq \pmb{c}$ as desired.

\item The case $\pmb{d} \not\preceq f'^{\ast}_{\langle \lambda \rangle}(\pmb{x})$:
We have
\begin{equation}
f'^{\ast}_{\langle \lambda \rangle}(\pmb{x})
\overset{(\mathrm{A})}{=} f''^{\ast}_{\langle \lambda \rangle}(\pmb{x})
\overset{(\mathrm{B})}{\succeq} \pmb{c},
\end{equation}
where (A) follows from the second case of (\ref{eq:i5igfy04wlhe}) and $\pmb{d} \not\preceq f'^{\ast}_{\langle \lambda \rangle}(\pmb{x})$,
and (B) follows from (\ref{eq:ckeqzkg47tsl}).
\end{itemize}

(Proof of (ii)):
For $i \in [F] \setminus \{\langle \lambda \rangle\}$, we have $f''_i(s) = f'_i(s)$ directly from the second case of (\ref{eq:fddot}).
We consider the case where $i = \langle \pmb{z} \rangle$ for some $\pmb{z} \in \mathcal{S}^L$.
Then we have $f'^{\ast}_{\langle \lambda \rangle}(\pmb{z}) \not\prec \pmb{d}$ because 
$|f'^{\ast}_{\langle \lambda \rangle}(\pmb{z})| \geq |\pmb{d}|+1$ by Lemma \ref{lem:fddot} (iii).
Therefore, by the second case of (\ref{eq:fddot}), we obtain $f''_i(s) = f'_i(s)$.

(Proof of (iii)):
We prove only that $\PREF^k_{F'', i}(\pmb{b}) \subseteq \PREF^k_{F', i}(\pmb{b})$ for any $i \in \mathcal{J}$ and $\pmb{b} \in \mathcal{C}^{\ast}$ because we can prove $\bar{\PREF}^k_{F'', i}(\pmb{b}) \subseteq \bar{\PREF}^k_{F', i}(\pmb{b})$ in the similar way.
To prove $\PREF^k_{F'', i}(\pmb{b}) \subseteq \PREF^k_{F', i}(\pmb{b})$, it suffices to prove that
for any $(i, \pmb{x}, \pmb{b}, \pmb{c}) \in \mathcal{J} \times \mathcal{S}^{+} \times \mathcal{C}^{\ast} \times \mathcal{C}^{\leq k}$, we have
\begin{align}
(f''^{\ast}_i(\pmb{x}) \succeq \pmb{b}\pmb{c}, f''_i(x_1) \succeq \pmb{b})
&\implies {}^{\exists}\pmb{x}' \in \mathcal{S}^{+}; (f'^{\ast}_i(\pmb{x}') \succeq \pmb{b}\pmb{c}, f'_i(x'_1) \succeq \pmb{b}) \label{eq:ju3g6hv4k9a4}
\end{align}
because this shows that for any $i \in \mathcal{J}$, $\pmb{b} \in \mathcal{C}^{\ast}$, and $\pmb{c} \in \mathcal{C}^k$, we have
\begin{align*}
\pmb{c} \in \PREF^k_{F'', i}(\pmb{b})
&\overset{(\mathrm{A})}{\iff} {}^{\exists}\pmb{x} \in \mathcal{S}^{+}; (f''^{\ast}_i(\pmb{x}) \succeq \pmb{b}\pmb{c}, f''_i(x_1) \succeq \pmb{b})\\
&\overset{(\mathrm{B})}{\implies} {}^{\exists}\pmb{x}' \in \mathcal{S}^{+}; (f'^{\ast}_i(\pmb{x}') \succeq \pmb{b}\pmb{c}, f'_i(x'_1) \succeq \pmb{b})\\
&\overset{(\mathrm{C})}{\iff} \pmb{c} \in \PREF^k_{F', i}(\pmb{b})
\end{align*}
as desired, where
(A) follows from (\ref{eq:pref1}),
(B) follows from (\ref{eq:ju3g6hv4k9a4}),
and (C) follows from (\ref{eq:pref1}).

Choose $(i, \pmb{x}, \pmb{b}, \pmb{c}) \in [F] \times \mathcal{S}^{+} \times \mathcal{C}^{\ast} \times \mathcal{C}^{\leq k}$ arbitrarily and assume
\begin{equation}
\label{eq:mtiglxn32kjk}
f''^{\ast}_i(\pmb{x}) \succeq \pmb{b}\pmb{c}
\end{equation}
and
\begin{equation}
\label{eq:pmg4lzz8qyo0}
f''_i(x_1) \succeq \pmb{b}.
\end{equation}
Then we have
\begin{equation}
\label{eq:n2a6qwpk4yif}
f'_i(x_1) \overset{(\mathrm{A})}{=} f''_i(x_1) \overset{(\mathrm{B})}{\succeq} \pmb{b},
\end{equation}
where (A) follows from (ii) of this lemma,
and (B) follows from (\ref{eq:pmg4lzz8qyo0}).

We prove (\ref{eq:ju3g6hv4k9a4}) by induction for $|\pmb{x}|$.
For the base case $|\pmb{x}| = 1$, we have
\begin{equation}
\label{eq:05c82yn9pr7f}
f'^{\ast}_i(\pmb{x})
= f'_i(x_1)
\overset{(\mathrm{A})} = f''_i(x_1)
= f''^{\ast}_i(\pmb{x})
\overset{(\mathrm{B})}{\succeq}
\pmb{b}\pmb{c}
\end{equation}
as desired, where
(A) follows from (ii) of this lemma,
and (B) follows from (\ref{eq:mtiglxn32kjk}).
By (\ref{eq:05c82yn9pr7f}) and (\ref{eq:n2a6qwpk4yif}), the claim (\ref{eq:ju3g6hv4k9a4}) holds for the base case $|\pmb{x}| = 1$.

We consider the induction step for $|\pmb{x}| \geq 2$.
We have
\begin{align}
f'_i(x_1)f''^{\ast}_{\trans''_i(x_1)}(\suff(\pmb{x}))
&\overset{(\mathrm{A})}{=} f''_i(x_1)f''^{\ast}_{\trans''_i(x_1)}(\suff(\pmb{x}))
\overset{(\mathrm{B})}{=} f''^{\ast}_i(\pmb{x}) \overset{(\mathrm{C})}{\succeq} \pmb{b}\pmb{c}, \label{eq:2k893w5sucds}
\end{align}
where
(A) follows from (ii) of this lemma,
(B) follows from (\ref{eq:fstar}),
and (C) follows from (\ref{eq:mtiglxn32kjk}).

Therefore, $f'_i(x_1) \succeq \pmb{b}\pmb{c}$ or $f'_i(x_1) \prec \pmb{b}\pmb{c}$ holds.
In the case $f'_i(x_1) \succeq \pmb{b}\pmb{c}$, the sequence $\pmb{x}' \coloneqq x_1$ satisfies $f'^{\ast}_i(\pmb{x}') \succeq \pmb{b}\pmb{c}$ and $f'_i(x'_1) = f'_i(x_1) \succeq \pmb{b}$ by (\ref{eq:n2a6qwpk4yif}) as desired.
Thus, now we assume $f'_i(x_1) \prec \pmb{b}\pmb{c}$.
Then we have 
\begin{align}
|f'_i(x_1)^{-1}\pmb{b}\pmb{c}|
&= - |f'_i(x_1)| + |\pmb{b}| + |\pmb{c}|
\overset{(\mathrm{A})}{=} - |f''_i(x_1)| + |\pmb{b}| + |\pmb{c}|
\overset{(\mathrm{B})}{\leq} |\pmb{c}|
\leq k,\label{eq:zebgryw0i016}
\end{align}
where
(A) follows from (ii) of this lemma,
and (B) follows from (\ref{eq:pmg4lzz8qyo0}).
By (\ref{eq:2k893w5sucds}), we have
\begin{equation}
\label{eq:0hrjyvle93b6}
f''^{\ast}_{\trans''_i(x_1)}(\suff(\pmb{x})) \succeq f'_i(x_1)^{-1}\pmb{b}\pmb{c}.
\end{equation}

We can see that there exists $\pmb{x}' \in \mathcal{S}^{+}$ such that
\begin{equation}
\label{eq:p2apgc7qchat}
f'^{\ast}_{\trans''_i(x_1)}(\pmb{x}') \succeq f'_i(x_1)^{-1}\pmb{b}\pmb{c}
\end{equation}
as follows.
\begin{itemize}
\item The case $\trans''_i(x_1) = \langle \lambda \rangle$: By (\ref{eq:zebgryw0i016}), we can apply (i) of this lemma to obtain that $\pmb{x}' \coloneqq \suff(\pmb{x})$ satisfies (\ref{eq:p2apgc7qchat}) from (\ref{eq:0hrjyvle93b6}).
\item The case $\trans''_i(x_1) \in \mathcal{J}$: By (\ref{eq:zebgryw0i016}) and (\ref{eq:0hrjyvle93b6}), we can apply the induction hypothesis to $(\trans''_i(x_1), \suff(\pmb{x}), \lambda, f'_i(x_1)^{-1}\pmb{b}\pmb{c})$.
\end{itemize}
Therefore, we have
\begin{align}
f'^{\ast}_i(x_1\pmb{x}')
&\overset{(\mathrm{A})}{=} f'_i(x_1)f'^{\ast}_{\trans'_i(x_1)}(\pmb{x}')
\overset{(\mathrm{B})}{=} f'_i(x_1)f'^{\ast}_{\trans''_i(x_1)}(\pmb{x}')
\overset{(\mathrm{C})}{\succeq} f'_i(x_1)f'_i(x_1)^{-1}\pmb{b}\pmb{c}
= \pmb{b}\pmb{c}, \label{eq:bb04kybrpm5q}
\end{align}
where
(A) follows from (\ref{eq:fstar}),
(B) follows from (\ref{eq:g0yqx2cwxfd1}),
and (C) follows from (\ref{eq:p2apgc7qchat}).
The induction is completed by (\ref{eq:n2a6qwpk4yif}) and (\ref{eq:bb04kybrpm5q}).
\end{proof}

\subsection{Proof of Lemma \ref{lem:psi-1}}
\label{subsec:proof-psi-1}

\begin{proof}[Proof of Lemma \ref{lem:psi-1}]
(Proof of (i)):
Assume that
\begin{equation}
\label{eq:ebpzimbqneev}
\pmb{b} \preceq \pmb{b}'.
\end{equation}
In the case $f'^{\ast}_{\langle \lambda \rangle}(\pmb{z}) \not\preceq \pref(\pmb{d})$, 
we have
\begin{equation}
\psi_{\pmb{z}}(\pmb{b})
\overset{(\mathrm{A})}{=} \pmb{b}
\overset{(\mathrm{B})}{\preceq} \pmb{b}'
\overset{(\mathrm{C})}{=} \psi_{\pmb{z}}(\pmb{b}'),
\end{equation}
where
(A) follows from the second case of (\ref{eq:psi-1}) and $f'^{\ast}_{\langle \lambda \rangle}(\pmb{z}) \not\preceq \pref(\pmb{d})$,
(B) follows from (\ref{eq:ebpzimbqneev}),
and (C) follows from the second case of (\ref{eq:psi-1}) and $f'^{\ast}_{\langle \lambda \rangle}(\pmb{z}) \not\preceq \pref(\pmb{d})$.

We consider the case $f'^{\ast}_{\langle \lambda \rangle}(\pmb{z}) \preceq \pref(\pmb{d})$ 
dividing into four cases by whether $\pref(\pmb{d}) \prec f'^{\ast}_{\langle \lambda \rangle}(\pmb{z})\pmb{b}$ and whether $\pref(\pmb{d}) \prec f'^{\ast}_{\langle \lambda \rangle}(\pmb{z})\pmb{b}'$.

\begin{itemize}
\item The case $\pref(\pmb{d}) \prec f'^{\ast}_{\langle \lambda \rangle}(\pmb{z})\pmb{b},\pref(\pmb{d}) \prec f'^{\ast}_{\langle \lambda \rangle}(\pmb{z})\pmb{b}'$:
We have
\begin{align*}
\psi_{\pmb{z}}(\pmb{b}) 
&\overset{(\mathrm{A})}{=}  f'^{\ast}_{\langle \lambda \rangle}(\pmb{z})^{-1}\pmb{d}\pref(\pmb{d})^{-1}(f'^{\ast}_{\langle \lambda \rangle}(\pmb{z})\pmb{b})\\
&\overset{(\mathrm{B})}{\preceq}  f'^{\ast}_{\langle \lambda \rangle}(\pmb{z})^{-1}\pmb{d}\pref(\pmb{d})^{-1}(f'^{\ast}_{\langle \lambda \rangle}(\pmb{z})\pmb{b}')
\overset{(\mathrm{C})}{=} \psi_{\pmb{z}}(\pmb{b}')
\end{align*}
as desired, where
(A) follows from the first case of (\ref{eq:psi-1}) and $\pref(\pmb{d}) \prec f'^{\ast}_{\langle \lambda \rangle}(\pmb{z})\pmb{b}$,
(B) follows from (\ref{eq:ebpzimbqneev}),
and (C) follows from the first case of (\ref{eq:psi-1}) and $\pref(\pmb{d}) \prec f'^{\ast}_{\langle \lambda \rangle}(\pmb{z})\pmb{b}'$.

\item The case $\pref(\pmb{d}) \prec f'^{\ast}_{\langle \lambda \rangle}(\pmb{z})\pmb{b},\pref(\pmb{d}) \not\prec f'^{\ast}_{\langle \lambda \rangle}(\pmb{z})\pmb{b}'$:
This case is impossible because (\ref{eq:ebpzimbqneev}) leads to $\pref(\pmb{d}) \prec f'^{\ast}_{\langle \lambda \rangle}(\pmb{z})\pmb{b}\preceq f'^{\ast}_{\langle \lambda \rangle}(\pmb{z})\pmb{b}'$, which conflicts with $\pref(\pmb{d}) \not\prec f'^{\ast}_{\langle \lambda \rangle}(\pmb{z})\pmb{b}'$.

\item The case $\pref(\pmb{d}) \not\prec f'^{\ast}_{\langle \lambda \rangle}(\pmb{z})\pmb{b}, \pref(\pmb{d}) \prec f'^{\ast}_{\langle \lambda \rangle}(\pmb{z})\pmb{b}'$:
By (\ref{eq:ebpzimbqneev}), we have
\begin{equation}
\label{eq:r5udbmcr4atk}
f'^{\ast}_{\langle \lambda \rangle}(\pmb{z})\pmb{b} \preceq f'^{\ast}_{\langle \lambda \rangle}(\pmb{z})\pmb{b}'.
\end{equation}
By (\ref{eq:r5udbmcr4atk}) and $\pref(\pmb{d}) \prec f'^{\ast}_{\langle \lambda \rangle}(\pmb{z})\pmb{b}'$,
exactly one of $\pref(\pmb{d}) \prec f'^{\ast}_{\langle \lambda \rangle}(\pmb{z})\pmb{b}$ and $\pref(\pmb{d}) \succeq f'^{\ast}_{\langle \lambda \rangle}(\pmb{z})\pmb{b}$ holds.
Since the former does not hold by $\pref(\pmb{d}) \not\prec f'^{\ast}_{\langle \lambda \rangle}(\pmb{z})\pmb{b}$,
the latter holds:
\begin{equation}
\label{eq:c3oxzym98dhz}
f'^{\ast}_{\langle \lambda \rangle}(\pmb{z})\pmb{b} \preceq \pref(\pmb{d}).
\end{equation}
Thus, we have
\begin{align*}
\psi_{\pmb{z}}(\pmb{b})
&\overset{(\mathrm{A})}{=} \pmb{b}
= f'^{\ast}_{\langle \lambda \rangle}(\pmb{z})^{-1}f'^{\ast}_{\langle \lambda \rangle}(\pmb{z})\pmb{b}
\overset{(\mathrm{B})}{\preceq} f'^{\ast}_{\langle \lambda \rangle}(\pmb{z})^{-1}\pref(\pmb{d})\\
&\preceq f'^{\ast}_{\langle \lambda \rangle}(\pmb{z})^{-1}\pmb{d}\pref(\pmb{d})^{-1}(f'^{\ast}_{\langle \lambda \rangle}(\pmb{z})\pmb{b}')
\overset{(\mathrm{C})}{=} \psi_{\pmb{z}}(\pmb{b}'),
\end{align*}
where
(A) follows from the second case of (\ref{eq:psi-1}) and $\pref(\pmb{d}) \not\prec f'^{\ast}_{\langle \lambda \rangle}(\pmb{z})\pmb{b}$,
(B) follows from (\ref{eq:c3oxzym98dhz}),
and (C) follows from the first case of (\ref{eq:psi-1}) and $\pref(\pmb{d}) \prec f'^{\ast}_{\langle \lambda \rangle}(\pmb{z})\pmb{b}'$.

\item The case $\pref(\pmb{d}) \not\prec f'^{\ast}_{\langle \lambda \rangle}(\pmb{z})\pmb{b},\pref(\pmb{d}) \not\prec f'^{\ast}_{\langle \lambda \rangle}(\pmb{z})\pmb{b}'$:
We have
\begin{equation}
\psi_{\pmb{z}}(\pmb{b}) \overset{(\mathrm{A})}{=} \pmb{b}
\overset{(\mathrm{B})}{\preceq} \pmb{b}' \overset{(\mathrm{C})}{=} \psi_{\pmb{z}}(\pmb{b}')
\end{equation}
as desired, where
(A) follows from the second case of (\ref{eq:psi-1}) and $\pref(\pmb{d}) \not\prec f'^{\ast}_{\langle \lambda \rangle}(\pmb{z})\pmb{b}$,
(B) follows from (\ref{eq:ebpzimbqneev}),
and (C) follows from the second case of (\ref{eq:psi-1}) and $\pref(\pmb{d}) \not\prec f'^{\ast}_{\langle \lambda \rangle}(\pmb{z})\pmb{b}'$.
\end{itemize}

(Proof of (ii)):
We consider the following three cases separately: (I) the case $f'^{\ast}_{\langle \lambda \rangle}(\pmb{z}) \preceq \pref(\pmb{d}) \prec f'^{\ast}_{\langle \lambda \rangle}(\pmb{zx})$, (II) the case $f'^{\ast}_{\langle \lambda \rangle}(\pmb{z}\pmb{x}) \preceq \pref(\pmb{d}) \prec f'^{\ast}_{\langle \lambda \rangle}(\pmb{z}\pmb{x}) \pmb{c}$, and (III) the other case:

\begin{itemize}
\item[(I)] The case $f'^{\ast}_{\langle \lambda \rangle}(\pmb{z}) \preceq \pref(\pmb{d}) \prec f'^{\ast}_{\langle \lambda \rangle}(\pmb{zx})$:
We have
\begin{equation}
\label{eq:98ohcwg7568z}
f'^{\ast}_{\langle \lambda \rangle}(\pmb{z}) \prec \pmb{d} \preceq f'^{\ast}_{\langle \lambda \rangle}(\pmb{zx})
\end{equation}
since
\begin{equation}
\pref(\pmb{d}) \bar{d}_l
\overset{(\mathrm{A})}{\not\in} \PREF^{\ast}_{F, \langle \lambda \rangle}
\overset{(\mathrm{B})}{=} \PREF^{\ast}_{F', \langle \lambda \rangle},
\end{equation}
where 
(A) follows from (\ref{eq:2cz6211gg810}),
and (B) follows from Lemma \ref{lem:duplicate} (ii) since $\varphi$ defined in (\ref{eq:gtet1tbkkabj}) is a homomorphism from $F'$ to $F$.
Therefore, by the second case of (\ref{eq:psi-1}), we obtain
\begin{equation}
\label{eq:vhnklq74di4m}
f''^{\ast}_{\langle \pmb{z} \rangle}(\pmb{x}) = f'^{\ast}_{\langle \lambda \rangle}(\pmb{z})^{-1}\pref(\pmb{d})\pmb{d}^{-1}(f'^{\ast}_{\langle \lambda \rangle}(\pmb{z}\pmb{x})).
\end{equation}

We consider the following two cases separately: (I-A) the case $f'^{\ast}_{\langle \lambda \rangle}(\pmb{z}) \prec f'^{\ast}_{\langle \lambda \rangle}(\pmb{zx}) = \pmb{d}, \pmb{c} = \lambda$ and (I-B) the other case.
\begin{itemize}
\item[(I-A)] The case $f'^{\ast}_{\langle \lambda \rangle}(\pmb{z}) \prec f'^{\ast}_{\langle \lambda \rangle}(\pmb{zx}) = \pmb{d}, \pmb{c} = \lambda$:
We have
\begin{align}
f'^{\ast}_{\langle \lambda \rangle}(\pmb{z})f''^{\ast}_{\langle \pmb{z} \rangle}(\pmb{x})\pmb{c}
&\overset{(\mathrm{A})}{=}  f'^{\ast}_{\langle \lambda \rangle}(\pmb{z})f'^{\ast}_{\langle \lambda \rangle}(\pmb{z})^{-1}\pref(\pmb{d})\pmb{d}^{-1}(f'^{\ast}_{\langle \lambda \rangle}(\pmb{z}\pmb{x}))\pmb{c}\nonumber\\
&\overset{(\mathrm{B})}{=}  f'^{\ast}_{\langle \lambda \rangle}(\pmb{z})f'^{\ast}_{\langle \lambda \rangle}(\pmb{z})^{-1}\pref(\pmb{d})\pmb{d}^{-1}\pmb{d}\pmb{c}\nonumber\\
&\overset{(\mathrm{C})}{=} \pref(\pmb{d})\nonumber\\
&\not\succ \pref(\pmb{d}), \label{eq:tslo5leloxhy}
\end{align}
where
(A) follows from (\ref{eq:vhnklq74di4m}),
(B) follows from $f'^{\ast}_{\langle \lambda \rangle}(\pmb{zx}) = \pmb{d}$,
and (C) follows from $\pmb{c} = \lambda$.

Hence, we have
\begin{align*}
\psi_{\pmb{z}}(f''^{\ast}_{\langle \pmb{z} \rangle} (\pmb{x}) \pmb{c})
&\overset{(\mathrm{A})}{=} f''^{\ast}_{\langle \pmb{z} \rangle} (\pmb{x}) \pmb{c}\\
&\overset{(\mathrm{B})}{=} f'^{\ast}_{\langle \lambda \rangle}(\pmb{z})^{-1}\pref(\pmb{d})\pmb{d}^{-1}f'^{\ast}_{\langle \lambda \rangle}(\pmb{z}\pmb{x})\pmb{c}\\
&\overset{(\mathrm{C})}{=} f'^{\ast}_{\langle \lambda \rangle}(\pmb{z})^{-1}\pref(f'^{\ast}_{\langle \lambda \rangle}(\pmb{z}\pmb{x}))\pmb{d}^{-1}\pmb{d}\\
&\overset{(\mathrm{D})}{=} f'^{\ast}_{\langle \lambda \rangle}(\pmb{z})^{-1}f'^{\ast}_{\langle \lambda \rangle}(\pmb{z})\pref(f'^{\ast}_{\langle \pmb{z} \rangle}(\pmb{x})))\pmb{d}^{-1}\pmb{d}\\
&= \pref(f'^{\ast}_{\langle \pmb{z}\rangle}(\pmb{x}))
\end{align*}
as desired, where
(A) follows from the second case of (\ref{eq:psi-1}) and (\ref{eq:tslo5leloxhy}),
(B) follows from (\ref{eq:vhnklq74di4m}),
(C) follows from $f'^{\ast}_{\langle \lambda \rangle}(\pmb{zx}) = \pmb{d}$ and $\pmb{c} = \lambda$,
and (D) follows from Lemma \ref{lem:f_T} (i), Lemma \ref{lem:fdot} (i), and $f'^{\ast}_{\langle \lambda \rangle}(\pmb{z}) \prec f'^{\ast}_{\langle \lambda \rangle}(\pmb{zx})$.

\item[(I-B)] The other case:
Then by (\ref{eq:98ohcwg7568z}), we have
\begin{equation}
\label{eq:nbm6p8q0eb2m}
\pmb{d} \prec f'^{\ast}_{\langle \lambda \rangle}(\pmb{zx})\pmb{c},
\end{equation}
since it does not hold that $f'^{\ast}_{\langle \lambda \rangle}(\pmb{z}) \prec f'^{\ast}_{\langle \lambda \rangle}(\pmb{zx}) = \pmb{d}, \pmb{c} = \lambda$ by the assumption of the case (I-B).

We have
\begin{align}
f'^{\ast}_{\langle \lambda \rangle}(\pmb{z})f''^{\ast}_{\langle \pmb{z} \rangle}(\pmb{x})\pmb{c}
&\overset{(\mathrm{A})}{=} f'^{\ast}_{\langle \lambda \rangle}(\pmb{z})f'^{\ast}_{\langle \lambda \rangle}(\pmb{z})^{-1}\pref(\pmb{d})\pmb{d}^{-1}(f'^{\ast}_{\langle \lambda \rangle}(\pmb{z}\pmb{x}))\pmb{c}\nonumber\\
&\overset{(\mathrm{B})}{\succ} f'^{\ast}_{\langle \lambda \rangle}(\pmb{z})f'^{\ast}_{\langle \lambda \rangle}(\pmb{z})^{-1}\pref(\pmb{d})\pmb{d}^{-1}\pmb{d}\nonumber\\
&= \pref(\pmb{d}) \label{eq:8n1fmv36lng1}\\
&\overset{(\mathrm{C})}{\succeq} f'^{\ast}_{\langle \lambda \rangle}(\pmb{z}) \label{eq:cl3zr8c7d4hf}
\end{align}
as desired, where
(A) follows from (\ref{eq:vhnklq74di4m}),
(B) follows from (\ref{eq:nbm6p8q0eb2m}),
and (C) follows from the assumption of the case (I).

Hence, we have
\begin{align*}
\psi_{\pmb{z}}(f''_{\langle \pmb{z} \rangle} (\pmb{x}) \pmb{c})
&\overset{(\mathrm{A})}{=}  f'^{\ast}_{\langle \lambda \rangle}(\pmb{z})^{-1}\pmb{d}\pref(\pmb{d})^{-1}(f'^{\ast}_{\langle \lambda \rangle}(\pmb{z})f''^{\ast}_{\langle \pmb{z} \rangle}(\pmb{x})\pmb{c})\\
&\overset{(\mathrm{B})}{=}  f'^{\ast}_{\langle \lambda \rangle}(\pmb{z})^{-1}\pmb{d}\pref(\pmb{d})^{-1}(f'^{\ast}_{\langle \lambda \rangle}(\pmb{z})f'^{\ast}_{\langle \lambda \rangle}(\pmb{z})^{-1}\\
&\quad \pref(\pmb{d})\pmb{d}^{-1}(f'^{\ast}_{\langle \lambda \rangle}(\pmb{z}\pmb{x})\pmb{c}))\\
&=  f'^{\ast}_{\langle \pmb{z} \rangle}(\pmb{x})\pmb{c}\\
&\overset{(\mathrm{C})}{=}  f'^{\ast}_{\langle \pmb{z} \rangle}(\pmb{x})\psi_{\pmb{zx}}(\pmb{c}),
\end{align*}
where
(A) follows from the first case of (\ref{eq:psi-1}), (\ref{eq:8n1fmv36lng1}), and (\ref{eq:cl3zr8c7d4hf}),
(B) follows from (\ref{eq:vhnklq74di4m}),
and (C) follows from the second case of (\ref{eq:psi-1}) and the assumption of the case (I).
\end{itemize}

\item[(II)] The case $f'^{\ast}_{\langle \lambda \rangle}(\pmb{z}\pmb{x}) \preceq \pref(\pmb{d}) \prec f'^{\ast}_{\langle \lambda \rangle}(\pmb{z}\pmb{x}) \pmb{c}$:
Then since $\pmb{d} \not \preceq f'^{\ast}_{\langle \lambda \rangle}(\pmb{z}\pmb{x})$, we have
\begin{equation}
\label{eq:zmtfdfdxntjl}
f''^{\ast}_{\langle \pmb{z} \rangle}(\pmb{x}) = f'^{\ast}_{\langle \pmb{z} \rangle}(\pmb{x})
\end{equation}
applying the second case of (\ref{eq:i5igfy04wlhe}).
Therefore, we have
\begin{align}
f'^{\ast}_{\langle \lambda \rangle}(\pmb{z})
&\preceq f'^{\ast}_{\langle \lambda \rangle}(\pmb{z}\pmb{x}) \nonumber\\
&\overset{(\mathrm{A})}{\preceq} \pref(\pmb{d}) \label{eq:t9bf3rrxb006}\\
&\overset{(\mathrm{A})}{\prec} f'^{\ast}_{\langle \lambda \rangle}(\pmb{z}\pmb{x}) \pmb{c}\nonumber\\
&\overset{(\mathrm{B})}{=} f'^{\ast}_{\langle \lambda \rangle}(\pmb{z})f'^{\ast}_{\langle \pmb{z} \rangle}(\pmb{x}) \pmb{c} \nonumber\\
&\overset{(\mathrm{C})}{=} f'^{\ast}_{\langle \lambda \rangle}(\pmb{z})f''^{\ast}_{\langle \pmb{z} \rangle}(\pmb{x}) \pmb{c}, \label{eq:1pzl6zbf9wuz}
\end{align}
where
(A)s follow from the assumption of the case (II), 
(B) follows from Lemma \ref{lem:f_T} (i) and Lemma \ref{lem:fdot} (i),
and (C) follows from (\ref{eq:zmtfdfdxntjl}).

Hence, we have 
\begin{align*}
\psi_{\pmb{z}}(f''^{\ast}_{\langle \pmb{z} \rangle}(\pmb{x})\pmb{c})
&\overset{(\mathrm{A})}{=} f'^{\ast}_{\langle \lambda \rangle}(\pmb{z})^{-1}\pmb{d}\pref(\pmb{d})^{-1}(f'^{\ast}_{\langle \lambda \rangle}(\pmb{z})f''^{\ast}_{\langle \pmb{z} \rangle}(\pmb{x})\pmb{c}) \\
&\overset{(\mathrm{B})}{=}  f'^{\ast}_{\langle \lambda \rangle}(\pmb{z})^{-1}\pmb{d}\pref(\pmb{d})^{-1}(f'^{\ast}_{\langle \lambda \rangle}(\pmb{z})f'^{\ast}_{\langle \pmb{z} \rangle}(\pmb{x})\pmb{c}) \\
&\overset{(\mathrm{C})}{=} f'^{\ast}_{\langle \lambda \rangle}(\pmb{z})^{-1}\pmb{d}\pref(\pmb{d})^{-1}(f'^{\ast}_{\langle \lambda \rangle}(\pmb{zx})\pmb{c}) \\
&= f'^{\ast}_{\langle \pmb{z} \rangle} (\pmb{x}) f'^{\ast}_{\langle \pmb{z} \rangle} (\pmb{x})^{-1} f'^{\ast}_{\langle \lambda \rangle}(\pmb{z})^{-1}\pmb{d}\pref(\pmb{d})^{-1}(f'^{\ast}_{\langle \lambda \rangle}(\pmb{zx})\pmb{c}) \\
&\overset{(\mathrm{D})}{=} f'^{\ast}_{\langle \pmb{z} \rangle} (\pmb{x}) f'^{\ast}_{\langle \lambda \rangle}(\pmb{zx})^{-1}\pmb{d}\pref(\pmb{d})^{-1}(f'^{\ast}_{\langle \lambda \rangle}(\pmb{zx})\pmb{c}) \\
&\overset{(\mathrm{E})}{=} f'^{\ast}_{\langle \pmb{z} \rangle} (\pmb{x}) \psi_{\pmb{zx}}(\pmb{c})
\end{align*}
as desired, where
(A) follows from the first case of (\ref{eq:psi-1}), (\ref{eq:t9bf3rrxb006}), and (\ref{eq:1pzl6zbf9wuz}),
(B) follows from (\ref{eq:zmtfdfdxntjl}),
(C) follows from Lemma \ref{lem:f_T} (i) and Lemma \ref{lem:fdot} (i),
(D) follows from Lemma \ref{lem:f_T} (i) and Lemma \ref{lem:fdot} (i),
and (E) follows from the first case of (\ref{eq:psi-1}) and the assumption of the case (II).

\item[(III)] The other case:
The following implication holds:
\begin{align}
f'^{\ast}_{\langle \lambda \rangle}(\pmb{z}) \prec \pmb{d} \preceq f'^{\ast}_{\langle \lambda \rangle}(\pmb{zx})
&\implies f'^{\ast}_{\langle \lambda \rangle}(\pmb{z}) \preceq \pref(\pmb{d}) \prec f'^{\ast}_{\langle \lambda \rangle}(\pmb{zx})\label{eq:y5f0ctncj25q}
\end{align}
Now, it does not hold that $f'^{\ast}_{\langle \lambda \rangle}(\pmb{z}) \preceq \pref(\pmb{d}) \prec f'^{\ast}_{\langle \lambda \rangle}(\pmb{zx})$ by the assumption of the case (III).
Hence, by the contraposition of (\ref{eq:y5f0ctncj25q}), we see that
$f'^{\ast}_{\langle \lambda \rangle}(\pmb{z}) \prec \pmb{d} \preceq f'^{\ast}_{\langle \lambda \rangle}(\pmb{zx})$ does not hold.
Therefore, we obtain
\begin{equation}
\label{eq:523mr9j3co7t}
f''^{\ast}_{\langle \pmb{z} \rangle}(\pmb{x}) = f'^{\ast}_{\langle \pmb{z} \rangle}(\pmb{x})
\end{equation}
applying the second case of (\ref{eq:i5igfy04wlhe}).

By the assumption of the case (III), neither $f'^{\ast}_{\langle \lambda \rangle}(\pmb{z}) \preceq \pref(\pmb{d}) \prec f'^{\ast}_{\langle \lambda \rangle}(\pmb{zx})$ nor $f'^{\ast}_{\langle \lambda \rangle}(\pmb{z}\pmb{x}) \preceq \pref(\pmb{d}) \prec f'^{\ast}_{\langle \lambda \rangle}(\pmb{z}\pmb{x}) \pmb{c}$ hold.
Hence, the following condition does not hold:
\begin{equation}
f'^{\ast}_{\langle \lambda \rangle}(\pmb{z})
\preceq \pref(\pmb{d})
\prec f'^{\ast}_{\langle \lambda \rangle}(\pmb{z}\pmb{x}) \pmb{c}
\overset{(\mathrm{A})}{=} f'^{\ast}_{\langle \lambda \rangle}(\pmb{z})f''^{\ast}_{\langle \pmb{z} \rangle}(\pmb{x})\pmb{c},
\end{equation}
where (A) follows from (\ref{eq:523mr9j3co7t}).
Therefore, by the second case of (\ref{eq:psi-1}), we have
\begin{equation}
\label{eq:n3hwci6ksq12}
\psi_{\pmb{z}}(f''^{\ast}_{\langle \pmb{z} \rangle}(\pmb{x})\pmb{c}) = f''^{\ast}_{\langle \pmb{z} \rangle} (\pmb{x})\pmb{c}.
\end{equation}

Thus, we have
\begin{align*}
f'^{\ast}_{\langle \pmb{z} \rangle} (\pmb{x}) \psi_{\pmb{zx}}(\pmb{c})
&\overset{(\mathrm{A})}{=} f''^{\ast}_{\langle \pmb{z} \rangle} (\pmb{x}) \psi_{\pmb{zx}}(\pmb{c})
\overset{(\mathrm{B})}{=} f''^{\ast}_{\langle \pmb{z} \rangle} (\pmb{x})\pmb{c}
\overset{(\mathrm{C})}{=} \psi_{\pmb{z}}(f''^{\ast}_{\langle \pmb{z} \rangle}(\pmb{x})\pmb{c})
\end{align*}
\end{itemize}
as desired, where
(A) follows from (\ref{eq:523mr9j3co7t}),
(B) follows from the second case of (\ref{eq:psi-1}) since $f'^{\ast}_{\langle \lambda \rangle}(\pmb{z}\pmb{x}) \preceq \pref(\pmb{d}) \prec f'^{\ast}_{\langle \lambda \rangle}(\pmb{z}\pmb{x}) \pmb{c}$ does not hold by the assumption of the case (III),
and (C) follows from (\ref{eq:n3hwci6ksq12}).

(Proof of (iii)):
We have $f'^{\ast}_{\langle \lambda \rangle}(\pmb{z}) \not\preceq \pref(\pmb{d})$ because $|f'^{\ast}_{\langle \lambda \rangle}(\pmb{z})| > |\pmb{d}|$ by Lemma \ref{lem:fddot} (iii).
Hence, by the second case of (\ref{eq:psi-1}), we obtain $\psi_{\langle \pmb{z} \rangle}(\pmb{b}) = \pmb{b}$ as desired.
\end{proof}

\section{List of Notations}
\label{sec:notation}

\begin{tabular}{lp{0.85\textwidth}}
  $\mathcal{A} \times \mathcal{B}$ & the Cartesian product of sets $\mathcal{A}$ and $\mathcal{B}$, that is, $\{(a, b) : a \in \mathcal{A}, b \in \mathcal{B}\}$, defined at the beginning of Section \ref{sec:preliminary}. \\
  $|\mathcal{A}|$ & the cardinality of a set $\mathcal{A}$, defined at the beginning of Section \ref{sec:preliminary}. \\
  $\mathcal{A}^k$ & the set of all sequences of length $k$ over a set $\mathcal{A}$, defined at the beginning of Section \ref{sec:preliminary}. \\
 $\mathcal{A}^{\geq k}$ & the set of all sequences of length greater than or equal to $k$ over a set $\mathcal{A}$, defined at the beginning of Section \ref{sec:preliminary}. \\
 $\mathcal{A}^{\leq k}$ & the set of all sequences of length less than or equal to  $k$ over a set $\mathcal{A}$, defined at the beginning of Section \ref{sec:preliminary}. \\
  $\mathcal{A}^{\ast}$ & the set of all sequences of finite length over a set $\mathcal{A}$, defined at the beginning of Section \ref{sec:preliminary}.\\
  $\mathcal{A}^{+}$ & the set of all sequences of finite positive length over a set $\mathcal{A}$, defined at the beginning of Section \ref{sec:preliminary}.\\
  $\mathcal{C}$ & the coding alphabet $\mathcal{C} = \{0, 1\}$, at the beginning of Section \ref{sec:preliminary}.\\
  $\bar{c}$ & the negation of $c \in \mathcal{C}$, that is, $\bar{0} = 1, \bar{1} = 0$ defined 
  at the beginning of the proof of Theorem \ref{thm:complete}.\\
  $f^{\ast}_i$ & defined in Definition \ref{def:f_T}. \\
  $F$ & simplified notation of a code-tuple $F(f_0, f_1, \allowbreak \ldots, f_{m-1}, \trans_0, \trans_1, \ldots, \trans_{m-1})$, 
  also written as $F(f, \trans)$, defined below Definition \ref{def:treepair}.\\
  $\bar{F}$ & an irreducible part of $F$, defined in Definition \ref{def:irr-part}.\\
  $|F|$ & the number of code tables of $F$, defined below Definition \ref{def:treepair}. \\
  $[F]$ & simplified notation of $[|F|] = \{0, 1, \allowbreak  2, \ldots, |F|-1\}$, defined below Definition \ref{def:treepair}.\\
 \end{tabular}
 \newpage \noindent
 \begin{tabular}{lp{0.85\textwidth}}
  $\mathscr{F}^{(m)}$ & the set of all $m$-code-tuples, defined after Definition \ref{def:treepair}.\\
  $\mathscr{F}$ & the set of all code-tuples, defined after Definition \ref{def:treepair}.\\
  $\mathscr{F}_{\ext}$ & the set of all extendable code-tuples, defined in Definition \ref{def:F_ext}. \\
  $\mathscr{F}_{k\hopt}$ & defined in Definition \ref{def:optimalset}. \\
  $\mathscr{F}_{\reg}$ & the set of all regular code-tuples, defined in Definition \ref{def:regular}. \\
  $h(F)$ & defined after Lemma \ref{lem:potential}.\\
  $L(F)$ & the average codeword length of a code-tuple $F$, defined in Definition \ref{def:evaluation}. \\
  $L_i(F)$ & the average codeword length of the $i$-th code table of $F$, defined in Definition \ref{def:evaluation}. \\
  $[m]$ & $\{0, 1, 2, \ldots, m-1\}$, defined at the beginning of Section \ref{sec:preliminary}.\\
  $\mathcal{P}^k_{F, i}$ & defined in Definition \ref{def:pref}.\\
  $\bar{\mathcal{P}}^k_{F, i}$ & defined in Definition \ref{def:pref}.\\
    $\mathcal{P}^{\ast}_{F, i}$ & defined in Definition \ref{def:prefstar}.\\
  $\bar{\mathcal{P}}^{\ast}_{F, i}$ & defined in Definition \ref{def:prefstar}.\\
  $\prefset^k_F$ & defined in Definition \ref{def:prefset}.\\
  $\pref(\pmb{x})$ & the sequence obtained by deleting the last letter of $\pmb{x}$, defined at the beginning of Section \ref{sec:preliminary}. \\
  $Q(F)$ & the transition probability matrix, defined in Definition \ref{def:transprobability}.\\
  $Q_{i, j}(F)$ & the transition probability, defined in Definition \ref{def:transprobability}.\\
  $\mathbb{R}$ & the set of all real numbers.\\
  $\mathbb{R}^m$ & the set of all $m$-dimensional real row vectors for an integer $m \geq 1$.\\
  $\mathcal{S}$ & the source alphabet, defined at the beginning of Section \ref{sec:preliminary}.\\
  $\suff(\pmb{x})$ & the sequence obtained by deleting the first letter of $\pmb{x}$, defined at the beginning of Section \ref{sec:preliminary}. \\
  $\pmb{x} \preceq \pmb{y}$ & \pmb{x} is a prefix of \pmb{y}, defined at the beginning of Section \ref{sec:preliminary}.\\
    $\pmb{x} \prec \pmb{y}$ & $\pmb{x} \preceq \pmb{y}$ and $\pmb{x} \neq \pmb{y}$, defined at the beginning of Section \ref{sec:preliminary}. \\
  $|\pmb{x}|$ & the length of a sequence $\pmb{x}$, defined at the beginning of Section \ref{sec:preliminary}.\\
  $\pmb{x}^{-1}\pmb{y}$ & the sequence $\pmb{z}$ such that $\pmb{x}\pmb{z} = \pmb{y}$ defined at the beginning of Section \ref{sec:preliminary}.\\
  $\lambda$ & the empty sequence, defined at the beginning of Section \ref{sec:preliminary}.\\
  $\mu(s)$ & the probability of occurrence of symbol $s$, defined at the beginning of Subsection \ref{subsec:evaluation}. \\
  $\pmb{\pi}(F)$ & defined in Definition \ref{def:regular}.\\
  $\sigma$ & the alphabet size, defined at the beginning of Section \ref{sec:preliminary}.\\
  $\trans^{\ast}_i$ & defined in Definition \ref{def:f_T}. \\
\end{tabular}

\section*{Acknowledgment}
This work was supported in part by JSPS KAKENHI Grant Number JP18H01436, and in part by KIOXIA.


\begin{thebibliography}{99}
\bibitem{Huffman1952}
	 D.~A.~Huffman,
	 ``A Method for the Construction of Minimum-Redundancy Codes,''
	 \emph{Proc. I.R.E.},
	 vol.~40, no.~9, pp.~1098--1102, 1952.
\bibitem{McMillan1956}
	B.~McMillan,
	``Two inequalities implied by unique decipherability,''
	\emph{IRE Trans.\ on Inf.\ Theory},
	 vol.~2, no.~4, pp.~115--116, Dec.~1956.

\bibitem{Yamamoto2015}
	 H.~Yamamoto, M.~Tsuchihashi, and J.~Honda, 
	 ``Almost Instantaneous Fixed-to-Variable Length Codes,''
	 \emph{IEEE Trans.\ on Inf.\ Theory},
	 vol.~61, no.~12, pp.~6432--6443, Dec.~2015.

\bibitem{Hashimoto2022}
	K.~Hashimoto and K.~Iwata,
	 "Optimality of Huffman Code in the Class of 1-bit Delay Decodable Codes,"
	 in \emph{IEEE Journal on Selected Areas in Information Theory},
	 doi: 10.1109/JSAIT.2022.3230745,
	 \emph{arXiv:2209.08874}. [Online]. Available: http://arxiv.org/abs/2209.08874

\bibitem{Hashimoto2021}
	K.~Hashimoto and K.~Iwata, 
	``On the Optimality of Binary AIFV Codes with Two Code Trees,''
	in \emph{Proc.\ IEEE Int.\ Symp.\ on Inf.\ Theory} (ISIT),
	Melbourne, Victoria, Australia (Virtual Conference), Jul.~2021, pp. 3173--3178.

\bibitem{Hashimoto2023}
	K.~Hashimoto and K.~Iwata, 
	``The Optimality of AIFV Codes in the Class of 2-bit Delay Decodable Codes,''
	in preparation.

\bibitem{Hu2017}
	W.~Hu, H.~Yamamoto, and J.~Honda, 
	``Worst-case Redundancy of Optimal Binary AIFV Codes and Their Extended Codes,''
	 \emph{IEEE Trans.\ on Inf.\ Theory},
	vol.~63, no.~8, pp.~5074--5086, Aug.~2017.

\bibitem{Fujita2020}
	R.~Fujita K.~Iwata, and H.~Yamamoto, 
	``Worst-case Redundancy of Optimal Binary AIFV-$m$ Codes for $m = 3, 5$,''
	in \emph{Proc.\ IEEE Int.\ Symp.\ on Inf.\ Theory} (ISIT),
	Los Angeles, CA, USA (Virtual Conference), Jun.~2020, pp.~2355--2359.	 

 \bibitem{IY:ISITA16}
	 K.~Iwata, and H.~Yamamoto,
	 ``A Dynamic Programming Algorithm to Construct Optimal Code Trees of AIFV Codes,''
	 in \emph{Proc.\ Int.\ Symp.\ Inf.\ Theory and its Appl.} (ISITA),
	 Monterey, CA, USA, Oct./Nov.\ 2016, pp.~641--645.

\bibitem{IY:ITW17}
	K.~Iwata, and H.~Yamamoto,
	``An Iterative Algorithm to Construct Optimal Binary AIFV-$m$ Codes,''
	in \emph{Proc.\ IEEE\ Inf.\ Theory Workshop} (ITW),
	Kaohsiung, Taiwan, Oct.~2017, pp.~519--523.

\bibitem{Sumigawa2017}
	K.~Sumigawa and H.~Yamamoto,
	``Coding of binary AIFV code trees,''
	in \emph{Proc.\ IEEE Int.\ Symp.\ on Inf.\ Theory} (ISIT),	 
	Aachen, Germany, Jun.~2017, pp.~1152--1156. 

 \bibitem{ISIT2018}
	 T.~Hiraoka and H.~Yamamoto, 
	 ``Alphabetic AIFV codes Constructed from Hu-Tucker Codes,''
	in \emph{Proc.\ IEEE Int.\ Symp.\ on Inf.\ Theory} (ISIT),
	Vail, CO, USA, Jun.~2018, pp.~2182--2186.

\bibitem{Fujita2018}
	R.~Fujita, K.~Iwata, and H.~Yamamoto, 
	``An Optimality Proof of the Iterative Algorithm for AIFV-$m$ Codes,''
	in \emph{Proc.\ IEEE Int.\ Symp.\ on Inf.\ Theory} (ISIT),
	Vail, CO, USA, Jun.~2018, pp.~2187--2191.

 \bibitem{ISITA2018}
	 T.~Hiraoka and H.~Yamamoto, 
	 ``Dynamic AIFV Coding,''
	 in \emph{Proc.\ Int.\ Symp.\ Inf.\ Theory and its Appl.} (ISITA),
	 Singapore, Oct.~2018, pp.~577--581.

\bibitem{Fujita2019}
	R.~Fujita, K.~Iwata, and H.~Yamamoto,
	``An Iterative Algorithm to Optimize the Average Performance of Markov Chains with Finite States,''
	in \emph{Proc.\ IEEE Int.\ Symp.\ on Inf.\ Theory} (ISIT),
	Paris, France,  Jul.\ 2019, pp. 1902-1906.

 \bibitem{Hashimoto2019}
	K.~Hashimoto, K.~Iwata, and H.~Yamamoto,
	``Enumeration and coding of compact code trees for binary AIFV codes,''
	in \emph{Proc.\ IEEE Int.\ Symp.\ on Inf.\ Theory} (ISIT),	 
	Paris, France, Jul.~2019, pp.~1527-1531.

 \bibitem{Golin2019}
	 M.~J.~Golin and E.~Y.~Harb,
	 ``Polynomial Time Algorithms for Constructing Optimal AIFV Codes,''
	 in \emph{Proc.\ Data Compression Conference} (DCC), Snowbird, UT, USA, Mar.~2019, pp.231--240.

 \bibitem{Golin2020}
	 M.~J.~Golin and E.~Y.~Harb,
	``Polynomial Time Algorithms for Constructing Optimal Binary AIFV-$2$ Codes,''
	2020, \emph{arXiv:2001.11170}. [Online]. Available: http://arxiv.org/abs/2001.11170

 \bibitem{ISIT2020}
	 H.~Yamamoto, K.~Imaeda, K.~Hashimoto, and K.~Iwata, 
	 ``A Universal Data Compression Scheme based on the AIFV Coding Techniques,''
	 in \emph{Proc.\ IEEE Int.\ Symp.\ on Inf.\ Theory} (ISIT),	 
	 Los Angeles, LA, USA (Virtual Conference),  Jun.~2020, pp.~2378--2382. 

 \bibitem{ITW2020}
	 K.~Iwata and H.~Yamamoto, 
	 ``An Algorithm for Construction the Optimal Code Trees for Binary Alphabetic AIFV-$m$ Codes,''
	 in \emph{Proc.\ IEEE\ Inf.\ Theory Workshop} (ITW), 
	 Riva de Garda, Italy (Virtual Conference), Apr.~2021 (postponed from 2020 to 2021), pp.~261--265.

\bibitem{Golin2021}
	 M.~J.~Golin and E.~Y.~Harb,
	``Speeding up the AIFV-2 dynamic programs by two orders of magnitude using Range Minimum Queries,''
	\emph{Theoretical Computer Science},
	vol.\ 865, no.\ 14, pp.~99--118, Apr.~2021.

 \bibitem{Golin2022}
	 M.~J.~Golin and A.~J.~L.~Parupat,
	 ``Speeding Up AIFV-$m$ Dynamic Programs by $m-1$ Orders of Magnitude,''
	 in \emph{Proc.\ IEEE Int.\ Symp.\ on Inf.\ Theory} (ISIT),	 
	 Espoo, Finland, Jun.--Jul.~2022, pp.~282--287. 

 \bibitem{Sugiura2018}
	R.~Sugiura, Y.~Kamamoto, N.~Harada, and T.~Moriya,
	``Optimal Golomb-Rice code extension for lossless coding of low-entropy exponentially distributed sources,'' 
	\emph{IEEE Trans Inf. Theory}, vol. 64, no. 4, pp.~3153--3161, Jan.~2018.

 \bibitem{Sugiura2022}
	 R.~Sugiura, Y.~Kamamoto, and T.~Moriya,
	 ``General Form of Almost Instantaneous Fixed-to-Variable-Length Codes and Optimal Code Tree Construction,''
	 2022, \emph{arXiv:2203.08437}. [Online]. Available: http://arxiv.org/abs/2203.08437
	 
\bibitem{Puterman}
	 M.~L.~Puterman,
	 \emph{Markov Decision Processes: Discrete Stochastic Dynamic Programming}.
	 Hoboken, NJ, USA: John~Wiley and Sons, 2005.

\end{thebibliography}

\end{document}